\documentclass[pra,a4paper,preprintnumbers,reprint,nofootinbib,superscriptaddress]{revtex4-2}
\pdfoutput=1 

\usepackage[utf8]{inputenc}  
\usepackage{amsmath, amssymb, amsfonts, amsthm}  
\usepackage{mathtools}  

\usepackage{dsfont}  
\usepackage{physics}  
\usepackage{graphicx}  
\usepackage{tikz}  
\usetikzlibrary{quantikz2}  
\usepackage{pifont} 

\usepackage[T1]{fontenc}
\usepackage[british]{babel}
\usepackage[sort&compress]{natbib}

\usepackage[babel]{microtype}  
\usepackage{xcolor}  
\usepackage[dvipsnames,x11names]{xcolor}  
\usepackage[colorlinks=true, citecolor=OliveGreen, linkcolor=Purple, urlcolor=Magenta]{hyperref}  
\usepackage{comment}  

\usepackage{thmtools}
\usepackage{thm-restate}
\newtheorem{remark}{Remark}

\newtheorem{corollary}{Corollary}
\newtheorem{lemma}{Lemma}

\DeclarePairedDelimiter{\of}{\lparen}{\rparen}

\makeatletter 
\newcommand{\doublewidetilde}[1]{{%
  \mathpalette\double@widetilde{#1}%
}}
\newcommand{\double@widetilde}[2]{%
  \sbox\z@{$\m@th#1\map{#2}$}%
  \ht\z@=.9\ht\z@
  \map{\box\z@}%
}

\newcommand{\dket}[1]{\vert#1\rangle\hspace{-.8mm}\rangle}
\newcommand{\dbra}[1]{\langle\hspace{-.8mm}\langle #1\vert}
\newcommand{\dketbra}[2]{\vert #1 \rangle\!\rangle \langle\!\langle #2 \vert}

\newcommand{\map}[1]{\mathcal{#1}}

\renewcommand{\mathbb}{\mathds}  
\newcommand{\id}{\mathds{1}}  
 
\definecolor{LavenderBlue}{RGB}{95, 0, 255} 

\usepackage{soul,xcolor}

\newcommand{\C}{\mathbb{C}}
\renewcommand{\S}{\mathrm{S}}

\newcommand{\I}{\textup{I}}
\renewcommand{\O}{\textup{O}}

\RequirePackage{xcolor}

\begin{document}

\title{Optimal physical approximations of pure-state cloning and transposition are complementary channels}

\author{Vanessa Brzi\'c}
\affiliation{Sorbonne Universit\'{e}, CNRS, LIP6, F-75005 Paris, France}
\author{Dmitry Grinko}
\affiliation{QuSoft, Amsterdam, The Netherlands}
\affiliation{Institute for Logic, Language and Computation, and Korteweg-de Vries Institute for Mathematics, University of Amsterdam, The Netherlands}
\author{Micha{\l} Studzi\'{n}ski}
\affiliation{International Centre for Theory of Quantum Technologies,
University of Gda\'{n}sk, Prof. Marii Janion~4, 80-309~Gda\'{n}sk, Poland}
\author{Marco Túlio Quintino}
\affiliation{Sorbonne Universit\'{e}, CNRS, LIP6, F-75005 Paris, France}

\begin{abstract}
State cloning and state transposition are fundamental transformations that, despite being
desirable, cannot be perfectly realised as a consequence of two distinct requirements of quantum theory:
cloning is forbidden by linearity, while transposition is ruled out by complete positivity.
In this work, we show that, in spite of these disparate conceptual restrictions, the best physically allowed
realisation of both transformations arises from a single physical process described by an isometry,
which simultaneously implements their best possible approximations. This complementary relation allows us to determine the exact quantum and private classical capacities of the universal symmetric $N\to N+K$ cloning channel in the fully general multicopy qudit setting, thereby closing the remaining gap in the quantum capacity problem for universal cloning. We provide an explicit circuit that simultaneously implements optimal pure-state $N\to K$ transposition and optimal $N\to N+K$ symmetric cloning, both attaining the same fidelity. Additionally, we prove that the optimal $N\to K$ pure state transposition channel is unique and realised by an estimation strategy. We also obtain the optimal mixed-state $N\to 1$ transposition, and discuss its differences from its pure-state analogue. 
\end{abstract}

\maketitle

The fundamental principles of quantum theory, such as linearity and complete positivity, give rise to a number of no-go theorems that challenge the expectations inherited from classical information processing. While copying
information is a basic primitive of classical computation, the perfect cloning of an unknown quantum state is forbidden by the no-cloning theorem, as a direct consequence of the linear structure of quantum mechanics \cite{1982Wotters_NCl, 1982Dieks}. These limitations, however, constitute defining features of quantum information processing. The no-cloning theorem forms a cornerstone of security proofs in quantum cryptography, ensuring that information encoded in quantum systems cannot be copied without introducing detectable disturbances \cite{Bennett_2014}. 
Likewise, the transposition map plays an important role in entanglement theory, through the positive partial transpose (PPT) criterion for entanglement detection \cite{Horodecki1996PPT,Horodecki2000direct}, is known to allow better performance in estimation tasks~\cite{Gisin1999Antiparallel, Demkowicz2005Estm, Miyazaki_2020,King2024Learning_adv} and it allows $2$-copy incompatible measurements to become compatible~\cite{Patra2026QuantumIncompatibility}. However, transposition of a quantum state is a paradigmatic example of a positive but not completely positive map, and therefore cannot be physically realised as a quantum channel \cite{Horodecki1996PPT,1996Peres_Sep}.

The problem of circumventing such no-go theorems by identifying the best physically admissible approximations to unphysical transformations has been extensively studied in quantum information theory \cite{1996Buzek_qubit12_cl,1997Gisin_qubitNM_cl,Werner1998_cloning,1999Keyl_cl,Scarani2005Clon,Dong2019positive,Fan2014Cloning,Horodecki2003multicopy,Bae2017,Buzek1999UniNOT,Rungta2001NOT,1998Bruss,Bruss1998state_dependent}, with a general semidefinite-programming framework~\cite{AudenaertDeMoor2002, Skrzypczyk2023SDP}. In the case of pure-state inputs, optimal universal $N\to N+K$ quantum cloning machines have been characterised for arbitrary finite dimension~\cite{Werner1998_cloning,1999Keyl_cl}. By contrast, previous results on optimal pure-state transposition were restricted to the qubit $N\to1$ scenario, through its unitary equivalence with the universal NOT~\cite{Buzek1999UniNOT}, and to the $1\to1$ setting in arbitrary finite dimension~\cite{Buscemi2006PhD}. When considering mixed states, the quantum cloning problem is often referred to as quantum broadcasting~\cite{Barnum1996Broadcast}, and most existing results are restricted to qubits~\cite{DAriano2005Superbroadcasting,Buscemi2006Superbroadcasting,Chiribella_2007Superbroad}. In particular, Ref.~\cite{Chiribella_2007Superbroad} analysed the optimal qubit $N\to K$ broadcast NOT, which is unitarily equivalent to mixed-state $N\to K$ transposition for qubits, while qudit $N\to1$ mixed-state transposition was studied in Ref.~\cite{Dong2019positive}.

Moreover, in certain cases, a relation between cloning and transposition has been identified.
The ancillary output of an isometry implementing the optimal universal $1\to2$ cloner was noticed to realise the optimal $1\to1$ transposition channel~\cite{Buzek1998cl,2003Buscemi_T}, equivalent to the so-called \textit{anticlone}. This is an instance of channel complementarity~\cite{Holevo2005ComplChannels}, in which two channels arise from different outputs of a single global isometry. However, since these results are not fully general in dimension and number of copies, it remains unclear whether the observed relationship between transposition and cloning is merely an artefact of these specific cases or reflects a fundamental connection between the two transformations. Identifying such complementarity relations reveals how different quantum-information tasks can emerge from the same underlying process, which is of direct importance, for example, in quantum Shannon theory, where the information available to the environment is closely tied to the communication capabilities of the original channel~\cite{SchumacherWestmoreland1998Privacy,
Devetak2005PrivateCapacity,DevetakShor2005Capacity, Holevo2005ComplChannels}

In this work, we prove that the connection between cloning and transposition persists in full generality for universal symmetric cloning of pure-state inputs. This shows that optimal physical realisations of forbidden quantum transformations can be related despite the disparate nature of the no-go theorems that prohibit them. More precisely, we consider quantum channels that approximate the transformation of $N$ copies of an arbitrary $d$-dimensional quantum state into $K$ copies of its transposed state, thereby defining an $N\to K$ transposition task. For pure input states, we identify the optimal $N\to K$ transposition channel and show that it is complementary to the optimal universal symmetric $N\to N+K$ quantum cloning channel. Our result enables us to determine the exact quantum capacity of the general $N\to N+K$ qudit cloning channel, extending previous capacity results to the fully general multicopy qudit setting.

Inspired by~\cite{Mancinska2025classification}, we further present an explicit circuit that simultaneously realises optimal transposition and symmetric cloning and discuss its gate efficiency.
Finally, we address optimal approximate $N\to 1$ transposition for mixed input states and determine the optimal performance in terms of the maximal visibility corresponding to the structural physical approximation. 

{\it Optimal $N\to K$  pure state transposition ---} %
The transposition map
$\map{T}:\mathcal{L}(\mathbb{C}^d)\to \mathcal{L}(\mathbb{C}^d)$
is defined by $\map{T}(\rho)=\rho^\mathsf{T}$, where the transpose is
taken in the computational basis, that is,
$\map{T}(\ketbra{i}{j})=\ketbra{j}{i}$.
Our goal is to find a quantum channel that approximates the transposition map in the multi-copy regime. More precisely, we are interested in a CPTP map $\map{C}:\mathcal{L}\left((\mathbb{C}^d)^{\otimes N}\right)
\to\mathcal{L}\left((\mathbb{C}^d)^{\otimes K}\right)$ that approximately converts $N$ copies of an arbitrary pure state into
$K$ copies of its transpose, that is,
$\map{C}(\ketbra{\psi}^{\otimes N})
\approx(\ketbra{\psi}^{\mathsf{T}})^{\otimes K}$. We quantify the quality of the approximation using state fidelity, where the fidelity between an arbitrary density matrix $\rho$ and a pure state $\ketbra{\psi}$ is given by $F(\rho,\ketbra{\psi}) := \Tr(\rho \ketbra{\psi})$. In particular, we consider the worst-case fidelity over all pure states. Finally, since quantum states are self-adjoint operators, they satisfy $\rho^\mathsf{T}=\overline{\rho}$, where $\overline{\rho}$ is the complex conjugation of $\rho$ in the computational basis.
%

Our first main result characterises the optimal worst-case fidelity achievable by completely positive and trace-preserving maps from $N$ input copies to $K$ output copies. 
This generalises the optimal universal NOT construction~\cite{Buzek1999UniNOT}, recovering the known fidelity $(N+1)/(N+d)$ for $d=2$ and $K=1$.

\begin{restatable}{theorem}{TransNK}\label{thm:TransNK}
The maximum worst-case fidelity for approximating the transpose map from $N$ input copies to $K$ output copies on pure states is given by 
\begin{equation}\label{eq:fid_trans}
    \max_{\map{C}
    \in\mathrm{CPTP}}
    \min_{\ket{\psi}}
    F\Big(\mathcal{C}\left(\ketbra{\psi}^{\otimes N}\right),\,{\ketbra{\psi}^\mathsf{T}}^{\otimes K}\Big)=\frac{d_S^{N}}{d_S^{N+K}}
\end{equation}
where $d_S^N:=\dim(\mathrm{sym}^N(\mathbb{C}^d))=\binom{N+d-1}{d-1}$. 

Since all input states $\ket{\psi}^{\otimes N}$ are supported on the symmetric subspace, the optimisation may, without loss of generality, be restricted to channels $\mathcal C:\mathcal L(\operatorname{sym}^N(\mathbb C^d))\to\mathcal L((\mathbb C^d)^{\otimes K})$. The optimal channel in Eq.~\eqref{eq:fid_trans} is then unique and is given by an estimation strategy,
\begin{align} \label{eq:T_CP}
    \mathcal{T}_\textup{CP}(\rho)
    :=& d_S^{N}
    \int \mathrm{d}\psi
    \Tr\left[\rho \ketbra{\psi}^{\otimes N}\right]\,{\ketbra{\psi}^{\mathsf{T}}}^{\otimes K} \\
    =&  \frac{d_S^{N}}{d_S^{N+K}} \Tr_{1\ldots N}\left[ \Pi^{(N+K)}_\textup{sym} \, \of*{\rho^\mathsf{T} \otimes \id_d^{\otimes K} }\right],
\end{align}
where $\Pi^{(N+K)}_{\mathrm{sym}}$ is the projector onto the $N+K$ symmetric subspace, $\id_d$ is the identity on $\mathbb{C}^d$, and $\Tr_{1\ldots N}$ is the partial trace on the first $N$ $d$-dimensional subsystems.
\end{restatable}

The proof of Thm.~\ref{thm:TransNK} is given in the Appendix~\ref{app:NKtrans}, where we first show that the optimal worst-case fidelity in Eq.~\eqref{eq:fid_trans} coincides with the optimal Haar-average fidelity, and then formulate the latter optimisation as a semidefinite program with strong duality~\cite{AudenaertDeMoor2002,Skrzypczyk2023SDP}. Also, in the Appendix~\ref{app:NKtrans} we prove that, without loss in performance, the optimisation can be restricted to covariant quantum channels, namely CPTP maps $\map{C}:\mathcal{L}(\mathbb{C}^{d})^{\otimes N}\to\mathcal{L}(\mathbb{C}^{d})^{\otimes K}$ satisfying 
\begin{align} \label{eq:cov1}
\mathcal{C}\left(U^{\otimes N} \rho (U^{\otimes N})^{\dagger}\right)
&= \overline{U}^{\otimes K}\, \mathcal{C}(\rho)\, (U^\mathsf{T})^{\otimes K},
\end{align}
for every unitary $U\in \mathrm{SU}(d)$ and every $\rho\in\mathcal{L}(\mathbb{C}^d)^{\otimes N}$, as well as
$
\mathcal{C}\left(V_\pi \rho V_\pi^\dagger\right)
= V_\sigma\, \mathcal{C}(\rho)\, V_\sigma^\dagger,
$
for all permutations $\pi\in S_N$ and $\sigma\in S_K$. Here, $V_\pi$ and $V_\sigma$ denote the unitary representations of the symmetric groups acting on $(\mathbb{C}^d)^{\otimes N}$ and $(\mathbb{C}^d)^{\otimes K}$, respectively. 
For quantum channels respecting the covariance relations of Eq.~\eqref{eq:cov1}, the fidelity
$ F \big(\mathcal{C}(\ketbra{\psi}^{\otimes N}),\,(\ketbra{\psi}^\mathsf{T})^{\otimes K}\big) $
does not depend on the state $\ket{\psi}$. 
Hence, as in~\cite{holevoBook,chiribella2005estimation,Werner1998_cloning,Mancinska2025classification,grinko2024linear,Grinko2023,Quintino2022deterministic,brzic2025HOQC_KN,ishizaka2009quantum,IshizakaHiroshima,studzinski2017port,mozrzymas2018optimal,studzinski2021degradation,studzinski2020efficient,kopszak2020multiport,mozrzymas2021optimal}, the optimal Haar average fidelity coincides with the optimal worst-case fidelity.

The map $\map{T}_\textup{CP}$ presented in Eq.~\eqref{eq:T_CP} is an entanglement breaking channel~\cite{Horodecki_2003}, described by measure-and-prepare strategy. In particular, one first performs the ``Hayashi measurement''~\cite{Hayashi1998estimation,2013Harrow_church}, with POVM elements $M_{\psi}:=d_S^{N}\ketbra{\psi}^{\otimes N}\mathrm{d}\psi$, and then, upon obtaining the outcome corresponding to $\ket{\psi}$, prepares the state $\of[\large]{\ketbra{\psi}^\mathsf{T}}^{\otimes K}$.
\normalsize

Although the map $\map{T}_\textup{CP}$ is defined on the restricted domain
$\mathcal{L} \left(\mathrm{sym}^N(\mathbb{C}^d)\right)$, it admits CPTP extensions to the full input space. One such extension
$\map{T}_\textup{CP}^\textup{ext}:\mathcal{L}((\mathbb{C}^d)^{\otimes N})\to\mathcal{L}((\mathbb{C}^d)^{\otimes K})$
is given by
\begin{align} \label{eq:T_CP^ext}
\mathcal{T}^\textup{ext}_\textup{CP}(\rho)
:=\;&
\frac{d_S^{N}}{d_S^{N+K}}
\Tr_{1\ldots N}\!\left[
\Pi^{(N+K)}_\textup{sym}
\left(
\rho^\mathsf{T}\otimes \id_d^{\otimes K}
\right)
\right]
\\
&\quad
+
\Tr\!\left[
\left(
\id_d^{\otimes N}
-
\Pi^{(N)}_\textup{sym}
\right)
\rho
\right]\sigma,
\end{align}
where $\sigma\in\mathcal{L}((\mathbb{C}^d)^{\otimes K})$ is an arbitrary quantum state. This extension is not unique, since the action of the channel outside the symmetric subspace can be chosen in different ways. However, this freedom does not affect inputs of the form $\rho=\ketbra{\psi}^{\otimes N}$, which are entirely supported on $\mathrm{sym}^N(\mathbb{C}^d)$. Therefore, $\mathcal{T}_\textup{CP}\!\left(\ketbra{\psi}^{\otimes N}\right)
=
\mathcal{T}^\textup{ext}_\textup{CP}\!\left(\ketbra{\psi}^{\otimes N}\right).$

An analogous statement holds for the optimal cloning channel: although its operationally relevant action is fixed on the symmetric input subspace, it may likewise be extended in a CPTP manner to the full $N$-qudit input space without changing its action on states of the form $\ketbra{\psi}^{\otimes N}$.

To connect the multicopy transposition protocol with single-copy approximations of transposition and structural physical approximations (SPA)~\cite{Horodecki2003multicopy,Bae2017}, we consider the reduced action of the optimal transposition channel $\mathcal{T}_\textup{CP}$ defined in Eq.~\eqref{eq:T_CP}. In particular, the reduced state of any single output system respects
\begin{equation} \label{eq:shrinking}
\Tr_{\overline{O_i}}\big(
\mathcal{T}_\textup{CP}(\ketbra{\psi}^{\otimes N})
\big)
=
\eta\,\ketbra{\psi}^\mathsf{T}
+
(1-\eta)\frac{\mathbb{1}}{d}.
\end{equation}
Here, $\Tr_{\overline{O_i}}$ denotes the partial trace on all output systems except the $i$-th one. Owing to permutation covariance, the reduced state in Eq.~\eqref{eq:shrinking} is independent of the choice of output system $i$. In Appendix~\ref{app:NKtrans}, we show that $\eta = \frac{N}{N+d}$, which is independent of the number of output copies $K$.
Furthermore, under the covariance assumptions, the channel $\mathcal{T}_\textup{CP}$ maximises the visibility $\eta$ appearing in Eq.~\eqref{eq:shrinking}. 
The coefficient $\eta = \frac{N}{N+d}$ is then the $N \to  K$ transposition analogue of the so-called ``black cow factor'' for optimal $N \to  N+K $ cloning~\cite{Werner1998_cloning,1998Bruss} discussed in the following. 

{\it Complementarity between transposition and cloning ---} %
By the Stinespring dilation
Thm.~\cite{Stinespring1955Dilation}, every quantum channel
$\mathcal{C}:\mathcal{L}(\mathcal{H}_{\I})\to
\mathcal{L}(\mathcal{H}_{\O})$ has an isometric extension
$V:\mathcal{H}_{\I}\to\mathcal{H}_{\O}\otimes\mathcal{H}_{E}$ such
that $\mathcal{C}(\rho)=\Tr_{E}(V\rho V^\dagger)$, where
$\mathcal{H}_{E}$ represents an environment of the channel. A channel $\mathcal{E}:\mathcal{L}(\mathcal{H}_{\I})\to
\mathcal{L}(\mathcal{H}_{E})$ is complementary to
$\mathcal{C}$ if the two channels arise from the same Stinespring
isometry, so that, for every
$\rho\in\mathcal{L}(\mathcal{H}_{\I})$,
\begin{align}
\mathcal{C}(\rho)
&=
\Tr_{E}\left(V\rho V^\dagger\right),
&
\mathcal{E}(\rho)
&=
\Tr_{\O}\left(V\rho V^\dagger\right).
\end{align}
Thus, whereas $\mathcal{C}$ describes the state received by the intended receiver, its complementary channel $\mathcal{E}$ describes the corresponding output state of the environment~\footnote{For a fixed channel, a complementary channel constructed from a minimal Stinespring dilation is unique up to a unitary isomorphism of the environment. Complementary channels arising from arbitrary dilations are related by an isometry~\cite{Holevo2005ComplChannels}.}

Complementary channels play an important role in various domains of quantum information theory. In particular, they play a crucial role in the study of quantum and private communication capacities~\cite{DevetakShor2005Capacity,Devetak2005PrivateCapacity}, in the theory of degradable and antidegradable channels~\cite{Cubitt2008Deg, Caruso2006Antideg}, in quantum error correction~\cite{Kretschmann2008PrivateCorrectable, Kretschmann2008Continuity}, and more broadly in quantum cryptography, where the environmental output is naturally interpreted as an eavesdropper's system~\cite{SchumacherWestmoreland1998Privacy}.

Therefore, identifying the complementary channel of optimal universal cloning has foundational value for understanding the operational properties of the cloning map. Let us recall a result from Werner~\cite{Werner1998_cloning}, where the optimal universal symmetric cloning channel that transforms $N$ copies
of a pure qudit state into $N+K$ copies, is defined as the solution of the optimisation problem
\begin{equation}\label{eq:fid_clone}
    \max_{\map{C}\in\mathrm{CPTP}}
  \min_{\ket{\psi}}
     F \left(\mathcal{C} \of[\Large]{\ketbra{\psi}^{\otimes N}},\ketbra{\psi}^{\otimes (N+K)}\right).
\end{equation}
As proved in~\cite{Werner1998_cloning,1999Keyl_cl}, when the domain of $\mathcal{C}$ is restricted to the symmetric subspace, the solution of Eq.~\eqref{eq:fid_clone} is unique and given by the Werner cloning channel, $
\map{W}_\textup{CP}: \mathcal{L} \left(\mathrm{sym}^N(\mathbb{C}^d)\right)\to \mathcal{L}\left(\mathbb{C}^d \right)^{\otimes (N+K)},$
\begin{align}\label{eq:cl_map}
\mathcal{W}_{\textup{CP}}(\rho)
&:=
\frac{d_S^{N}}{d_S^{N+K}}\,
\Pi^{(N+K)}_{\mathrm{sym}}\,
\big(\rho\otimes \mathbb{1}_d^{\otimes K}\big)\,
\Pi^{(N+K)}_{\mathrm{sym}},
\end{align}
where, $\Pi^{(N+K)}_{\mathrm{sym}}$ is the projector onto the $(N+K)$ symmetric subspace. Werner cloning channel attains the  fidelity of $F \of[\big]{\mathcal{W}_{\textup{CP}} \of[\Large]{\ketbra{\psi}^{\otimes N}},\ketbra{\psi}^{\otimes (N+K)}}={d_S^{N}}/{d_S^{N+K}}$ for any quantum state $\ket{\psi}$. This leads us to our second main result, given by the following theorem.

\begin{restatable}{theorem}{ComplTC}\label{thm:ComplTC}
For any dimension $d$ and any $N,K\in\mathbb{N}$, the optimal universal symmetric pure-state cloning channel $\map{W}_{\mathrm{CP}}$ from $N$ to $N+K$ copies presented in Eq.~\eqref{eq:cl_map} and the optimal pure-state transposition channel $\map{T}_{\mathrm{CP}}$ from $N$ to $K$, defined in Eq.~\eqref{eq:T_CP}, are complementary channels.
\end{restatable}

The proof of Thm.~\ref{thm:ComplTC} is given in the Appendix~\ref{appendix:C}. Let us note that, due to covariance, changing the basis of transposition amounts to applying a unitary on the transposition output, so the cloning-transposition complementarity is unchanged up to an environment unitary~\footnote{In particular, the action of the transposition map on an arbitrary basis $\{U\ket{i}\}_i$ is defined as $\map{T}_U(\rho):=U\,\map{T}(U^\dagger \rho U)\,U^\dagger$. Using the covariance property of the transposition map, one finds $ \map{T}_U(\rho) = U U^\mathsf{T}\,\map{T}(\rho)\,\overline{U}U^\dagger$.
Defining $A:=UU^\mathsf{T}\in\mathrm{U}(d)$, this can be written as $\map{T}_U(\rho)=A\,\map{T}(\rho)\,A^\dagger$, hence the transposition in another basis is just the transposition on the computational basis followed by a unitary operation.}.
The theorem extends the previously known cloning-transposition relation~\cite{Buzek1998cl,2003Buscemi_T,Buzek1999UniNOT} beyond the special cases in which it was originally identified. Importantly, the complementarity established in Thm.~\ref{thm:ComplTC} is independent of how the channels are implemented and does not rely on any particular circuit realisation.

Moreover, the common isometric dilation admits a simple interpretation that reflects the covariance symmetries of both tasks. Under an arbitrary unitary rotation $U$ of the input state, the cloning output transforms according to $U^{\otimes(N+K)}$, whereas the transposition output transforms according to the conjugate representation $\overline{U}^{\otimes K}$, as in Eq.~\eqref{eq:cov1}. To make this explicit, consider the isometry
$V:\mathrm{sym}^N(\mathbb{C}^d)\to\mathrm{sym}^{N+K}(\mathbb{C}^d)\otimes(\mathbb{C}^d)^{\otimes K}$,
which can be written as
\begin{equation}\label{eq:circuit_isometry}
V\ket{\phi}
=
\sqrt{\frac{d_S^N}{d_S^{N+K}}}
\left(
\Pi_{\mathrm{sym}}^{(N+K)}
\otimes
\mathbb{1}_d^{\otimes K}
\right)
\left(
\ket{\phi}
\otimes
\dket{\mathbb{1}}^{\otimes K}
\right),
\end{equation}
where $\dket{\mathbb{1}}=\sum_{i=1}^d\ket{i}\ket{i}$ is an unnormalised maximally entangled vector. From Eq.~\eqref{eq:circuit_isometry}, we see that the isometry can be realised, upon successful projection, by introducing $K$ maximally entangled pairs, combining one half of each pair with the $N$ input systems, and projecting the resulting $N+K$ systems onto the symmetric subspace. These $N+K$ systems form the cloning output, while the remaining halves constitute the environment. Moreover, using the symmetry properties $(U\otimes\overline{U})\dket{\mathbb{1}}=\dket{\mathbb{1}}$ together with the invariance of $\Pi_{\mathrm{sym}}^{(N+K)}$ under the $U^{\otimes(N+K)}$, we obtain the intertwining relation
\begin{equation}\label{eq:circuit_intertwining}
V U^{\otimes N}
=
\left(
U^{\otimes(N+K)}
\otimes
\overline{U}^{\otimes K}
\right)V.
\end{equation}
Thus, the cloning output transforms according to the $(N+K)$-fold representation, while the environment transforms according to the conjugate representation required for transposition. By Thms.~\ref{thm:TransNK} and~\ref{thm:ComplTC}, tracing out the cloning output yields the optimal transposition channel, while tracing out the environment yields the optimal cloning channel.

Thm.~\ref{thm:ComplTC} provides direct implications for quantum Shannon theory, in particular, when discussing the quantum capacity. The quantum capacity of a channel can be understood intuitively as the 
highest asymptotic rate at which quantum information can be transmitted, 
or equivalently entanglement can be reliably generated, between a sender 
and a receiver over many uses of the channel. Although determining this capacity is difficult in general, for degradable channels it admits a single-letter characterisation.~\cite{DevetakShor2005Capacity,Yard2005CapacityDegrb}.
A channel $\mathcal{C}:\mathcal{L}(\mathcal{H}_{\I})\to
\mathcal{L}(\mathcal{H}_{\O})$ with complementary channel $\mathcal{E}:\mathcal{L}(\mathcal{H}_{\I})\to
\mathcal{L}(\mathcal{H}_{E})$ is degradable if there exists a CPTP map $\map{D}:\mathcal L(\mathcal H_{\mathrm O})\to\mathcal L(\mathcal H_{\mathrm E})$ such that $\mathcal{E}=\map{D}\circ\map{C}$ \cite{DevetakShor2005Capacity}. 

Since Thm.~\ref{thm:ComplTC} identifies $\map{T}_{\textup{CP}}^{(N\to K)}$ as a complement of $\map{W}_{\textup{CP}}$, these properties can be determined explicitly for the optimal cloning channel.

\begin{restatable}[Degradability and quantum capacity of optimal cloning]
{corollary}{CloningCapacity}\label{cor:cloning_capacity}
The optimal universal symmetric cloning channel
$\map{W}_{\textup{CP}}$ is degradable, that is,
\begin{align}\label{eq:deg}
\map{T}_{\textup{CP}}^{(N\to K)}
=
\map{T}_{\textup{CP}}^{(N+K\to K)}
\circ
\map{W}_{\textup{CP}},
\end{align}
where the degrading channel is $\map{T}_{\textup{CP}}^{(N+K\to K)}$.
Its quantum and private classical capacities are
\begin{align}
Q\big(\map{W}_{\textup{CP}}\big)
=
C_p\big(\map{W}_{\textup{CP}}\big)
=
\log_2\left(\frac{d_S^{N+K}}{d_S^K}\right).
\end{align}
By contrast, the optimal transposition channel satisfies
\begin{align}
Q\big(\map{T}_{\textup{CP}}^{(N\to K)}\big)
=
C_p\big(\map{T}_{\textup{CP}}^{(N\to K)}\big)
=
0.
\end{align}
\end{restatable}

The proof of the Cor.~\ref{cor:cloning_capacity} follows from the fact that the universal cloning channel is a Hadamard channel, that is, its complementary channel is entanglement breaking~\cite{King2005Hadm_ch, Bradler2010Hadm_Entgbreak}. Cor.~\ref{cor:cloning_capacity} extends previous results on the quantum capacity of the universal cloning channels~\cite{BradlerDutilHaydenMuhammad2010,Bradler2009Adit_additiveCh,JochymOConnor2011}, previously restricted to $d=2$ and specific cases of the number of inputs and outputs. Additionally, since the universal cloning channel is a Hadamard channel, it is also conjugate degradable, and thus this proves the conjecture presented in Ref.~\cite{BradlerDutilHaydenMuhammad2010}.

Operationally, relation~\eqref{eq:deg} shows that the cloning output is the stronger of the two complementary outputs: given only $\map{W}_{\mathrm{CP}}(\rho)$, one can apply the fixed degrading channel $\map{T}_{\mathrm{CP}}^{(N+K\to K)}$ to reconstruct exactly $\map{T}_{\mathrm{CP}}^{(N\to K)}(\rho)$, without any knowledge of the input state $\rho$. This is analogous to degraded quantum broadcast channels, where one receiver is operationally stronger than the other~\cite{Wang2017HdmQmBroadCh}.

{\it A quantum circuit for transposition and cloning} --- %
\begin{figure}
\centering
\includegraphics[width=\columnwidth]{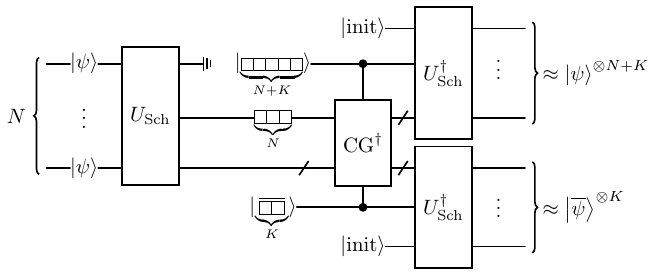}
\caption{A circuit realisation of simultaneous optimal pure-state cloning and transposition. The irrep labels $(N+K)$ and $\overline{(K)}$ select the two output irreps of the inverse Clebsch--Gordan transform, while the input label $(N)$ and its unique multiplicity state select the corresponding intertwining block. At the end of the circuit, tracing out the bottom $K$ wires gives the optimal cloning channel, while tracing out the top $N+K$ wires gives the optimal transposition channel.}
\label{fig:circuit}
\end{figure}

Inspired by the circuit constructions for unitary-equivariant and permutation-invariant channels developed in Ref.~\cite{Mancinska2025classification}, in Fig.~\ref{fig:circuit} we present a quantum circuit that simultaneously realises optimal pure-state cloning and transposition. Details on the circuit implementation and the required concepts from representation theory are in the Appendix~\ref{app:circuit}.

As a brief intuition, we start with a symmetric input that, in the Schur basis, is entirely supported on the irrep $W_{(N)}$, and the inverse Clebsch--Gordan transform implements the $U(d)$-equivariant embedding $W_{(N)}\rightarrow W_{(N+K)}\otimes\overline{W_{(K)}}$, , where $W_{(N)}$ occurs with multiplicity one. In the final step, the Schur transform simply maps these two output irreps back to the physical systems, yielding the $N+K$ cloning outputs and the $K$ transposition outputs, respectively.

Let us also note that, although the use of the CG transform has conceptual clarity, its direct implementation as a circuit gate is not efficient. By contrast, Ref.~\cite{Mancinska2025classification} provides an efficient implementation of this embedding using a sequence of Clebsch--Gordan and dual Clebsch--Gordan transforms, with polynomial gate complexity. Whether the recent high-dimensional Schur-transform techniques of Ref.~\cite{hdkrovi25burchardt} can further reduce the dimension dependence of the full construction, including the dual Clebsch--Gordan step, remains an interesting open question.

{\it Optimal $N\to 1$ transposition on mixed states ---}%
We now turn to mixed input states and consider the problem of $N\to 1$ transposition. In the pure-state setting, performance was quantified by the worst-case fidelity over all pure input states and, due to covariance, the worst-case and average fidelities coincide. For mixed states, this equivalence no longer holds: worst-case and average fidelities are generally different, and there is no canonical measure for averaging over density operators. For instance, the Hilbert–Schmidt and Bures measures are inequivalent and have different relevant properties~\cite{2001Zyczkowski_geometry,2001Zyczkowski_meas,Osipov2010}. Moreover, the fidelity between two density operators, $F(\rho,\sigma):=\Tr\left(\sqrt{\sqrt{\rho}\,\sigma\,\sqrt{\rho}}\right)^2$, is no longer a bilinear function, which makes its analysis significantly more challenging.

We therefore adopt the visibility with respect to white noise as the figure of merit in the mixed-state setting, which is the standard performance quantifier in the structural physical approximation (SPA) approach~\cite{Bae2017,Horodecki2003multicopy} to non-CP maps. This also allows us to employ the methods of $N$-copy CP maps introduced in Ref.~\cite{Dong2019positive}. 
 
Our third main result characterises the optimal $N\to 1$ approximation of density-matrix transposition under this figure of merit.

\begin{restatable}{theorem}{TransMix}\label{thm:TransMix}
There exists a CPTP linear map  ${\map{C}_N : \mathcal{L} (\mathbb{C}{^d})^{\otimes N} \to \mathcal{L}\left(\mathbb{C}{^{d}}\right)}$ such that 
\begin{equation} \label{eq:eta_mixed}
\map{C}_N (\rho^{\otimes N})
=\eta\,\rho^\mathsf{T}+(1-\eta) \frac{\mathbb{1}}{d}, \quad \forall \rho\in\mathcal L(\mathbb{C}^d),\; \text{s.t.}\; \Tr(\rho) = 1,
\end{equation}
if and only if $-\tfrac{1}{d-1} \leq \eta\leq \eta_{\max}$, where
\begin{equation} \label{eq:eta_max}
\eta_{\max}=
\begin{cases}
\frac{1}{d+1} & \textup{if } N\le d-1,\\[4pt]
\frac{N}{d(d-1)+N} & \textup{if } N\ge d-1.
\end{cases}
\end{equation}
\end{restatable}

Thm.~\ref{thm:TransMix} extends Thm.~10 of Ref.~\cite{Dong2019positive} by showing that one of the bounds identified there is tight and determines the exact parameter range in which the transposition map admits a SPA in the multicopy regime. The proof presented in the Appendix~\ref{appendix:D} combines representation-theoretic tools based on Jucys-Murphy elements with Thm.~2 of Ref.~\cite{Dong2019positive}, which provides a minimal-eigenvalue characterisation of $N$-copy completely positive maps.

It is instructive to compare the maximal visibility for arbitrary density matrices obtained in Thm.~\ref{thm:TransMix} with the maximal visibility achievable when the approximation is restricted to pure states, as in Eq.~\eqref{eq:shrinking}. The upper bound $\eta=N/(N+d)$ is attained by the single-site marginal of the optimal transposition channel in Eq.~\eqref{eq:shrinking}. More generally, due to the covariance properties of the transposition map, the same techniques used in Thm.~\ref{thm:TransNK} show that there exists a CPTP map $\map{C}_N:\mathcal{L}(\mathbb{C}^{d})^{\otimes N}\to\mathcal{L}(\mathbb{C}^{d})$ such that
\begin{equation} \label{eq:eta_pure}
\map{C}_N\of[\big]{\ketbra{\psi}^{\otimes N}}
=\eta\,\ketbra{\psi}^\mathsf{T}+(1-\eta)\frac{\mathbb{1}}{d}, \quad \forall\,\ket{\psi}\in\mathbb{C}^d,
\end{equation}
if and only if $-\tfrac{1}{d-1}\leq\eta\leq \tfrac{N}{N+d}$.

For $d=2$, the critical visibility for pure states coincides with the maximal visibility $\eta_{\max}$ for arbitrary density matrices given in Eq.~\eqref{eq:eta_max}. For $N > 1$ and $d > 2$, the range of admissible visibilities for pure-state inputs is strictly larger than the corresponding range for arbitrary mixed-state inputs.

{\it Discussion ---}%
We characterised the optimal physical approximation of $N\to K$ pure-state transposition and proved that this channel is complementary to the optimal universal pure-state $N\to N+K$ cloner for arbitrary $d$, $N$, and $K$. 
This complementarity has a direct interpretation in quantum communication and cryptography. The common isometry that delivers the optimal cloned state to one output necessarily delivers the optimal transposition state (up to an isometry) to another output, and \textit{vice versa}. From an adversarial perspective this can be interpreted as information leaked to a possible eavesdropper.

We also prove that optimal $N\to K$ pure-state transposition is entanglement breaking, which implies that optimal pure-state cloning is a Hadamard channel and hence degradable, and allows us to calculate explicitly its quantum capacity and its private classical capacity for every $d,N,K$, beyond the cases covered previously~\cite{BradlerDutilHaydenMuhammad2010,Bradler2009Adit_additiveCh,JochymOConnor2011}. Operationally, this means that, despite the input to the common isometry being unknown, the receiver on the cloning side can completely reconstruct the output on the transposition side. Additionally, we present a possible circuit realisation for the isometry that simultaneously realises optimal cloning and transposition.

For mixed inputs, we determine the exact interval of $\eta$ for which the noisy transpose $\rho^{\otimes N}\mapsto\eta\rho^\mathsf{T}+(1-\eta)\mathbb{1}/d$ admits a CPTP realisation, thereby identifying the optimal $N\to1$ multicopy SPA. 
We note that, for the conventional SPA of the single-copy partial-transpose map, entanglement can nevertheless be certified when the minimum output eigenvalue falls below a known threshold~\cite{Horodecki2000direct}. It remains an open question whether the multicopy results obtained here can lead to new operational protocols for entanglement detection.

Previous results on qubit broadcasting and universal-NOT broadcasting~\cite{DAriano2005Superbroadcasting,Buscemi2006Superbroadcasting,Chiribella_2007Superbroad} also motivate the question of whether the cloning--transposition complementarity established here extends to mixed-state broadcasting beyond qubits. Since relevant figures of merit for mixed states are generally inequivalent, such an extension may lead to distinct notions of optimality. It would therefore be interesting to determine whether a form of complementarity between broadcasting and transposition broadcasting persists in this more general setting, and whether it could provide a route towards mixed-state telecloning protocols~\cite{Murao1999telecloning,Dur2000telecloning}. Additionally,  we expect that interesting and relevant channel complementarity relations may hold for general unitary-equivariant and permutation-invariant channels. For such generalisations, the classification and circuit approach  of~\cite{Mancinska2025classification} should provide the framework for extending our results beyond transposition and cloning. 

Finally, a natural related direction is to analyse the possible relationship between unitary/channel cloning~\cite{Chiribella2008cloning,Sekatski2025cloning} and unitary/channel complex conjugation~\cite{Miyazaki2017conjugation,Ebler2022conjugation}. We believe some of the methods used here apply for analysing this unitary generalisation of the pure state problem.

{\it Acknowledgments ---} %
We acknowledge Dagmar Bruss, Adán Cabello, Nicolas Gisin, Barbara Kraus, Elias Theil, and Mio Murao, for useful discussions. VB and MTQ acknowledge support by QuantEdu France, a State aid managed by the French National Research Agency for France 2030 with the reference ANR-22-CMAS-0001.  MTQ is supported by the French Agence Nationale de la Recherche (ANR) under grant JCJC HOQO-KS. DG acknowledges support by NWO grant NGF.1623.23.025 (“Qudits in theory and experiment”). M.S. acknowledges that the project is co-financed by the Polish National Agency for Academic
Exchange (NAWA) under the Polonium Programme (grant no.~BPN/BFR/2024/1/00026/U/00001),
which supported his research mobility. 
M.S. also acknowledges support from the Polish National Science Centre (NCN) through the Sonata Bis (grant no.~UMO-2024/54/E/ST2/00316), which supported the research presented in this work.\\

{\it AI disclosure ---} This work was initiated and largely completed before current AI tools became available. The scientific content, including the results, proofs, and conclusions, was developed by the authors. At a later stage, AI tools were used to improve the language and presentation of the manuscript and to review the proofs. In a limited number of cases, revisions and simplifications were suggested. All such suggestions were independently assessed by the authors, and only a subset was incorporated into the final manuscript.


\clearpage

\onecolumngrid

\appendix

\let\oldaddcontentsline\addcontentsline
\renewcommand{\addcontentsline}[3]{} 
\section*{Appendix}
\let\addcontentsline\oldaddcontentsline 

\tableofcontents
\section{Preliminaries}

Throughout the appendix, we will make use of the Choi--Jamiołkowski isomorphism. For a linear map ${\map{C}:\mathcal{L}(\mathcal{H}_{\I})\to \mathcal{L}(\mathcal{H}_{\O})}$, 
its Choi operator is the operator $J_{\I\O}\in \mathcal{L}(\mathcal{H}_{\I}\otimes \mathcal{H}_{\O})$ defined by
\begin{equation}
J_{\I\O}
:=
\sum_{i,j}
\ketbra{i}{j}_{\I}\otimes \map{C}\!\left(\ketbra{i}{j}_{\I}\right),
\end{equation}
where $\{\ket{i}\}_{i}$ is the computational basis of $\mathcal{H}_{\I}$. Furthermore, the Choi vector of a linear operator $A:\mathcal{H}_\I \to \mathcal{H}_\O$ is defined as
\begin{equation}
\dket{A}_{\I\O}
:= \sum_i \ket{i}_\I \otimes (A\ket{i})_\O.
\end{equation}
Hence, we have that  $\dket{\mathbb{1}}_{\I\I}\in \mathcal{H}_{\I}\otimes\mathcal{H}_{\I}$ is the (unnormalised) maximally entangled vector
$
\dket{\mathbb{1}}_{\I\I}
=
\sum_i
\ket{i}_{\I}\otimes\ket{i}_{\I},
$ 
 and that 
\begin{equation}
\dket{A}_{\I\O}
=
(\mathbb{1}_{\I}\otimes A\,)\dket{\mathbb{1}}_{\I\I}.
\end{equation}

It is also useful to note the identity
\begin{equation}
\dket{A}_{\I\O}
=
(\mathbb{1}\otimes A)\dket{\mathbb{1}}_{\I\I}
= 
(A^\mathsf{T}\otimes \mathbb{1})\dket{\mathbb{1}}_{\O\O},
\end{equation}
where the transpose is taken with respect to the computational basis $\{\ket{i}\}_i$.

The action of the map $\map{C}$ on an arbitrary operator $\rho\in \mathcal{L}(\mathcal{H}_{\I})$ can then be written as
\begin{equation}
\map{C}(\rho)
=
\Tr_{\I}\!\left[
\left(\rho^\mathsf{T}\otimes \mathbb{1}_{\O}\right)J_{\I\O}
\right],
\end{equation}
where the transpose is taken in the same basis used in the definition of the Choi operator, and $\Tr_{\I}$ denotes the partial trace over the input space $\mathcal{H}_{\I}$. For the $N\!\to\!K$ pure-state transposition problem, the input and output Hilbert spaces are 
$\mathcal{H}_{\I}:=(\mathbb{C}^{d})^{\otimes N}$ and $\mathcal{H}_{\O}:=(\mathbb{C}^{d})^{\otimes K}$.
For the $N\!\to\!N+K$ cloning problem, the output space is instead 
$\mathcal{H}_{\O}:=(\mathbb{C}^{d})^{\otimes (N+K)}$.

We denote the projector onto the symmetric subspace $\mathrm{sym}^{M}(\mathbb{C}^{d})\subseteq (\mathbb{C}^{d})^{\otimes M},$ with $\Pi_{\mathrm{sym}}^{(M)}$, while its dimension we denote by
\begin{equation}
d_{S}^{M}:=\dim \bigl(\mathrm{sym}^{M}(\mathbb{C}^{d})\bigr)
= \tr(\Pi_{\mathrm{sym}}^{(M)}) 
= \binom{M+d-1}{d-1}.
\end{equation}

Finally, we will repeatedly use the identity
\begin{equation}
\int \mathrm{d}\psi \;
\ketbra{\psi}^{\otimes M}
=
\frac{\Pi_{\mathrm{sym}}^{(M)}}{d_{S}^{M}},
\end{equation}
where $\mathrm{d}\psi$ denotes the Haar measure over pure states  ${\ket{\psi}\in\mathbb{C}^d}$.

\section{Optimal approximation of $N\to K$ pure state transposition}\label{app:NKtrans}

\subsection{Covariance properties of the problem}

\begin{lemma} \label{lemma:covariance}
Let $\map{C}:\mathcal{L}(\mathbb{C}^{d})^{\otimes N}\to\mathcal{L}(\mathbb{C}^{d})^{\otimes K}$ 
be a quantum channel that attains an average fidelity $F_{\mathrm{avg}, \map{C}}$ in the $N\to K$ state transposition task. That is,
\begin{equation}\label{eq:fid_trans_ifx}
F_{\mathrm{avg}, \map{C}}:=
\int \! \mathrm{d}\psi
F\!\Big(\mathcal{C}\left(\ketbra{\psi}^{\otimes N}\right),\,{\ketbra{\psi}^\mathsf{T}}^{\otimes K}\Big).
\end{equation}

There exists another quantum channel 
$\map{C}':\mathcal{L}(\mathbb{C}^{d})^{\otimes N}\to\mathcal{L}(\mathbb{C}^{d})^{\otimes K}$ 
that attains the same performance $F_{\mathrm{avg}, \map{C}}$ and satisfies the unitary covariance relation
\begin{align} \label{eq:cov1_SM}
\mathcal{C}'\!\left(U^{\otimes N} \rho (U^{\otimes N})^{\dagger}\right)
&= \overline{U}^{\otimes K}\, \mathcal{C}'(\rho)\, (U^\mathsf{T})^{\otimes K},\\
\forall U\in \mathrm{SU}(d), &\quad \forall \rho\in\mathcal{L}(\mathbb{C}^d)^{\otimes N},
\end{align}
as well as the permutation covariance relation
\begin{align} \label{eq:cov2_SM}
\mathcal{C}'\!\left(V_\pi \rho V_\pi^\dagger\right)
&= V_\sigma\, \mathcal{C}'(\rho)\, V_\sigma^\dagger, 
\quad  \forall \pi\in S_N, \quad \sigma\in S_K,
\end{align}
where $V_\pi$ and $V_\sigma$ denote the unitary representations of the symmetric groups acting on $(\mathbb{C}^d)^{\otimes N}$ and $(\mathbb{C}^d)^{\otimes K}$, respectively.
\end{lemma}

\begin{proof}

We first show that, for the $N\to K$ transposition problem, one can, without loss of performance, restrict the search to channels satisfying the covariance condition
\begin{align}
\mathcal{C}'\!\left(U^{\otimes N} \rho (U^{\otimes N})^{\dagger}\right)
&= \overline{U}^{\otimes K}\, \mathcal{C}'(\rho)\, (U^\mathsf{T})^{\otimes K}, \\
\forall U\in \mathrm{SU}(d), &\quad \forall \rho\in\mathcal{L}(\mathbb{C}^d)^{\otimes N}.
\end{align}
This condition can equivalently be expressed as a symmetry property of the Choi operator. Namely, we have
\begin{align}
J &= (\mathbb{1} \otimes \mathcal{C}') (\dketbra{\mathbb{1}}{\mathbb{1}}_{\I' \I})\\
&= (\mathbb{1} \otimes \mathcal{C}') 
\left(
\left(U^{\otimes N} \otimes \overline{U}^{\otimes N}\right)
\dketbra{\mathbb{1}}{\mathbb{1}}
\left(U^{\otimes N} \otimes \overline{U}^{\otimes N}\right)^{\dagger}
\right)\\
&= \left(U_{\I}^{\otimes N} \otimes U_{\O}^{\otimes K}\right)
(\mathbb{1} \otimes \mathcal{C}') 
\left(\dketbra{\mathbb{1}}{\mathbb{1}} \right)
\left(U_{\I}^{\otimes N} \otimes U_{\O}^{\otimes K}\right)^{\dagger},
\end{align}
where $\mathcal{L}(H_{\I}) \cong  \mathcal{L}(H_{\I'}) \cong \mathcal L(\mathcal H^{\otimes N})$, $\mathcal{L}(H_{\O}) \cong \mathcal L(\mathcal H^{\otimes K})$, and $J \in \mathcal{L}(H_{\I} \otimes H_{\O})$. 
In the above we used the identity $(U \otimes \overline U)\dket{\mathbb{1}}=\dket{\mathbb{1}}$ for the maximally entangled vector. Hence,
\begin{align}
[J,\,U^{\otimes N}\otimes U^{\otimes K}] = 0, \quad \forall U\in \mathrm{SU}(d).
\end{align}

Given any feasible Choi operator $J$, we define its unitary twirl
\begin{align}
\widetilde J
:= \int \mathrm{d}U\;
\bigl(U^{\otimes N}\otimes U^{\otimes K}\bigr)\,
J\,\bigl(U^{\otimes N}\otimes U^{\otimes K}\bigr)^\dagger .
\end{align}

Let us now show that if $J \geq 0$ and $\Tr_O(J)=\id_{\I}$, then $\widetilde J\ge 0$ and $\Tr_{\O}(\widetilde J)=\mathbb 1_{\I}$. Namely, since $\widetilde{J}$ is a convex combination of positive semidefinite operators, we have $\widetilde{J}\geq0$. Furthermore,
\begin{align}
&\Tr_\O(\widetilde{J}) = \Tr_\O\int \mathrm{d}U\;
\bigl(U^{\otimes N}\otimes U^{\otimes K}\bigr)\,
J\,\bigl(U^{\otimes N}\otimes U^{\otimes K}\bigr)^\dagger \\ 
&=\Tr_\O\int \mathrm{d}U\;
\bigl(U^{\otimes N}\otimes (U^\dagger U)^{\otimes K}\bigr)\,
J\,\bigl(U^{\otimes N}\otimes \id^{\otimes K} \bigr)^\dagger \\ 
&=\Tr_\O\int \mathrm{d}U\;
\bigl(U^{\otimes N}\otimes \id^{\otimes K}\bigr)\,
J\,\bigl(U^{\otimes N}\otimes \id^{\otimes K} \bigr)^\dagger \\ 
&= \int \mathrm{d}U\;
\bigl(U^{\otimes N}\bigr)\,
\Tr_\O(J)\,\bigl(U^{\otimes N}\bigr)^\dagger \\ 
&= \int \mathrm{d}U\;
\bigl(U^{\otimes N}\bigr)\,
\id_{\I} \,\bigl(U^{\otimes N}\bigr)^\dagger \\ 
&= \id_\I.
\end{align}
Hence, when $J$ is feasible, $\widetilde J$ is also feasible. 

Let us now define the performance operator for the average-fidelity problem as $\Omega:=\Pi_{\mathrm{sym}}^{(N+K)}/d_S^{N+K}$, so that the average fidelity of a channel with Choi operator $J$ is given by $\Tr(J\Omega)$.

Since $\Omega$ is invariant under $U^{\otimes (N+K)}$,
\begin{align}
\Tr(\widetilde J\,\Omega)=\Tr(J\,\Omega).
\end{align}

Additionally, without loss of performance, we can restrict the search to Choi operators satisfying the permutation covariance condition
\begin{align}
[J,\,V_{\pi,\textup{I}}\otimes V_{\sigma,\textup{O}}] = 0, \quad \forall \pi\in S_N, \; \sigma\in S_K,
\end{align}
where $V_{\pi,\textup{I}}$ and $V_{\sigma,\textup{O}}$ are the unitary representations of the permutation groups $S_N$ and $S_K$ acting on the input and output systems, respectively.

To obtain a channel satisfying both covariance properties simultaneously, we apply the permutation twirl to $\widetilde J$ and define
\begin{align}
J'
:= \frac{1}{N!K!}\sum_{\pi\in S_N,\sigma\in S_K}
\bigl(V_{\pi,\textup{I}}\otimes V_{\sigma,\textup{O}}\bigr)\,
\widetilde J\,\bigl(V_{\pi,\textup{I}}\otimes V_{\sigma,\textup{O}}\bigr)^\dagger .
\end{align}

Then $J'\ge 0$ and $\Tr_{\O}(J')=\mathbb 1_{\I}$, hence $J'$ is feasible. Since $\Omega$ is invariant under $V_{\pi,\textup{I}}\otimes V_{\sigma,\textup{O}}$, we have
\begin{align}
\Tr(J'\,\Omega)=\Tr(\widetilde J\,\Omega)=\Tr(J\,\Omega).
\end{align}
Moreover, because permutation operators commute with $U^{\otimes N}\otimes U^{\otimes K}$, the operator $J'$ satisfies both
\begin{align}
[J',\,U^{\otimes N}\otimes U^{\otimes K}] = 0, \quad \forall U\in \mathrm{SU}(d),
\end{align}
and
\begin{align}
[J',\,V_{\pi,\textup{I}}\otimes V_{\sigma,\textup{O}}] = 0, \quad \forall \pi\in S_N,\; \sigma\in S_K.
\end{align}
Hence the corresponding channel $\mathcal C'$ attains the same performance as $\mathcal C$ and satisfies Eqs.~\eqref{eq:cov1_SM} and~\eqref{eq:cov2_SM}.

\end{proof}

Let us now show that the average-fidelity formulation is equivalent to the worst-case formulation.

\begin{lemma}\label{lemma:avg_wc}
Let $\map{C}:\mathcal{L}(\mathbb{C}^{d})^{\otimes N}\to\mathcal{L}(\mathbb{C}^{d})^{\otimes K}$ be a quantum channel that attains fidelity in the $N\to K$ state transposition task, 
\begin{align}
F_{\map{C}}(\psi)
&:=F \Big(
\mathcal{C}(\ketbra{\psi}^{\otimes N}),
{\ketbra{\psi}^{\mathsf{T}}}^{\otimes K}
\Big).
\end{align}

Then the optimal average and worst-case fidelities coincide:
\begin{equation}
\max_{\map{C}\in\mathrm{CPTP}}\min_{\ket{\psi}}F_{\map{C}}(\psi)
=
\max_{\map{C}\in\mathrm{CPTP}} \int\!\mathrm{d}\psi\,F_{\map{C}}(\psi).
\end{equation}
\end{lemma}

\begin{proof}
Let $F_{\mathrm{avg}, \map{C}}
:=\int\mathrm{d}\psi\,F_{\map{C}}(\psi)$ be the average fidelity  of $\map{C}$ and let $F_{\mathrm{wc}, \map{C}}:=\min_{\ket{\psi}}F_{\map{C}}(\psi)$
be the worst-case fidelity. By Lemma~\ref{lemma:covariance}, every $\map{C}$ admits a covariant channel $\map{C}'$ with the same average fidelity, $F_{\mathrm{avg}, \map{C}'}=F_{\mathrm{avg}, \map{C}}$. 

Once we established that, without loss of generality, the optimisation may be restricted to covariant channels, we can show that the fidelity $F_{\map{C}'}(\psi)$ of any covariant channel is independent of the input state and therefore coincides with its worst-case fidelity.

Namely, we have,
{\small
\begin{align}
F_{\map{C}'}(\psi)
&=
F\left(
\map{C}'\left(
U^{\otimes N}\ketbra{\phi}^{\otimes N}
(U^\dagger)^{\otimes N}
\right),
\left(
U\ketbra{\phi}U^\dagger
\right)^{\mathsf T\otimes K}
\right)
\\
&=
F\left(
\overline{U}^{\otimes K}
\map{C}'(\ketbra{\phi}^{\otimes N})
(U^{\mathsf T})^{\otimes K},
\overline{U}^{\otimes K}
{\ketbra{\phi}^{\mathsf T}}^{\otimes K}
(U^{\mathsf T})^{\otimes K}
\right)
\\
&=
F\left(
\map{C}'(\ketbra{\phi}^{\otimes N}),
{\ketbra{\phi}^{\mathsf T}}^{\otimes K}
\right)
=
F_{\map{C}'}(\phi).
\end{align}
}
where in the second equality we used the covariance relation
\begin{align}
\map{C}'\left(
U^{\otimes N}\rho(U^\dagger)^{\otimes N}
\right)
=
\overline{U}^{\otimes K}
\map{C}'(\rho)
(U^{\mathsf T})^{\otimes K}.
\end{align}
Hence, $F_{\map{C}'}(\psi)$ is independent of the pure state
$\ket{\psi}$ and consequently, \begin{equation}
    F_{\mathrm{avg},\map{C}'} = F_{\mathrm{wc},\map{C}'}.
\end{equation}

Since every channel $\map{C}$ admits a covariant channel $\map{C}'$
with $F_{\mathrm{avg},\map{C}'}=F_{\mathrm{avg},\map{C}}$, and since
for covariant channels $F_{\mathrm{wc},\map{C}'}=
F_{\mathrm{avg},\map{C}'}$, we have
\begin{align}
\max_{\map{C}\in\mathrm{CPTP}}
F_{\mathrm{wc},\map{C}}
\geq
F_{\mathrm{wc},\map{C}'}
=
F_{\mathrm{avg},\map{C}'}
=
F_{\mathrm{avg},\map{C}}.
\end{align}
Since this holds for every channel $\map{C}$, we may also maximize over
$\map{C}$ on the right-hand side, yielding
\begin{align}
\max_{\map{C}\in\mathrm{CPTP}}
F_{\mathrm{wc},\map{C}}
\geq
\max_{\map{C}\in\mathrm{CPTP}}
F_{\mathrm{avg},\map{C}}.
\end{align}
On the other hand,
$F_{\mathrm{wc},\map{C}}\leq F_{\mathrm{avg},\map{C}}$ for every
$\map{C}$, and therefore
\begin{align}
\max_{\map{C}\in\mathrm{CPTP}}
F_{\mathrm{wc},\map{C}}
\leq
\max_{\map{C}\in\mathrm{CPTP}}
F_{\mathrm{avg},\map{C}}.
\end{align}
Thus,
\begin{align}
\max_{\map{C}\in\mathrm{CPTP}}
F_{\mathrm{avg},\map{C}}
=
\max_{\map{C}\in\mathrm{CPTP}}
F_{\mathrm{wc},\map{C}}.
\end{align}
\end{proof}

\subsection{An SDP formulation of the problem}

In this section, we use a general semidefinite-programming framework~\cite{AudenaertDeMoor2002,Skrzypczyk2023SDP}.

\begin{lemma}\label{ref:dual} The maximum average fidelity for approximating the transpose map from $N$ input copies to $K$ output copies on pure states, given by \begin{equation}\label{eq:fid_trans_appendix} \max_{\mathcal{C}\in\mathrm{CPTP}} \int \! \mathrm{d}\psi F\!\Big( \mathcal{C}(\ketbra{\psi}^{\otimes N}), {\ketbra{\psi}^\mathsf{T}}^{\otimes K} \Big), \end{equation} can be expressed as the following semidefinite program \begin{align} \max_{J \in \mathcal{L}(\mathcal{H}_\I\otimes \mathcal{H}_\O )}\quad & \Tr\!\left( J \frac{\Pi^{(N+K)}_{\mathrm{sym}}}{d_S^{N+K}} \right) \label{eq:SDP_primal} \\ \textup{s.t.}\quad J& \ge 0, \\ \Tr_{\O}(J)&=\mathbb{1}_{\I}, \end{align} where $J\in \mathcal L(\mathcal H_{\I}) \otimes \mathcal L(\mathcal H_{\O})$ is the Choi operator of $\mathcal C$, where $\mathcal L(\mathcal H_{\I}) \cong \mathcal L(\mathbb{C}^d)^{\otimes N}$ and $\mathcal L(\mathcal H_{\O}) \cong \mathcal L(\mathbb{C}^d)^{\otimes K}$, $\Tr_O$ is the partial trace over the space $\mathcal{H}_\O$. 

The dual SDP corresponding to~\eqref{eq:SDP_primal} is \begin{align} \min_{X \in \mathcal L(\mathbb{C}^d)^{\otimes N}} \; & \Tr(X) \label{eq:dual} \\ \textup{s.t.}\quad & X \otimes \mathbb{1}_{d^K} \ge \frac{\Pi^{(N+K)}_{\mathrm{sym}}}{d_S^{N+K}}. 
\end{align} 

\end{lemma} 

\begin{proof}

Starting from Eq.~\eqref{eq:fid_trans_appendix} we can rewrite,
\begin{align}
&\int \! \mathrm{d}\psi
F\!\Big(
\mathcal{C}(\ketbra{\psi}^{\otimes N}),
{\ketbra{\psi}^\mathsf{T}}^{\otimes K}
\Big)\\
&= \int \! \mathrm{d}\psi
\Tr\!\Big(
\Tr_\I\!\big(({{\ketbra{\psi}^\mathsf{T}}^{\otimes N}} \otimes \id^{\otimes K})J\big)
{\ketbra{\psi}^\mathsf{T}}_{\O}^{\otimes K}
\Big)\\
&= \int \! \mathrm{d}\psi
\Tr\!\Big(
J \; {{\ketbra{\psi}^\mathsf{T}}^{\otimes N}}  \otimes 
{\ketbra{\psi}^\mathsf{T}}_{\O}^{\otimes K}
\Big)\\
&=\Tr\!\Big(J_{\I\O}\int \! \mathrm{d}\psi 
({{\ketbra{\psi}^\mathsf{T}}^{\otimes N+K}})\Big)\\
&=\Tr\!\Big(J_{\I\O}\frac{\Pi^{(N+K)}_{\mathrm{sym}}}{d_S^{N+K}}\Big)
\end{align}
where $\mathcal{L}(H_{\I}) \cong \mathcal L(\mathcal H^{\otimes N})$ and $\mathcal{L}(H_{\O}) \cong \mathcal L(\mathcal H^{\otimes K})$. 
In the last equality we used that
$\int \mathrm d\psi\, \ketbra{\psi}^{\otimes (N+K)}=
\frac{\Pi^{(N+K)}_{\mathrm{sym}}}{d_S^{N+K}}$ and that the symmetric projector is invariant under transpose. Let us now introduce 
$\Omega
= \frac{\Pi^{(N+K)}_{\mathrm{sym}}}{d_S^{N+K}}$, which we call a performance operator,
to rewrite Eq.~\eqref{eq:SDP_primal} as,
\begin{align} 
\max_{J}\quad & \Tr(J\Omega) \label{eq:SDP_primal_compact}\\
\textup{s.t.}\quad & J \geq 0,\qquad \Tr_{\O}(J)=\mathbb{1}_{\I}.
\end{align}

The formulation above is a direct multicopy instance of the Choi-operator
framework of Ref.~\cite{AudenaertDeMoor2002}: the Choi operator is the primal variable, complete positivity is imposed by the semidefinite constraint $J\geq 0$, trace preservation is an affine constraint, and a feasible dual operator provides an upper bound on the achievable
fidelity.

To obtain the dual SDP, we follow the standard Lagrangian
approach~\cite{AudenaertDeMoor2002,Skrzypczyk2023SDP}. We introduce a Hermitian dual variable $X=X^\dagger$ corresponding to the equality constraint while $J\geq 0$ defines the domain of the primal variable. The Lagrangian is therefore
\begin{align}
\mathcal{L}(J,X)
&=
\Tr(J\Omega)
+
\Tr\left[
X_{\I}
\left(
\mathbb{1}_{\I}-\Tr_{\O}(J)
\right)
\right]
\\
&=
\Tr\left[
J
\left(
\Omega-X_{\I}\otimes\mathbb{1}_{\O}
\right)
\right]
+
\Tr(X_{\I}),
\end{align}
where we used the adjoint relation
$\Tr\left[X_{\I}\Tr_{\O}(J)\right]
=
\Tr\left[(X_{\I}\otimes\mathbb{1}_{\O})J\right]$. For fixed $X$, the dual function is
\begin{align}
g(X)
:=
\sup_{J\geq 0}\mathcal{L}(J,X).
\end{align}
In order to make $g(X)$ finite, we impose,
\begin{align}
X_{\I}\otimes\mathbb{1}_{\O}\geq\Omega_{\I\O}.
\end{align}
Whenever this condition holds, the first term in the Lagrangian is non-positive for every $J\geq 0$, and its supremum is attained at $J=0$, giving $g(X) = \Tr(X)$. The dual SDP is therefore
\begin{align}
\min_{X=X^\dagger}\quad
&\Tr(X)
\\
\textup{s.t.}\quad
&X_{\I}\otimes\mathbb{1}_{\O}
\geq
\Omega_{\I\O}.
\end{align}

Finally, strong duality follows directly from strict primal feasibility. In particular, the operator
$J_0:=\mathbb{1}_{\I}\otimes\mathbb{1}_{\O}/d_{\O}$ satisfies
$J_0>0$ and $\Tr_{\O}(J_0)=\mathbb{1}_{\I}$. Hence, the primal and dual optimal values coincide.

\end{proof}

\subsection{Proof of optimality}

\begin{lemma}\label{lemma:avg_trans}
The maximum average fidelity for approximating the transpose map from $N$ input copies to $K$ output copies on pure states is given by
\begin{equation}\label{eq:fid_trans_avg}
    \max_{\map{C}
    \in\mathrm{CPTP}}
    \int \! \mathrm{d}\psi
    F\!\Big(\mathcal{C}\left(\ketbra{\psi}^{\otimes N}\right),\,{\ketbra{\psi}^\mathsf{T}}^{\otimes K}\Big)\!=\!\frac{d_S^{N}}{d_S^{N+K}}
\end{equation}
where $d_S^N:=\dim(\mathrm{sym}^N(\mathbb{C}^d))=\binom{N+d-1}{d-1}$. 

The optimal performance is attainable by an estimation strategy, which acts on the symmetric subspace as ${\map{T}_\textup{CP}: \mathcal{L}\left(\mathrm{sym}^N\left(\mathbb{C}^d\right)\right)\to\mathcal{L}(\mathbb{C}^{d})^{\otimes K}}$
\begin{align} \label{eq:T_CP_avg}
    \mathcal{T}_\textup{CP}(\rho)
    :=& d_S^{N}
    \int \mathrm{d}\psi
    \Tr\!\left[\rho \ketbra{\psi}^{\otimes N}\right]\,{\ketbra{\psi}^{\mathsf{T}}}^{\otimes K} \\
    =&  \frac{d_S^{N}}{d_S^{N+K}} \Tr_{1\ldots N}\left[ \Pi^{(N+K)}_\textup{sym} \, \of*{\rho^\mathsf{T} \otimes \id_d^{\otimes K} }\right],
\end{align}
where $\Pi^{(N+K)}_{\mathrm{sym}}$ is the projector onto the $N+K$ symmetric subspace, $\id_d$ is the identity on $\mathbb{C}^d$, and $\Tr_{1\ldots N}$ is the partial trace on the first $N$ $d$-dimensional subsystems.
\end{lemma}
\begin{proof}

Before presenting the proof, we note that the attainability of the bound follows almost directly from the fact that pure states can be estimated with average fidelity $\frac{d_S^N}{d_S^{N+1}}$ using the Hayashi measurement~\cite{Hayashi1998estimation,2013Harrow_church}. One can apply the Hayashi measurement to estimate the input state $\ket{\psi}$ and then prepare the state ${\ketbra{\psi}^\mathsf{T}}^{\otimes K}$ as the output. A straightforward calculation shows that this strategy achieves an average fidelity of $\frac{d_S^N}{d_S^{N+K}}$. That being said, we will present the proof without explicitly relying on previous results. Instead, we show that the optimal value can be attained by an estimation strategy through a direct analysis of the Choi operator of the optimal channel.

\paragraph{ Attainability of the optimal value} Let us now show that the optimal value of the optimisation problem is attainable. 
Since $\Omega$ is supported on the symmetric subspace, and under the covariance assumption Schur's lemma implies that $J$ decomposes as $J = \alpha\,\Pi^{(N+K)}_{\mathrm{sym}} + B$, where $B\ge 0$ and $B\Pi^{(N+K)}_{\mathrm{sym}}=\Pi^{(N+K)}_{\mathrm{sym}}B=0$, we have
\begin{align}\label{eq:alpha}
\Tr(J\Omega)
=\Tr\!\left(J \frac{\Pi^{(N+K)}_{\mathrm{sym}}}{d_S^{N+K}}\right)
=\frac{1}{d_S^{N+K}}\Tr\!\left(\alpha\,\Pi^{(N+K)}_{\mathrm{sym}}\right)
=\alpha.
\end{align}
Hence, maximising $\Tr(J\Omega)$ is equivalent to maximising the coefficient $\alpha$. 

This motivates the following ansatz for the Choi operator
\begin{align}\label{eq:choi_estm}
J
&= \frac{d^{N}_S\Pi^{(N+K)}_{\mathrm{sym,\; in,out}}}{d^{N+K}_S }+ (\mathbb{1}_{\I} -\Pi^{N}_{\mathrm{sym},\; \I} ) \otimes \sigma_{\O}.
\end{align}
where $\Tr(\sigma) = 1$ and $\sigma \geq 0$. 
Since $\Pi^{(N+K)}_{\mathrm{sym,\; in,out}}$ and $(\mathbb{1}_{\I} -\Pi^{N}_{\mathrm{sym},\; \I} ) \otimes \sigma_{\O}$ are positive semidefinite operators with orthogonal support, we conclude that $J \geq 0$. Using the identity 
$\Tr_O(\Pi^{(N+K)}_{\mathrm{sym}}) = \frac{d_S^{N+K}}{d_S^N}\,\Pi^{(N)}_{\mathrm{sym}}$, 
one can directly verify that the ansatz given in Eq.~\eqref{eq:choi_estm} satisfies 
$\Tr_{\O}(J)=\mathbb{1}_{\I}$, and therefore represents a valid feasible solution of Eq.~\eqref{eq:SDP_primal}.

We have then an explicit quantum channel, described by its Choi operator $J$, that satisfies the constraints of the optimisation problem and attains an average fidelity of $\Tr(J \Pi^{(N+K)}_{\mathrm{sym}}) \frac{1}{d_S^{N+K}} = \frac{d_S^N}{d_S^{N+K}}$. 
 \vspace*{3mm}
 
\paragraph{Attainability via estimation}
Let us now show that the ansatz given by Eq.~\eqref{eq:choi_estm} represents an estimation strategy. We start by recalling that, if $J$ is the Choi operator of a linear map $\mathcal{C}$, for any linear operator $\rho$ on the input space, we can write the action of $\mathcal{C}$ as $\mathcal C(\rho)
=
\Tr_{\I}
\!\left[
\left(
\rho_{\I}^\mathsf{T}
\otimes
\mathbb{1}_{\O}
\right)
J
\right].$

Hence, using the identity $\int \mathrm{d}\psi\; (\ketbra{\psi})^{\otimes (N+K)}
= \frac{\Pi^{(N+K)}_{\mathrm{sym}}}{d_S^{N+K}}$~\cite{2013Harrow_church} we have,
\begin{align}\label{eq:trans_map}
\mathcal{C}(\rho)
&= \Tr_{\I}\!\left[
(\rho_{\I}^\mathsf{T}\!\otimes\!\mathbb{1}_{\O})\, J
\right] \\
&= \Tr_{\I}\!\Big[
(\rho_{\I}^\mathsf{T}\!\otimes\!\mathbb{1}_{\O})
\Big(
\frac{d_S^{N}}{d_S^{N+K}}\,
\Pi^{(N+K)}_{\mathrm{sym},\,\I\O}
\Big) \notag\\
&\qquad
+ (\rho_{\I}^\mathsf{T}\!\otimes\!\mathbb{1}_{\O})
\big((\mathbb{1}_{\I}-\Pi^{N}_{\mathrm{sym},\,\I})
\otimes \sigma_{\O}\big)
\Big] \notag\\
&= d_S^{N}
\int \mathrm{d}\psi\;
\Tr\!\big(\rho^\mathsf{T}\ketbra{\psi}^{\otimes N}\big)\,
(\ketbra{\psi})_{\O}^{\otimes K} \notag\\
&\qquad
+ \Tr\!\big(\rho^\mathsf{T}(\mathbb{1}-\Pi^{N}_{\mathrm{sym}})\big)\,
\sigma_{\O} \notag\\
&= d_S^{N}
\int \mathrm{d}\psi\;
\Tr\!\big(\rho\,{\ketbra{\psi}^\mathsf{T}}^{\otimes N}\big)\,
(\ketbra{\psi})_{\O}^{\otimes K} \notag\\
&\qquad
+ \Tr\!\big(\rho^\mathsf{T}(\mathbb{1}-\Pi^{N}_{\mathrm{sym}})\big)\,
\sigma_{\O}.
\end{align}

Using the invariance of the Haar measure under complex conjugation, the change of variable yields
\begin{align}
\mathcal{C}(\rho)
=
d_S^{N}
\int \mathrm{d}\psi\;
\Tr\!\of[\big]{\rho\,\ketbra{\psi}^{\otimes N}}\,
{\ketbra{\psi}^\mathsf{T}}_{\mathrm O}^{\otimes K}&\\
+
\Tr\!\big(\rho^\mathsf{T}(\mathbb{1}-\Pi^{N}_{\mathrm{sym}})\big)\,
\sigma_{\mathrm O}.&
\end{align}

When restricted to inputs supported on the symmetric subspace, the second term in Eq.~\eqref{eq:trans_map} vanishes because $\Tr\!\big(\rho^\mathsf{T}(\mathbb{1}-\Pi^{N}_{\mathrm{sym}})\big)=0$. Hence,
\begin{align}\label{eq:trans_main}
    \mathcal T_{\mathrm{CP}}(\rho)
=
d_S^{N}
\int \mathrm{d}\psi\;
\Tr\!\big(\rho\,\ketbra{\psi}^{\otimes N}\big)\,
{\ketbra{\psi}^\mathsf{T}}_{\mathrm O}^{\otimes K}.
\end{align}

We can now recognise that Eq.~\eqref{eq:trans_main} is an estimation map,
\begin{equation}\label{eq:estm}
\mathcal E(\rho)
=
\int \mathrm{d}\phi\;
\Tr\!\left(M_\phi\,\rho\right)\,
\sigma_\phi,
\end{equation}
where $\{M_\phi\}$ is a POVM on the symmetric subspace $\mathrm{sym}^N(\mathbb C^d)$ considered in \cite{Hayashi1998estimation, 2013Harrow_church}, defined by
\begin{equation}
M_\phi
=
d^{N}_S\,\ketbra{\phi}^{\otimes N} \mathrm{d} \phi
\end{equation}
and the output states 
\begin{equation}
\sigma_\phi
=
{\ketbra{\phi}^\mathsf{T}}^{\otimes K}
\end{equation}
corresponding to $K$ copies of the transpose of the estimated state.

\vspace*{3mm}

\paragraph{A tight upper bound from the dual problem}

In order to obtain an upper bound on the optimal value of the primal problem, we consider the dual of the SDP given by Eq.~\eqref{eq:SDP_primal}. As in any convex maximisation problem, any feasible solution to the dual problem provides an upper bound on the primal~\cite{Skrzypczyk2023SDP}. As presented in Lemma~\ref{ref:dual}, the dual problem of Eq.~\eqref{eq:SDP_primal} reads as
\begin{align}
\min_{X}\quad & \Tr(X) \label{eq:dual_obj}\\
\textup{s.t.}\quad & X_{\I}  \otimes \mathbb{1}_{\O} \geq \Omega_{\I\; \O},\label{eq:dual_inq}
\end{align}
where $\mathcal{L}(H_{\I}) = \mathcal L(\mathcal H^{\otimes N})$ and $\mathcal{L}(H_{\O}) = \mathcal L(\mathcal H^{\otimes K})$, and $\Omega = \frac{\Pi^{(N+K)}_{\mathrm{sym}}}{d_S^{N+K}}$.

Given the structure of the performance operator $\Omega$, it is natural to look for an ansatz for $X$ that has support on the symmetric subspace of the input system. Hence we consider the ansatz $X_{\I} = \alpha \Pi^{(N)}_{\mathrm{sym}}$, such that
\begin{equation}\label{eq:dual_ansatz}
\alpha \Pi^{(N)}_{\mathrm{sym}}\otimes \mathbb{1}_{\O}
\;\ge\;
\frac{\Pi^{(N+K)}_{\mathrm{sym}}}{d_S^{N+K}} .
\end{equation}

In order for this to be a valid ansatz, we need to show that the operator inequality indeed holds for some $\alpha$. Namely, let $\lambda = (N)$ be the one-row Young diagram. Since the multiplicity space $V_{(N)}$, carrying irrep of $\S_N$,  is one-dimensional, we have $\Pi_{\mathrm{sym}}^{(N)} = E_{11}^{(N)}$, where $E_{11}^{(N)}$ is the only matrix unit for irrep $\lambda = (N)$. By the Pieri rule for matrix units~\cite{fulton1991representation,ram1992matrix,Yoshida2024PBT}, we have
\begin{equation}
E_{11}^{(N)} \otimes \mathbb{1}_d
=
\sum_{\mu \in (N)+\square} E^{\mu}_{1_\mu 1_\mu},
\end{equation}
where $1_\mu$ are labels for some basis vectors in $V_{\mu}$.
Since $(N)+\square = \{(N+1),(N,1)\}$, this becomes
\begin{equation}
\Pi_{\mathrm{sym}}^{(N)} \otimes \mathbb{1}_d
=
\Pi_{\mathrm{sym}}^{(N+1)} + \Pi_{(N,1)} (\Pi_{\mathrm{sym}}^{(N)} \otimes \mathbb{1}_d)
\;\ge\;
\Pi_{\mathrm{sym}}^{(N+1)} .
\end{equation}

Iterating this argument and tensoring with $\mathbb{1}_d^{\otimes K}$ yields
\begin{equation}
\Pi_{\mathrm{sym}}^{(N)} \otimes \mathbb{1}_d^{\otimes K}
=
\Pi_{\mathrm{sym}}^{(N+K)}
+
\sum_{\substack{\tilde{\lambda} \in Y_{N,K}^{(d)} \\ \tilde{\lambda} \neq (N+K)}} \sum_{T \in \mathrm{Paths}((N),\tilde{\lambda})}
E^{\tilde{\lambda}}_{T,T},
\end{equation}
where $Y_{N,K}^{(d)}$ denotes the set of Young diagrams $\tilde{\lambda}$ with $N+K$ boxes that can be obtained from the one-row diagram $(N)$ by successively adding $K$ boxes in admissible positions, and $\mathrm{Paths}((N),\tilde{\lambda})$ is the set of these paths for every $\tilde{\lambda}$. Hence
\begin{equation}
\Pi_{\mathrm{sym}}^{(N)} \otimes \mathbb{1}_d^{\otimes K}
\;\ge\;
\Pi_{\mathrm{sym}}^{(N+K)},
\end{equation}
which implies that Eq.~\eqref{eq:dual_ansatz} holds whenever
\begin{equation}
\alpha \geq \frac{1}{d_S^{N+K}}.
\end{equation}
Hence, to minimise the objective function for this ansatz, we take
\begin{equation}
\alpha_{\textup{min}} = \frac{1}{d_S^{N+K}},
\end{equation}
which gives
\begin{equation}
X_{\textup{ans}} = \frac{\Pi_{\mathrm{sym}}^{N}}{d_S^{N+K}}.
\end{equation}
Substituting into Eq.~\eqref{eq:dual_obj}, we obtain
\begin{equation}
\Tr(X_{\textup{ans}}) = \frac{d_S^{N}}{d_S^{N+K}}.
\end{equation}

Since this value matches the value attained via the primal formulation, the fidelity $\frac{d_S^{N}}{d_S^{N+K}}$ is optimal for the $N\to K$ transposition problem.

\end{proof}

\subsection{Uniqueness properties}

We start this section with a useful lemma that will be used in the proof of the uniqueness of the optimal strategy.

\begin{lemma}\label{lemma:alpha0}
    Let $A\in\mathcal{L}(\mathbb{C}^d)^{\otimes N}$ be a linear operator such that $\Tr(A \ketbra{\psi}^{\otimes N})=0$ for all vectors $\ket{\psi}\in\mathbb{C}^d$. Then, if $\ket{\phi}\in \mathrm{sym}^N\left(\mathbb{C}^d\right) $ is a vector in the symmetric subspace, we have $\bra{\phi}A\ket{\phi}=0$.
\end{lemma}

\begin{proof}
Our goal is to show that, for any two vectors $\ket{\psi_0},\ket{\psi_1}\in\mathbb{C}^d$, we have
\begin{equation}
\bra{\psi_0}^{\otimes N}A\ket{\psi_1}^{\otimes N}=0.
\end{equation}
This is enough to complete the proof, since any vector $\ket{\phi}\in \mathrm{sym}^N(\mathbb{C}^d)$ can be written as a linear combination of vectors of the form $\ket{\psi}^{\otimes N}$, namely $\ket{\phi} = \sum_i a_i \ket{\psi_i}^{\otimes N},$
and therefore
\begin{align}
\bra{\phi}A\ket{\phi}
&= \sum_{i,j} a_i^* a_j \bra{\psi_i}^{\otimes N}A\ket{\psi_j}^{\otimes N} \\
&= 0.
\end{align}

Let us consider $\ket{\psi_0}, \ket{\psi_1} \in \mathbb{C}^d$ be arbitrary vectors, and define
\begin{equation}
\ket{\chi}:= a \ket{\psi_0} + b \ket{\psi_1}, \qquad a,b\in\mathbb{C}.
\end{equation}
Now consider
\begin{align}
f(a,b)
&:= \Tr\!\left(A \ketbra{\chi}^{\otimes N}\right) \\
&= \bra{\chi}^{\otimes N}A\ket{\chi}^{\otimes N} \\
&= \Big( \overline{a} \bra{\psi_0} + \overline{b} \bra{\psi_1}\Big)^{\otimes N}  A \Big( a \ket{\psi_0} + b \ket{\psi_1}\Big)^{\otimes N}. 
\end{align}
By assumption, $f(a,b)=0$ for all $a,b\in\mathbb{C}$. The function $f(a,b)$ is a homogeneous polynomial of degree $2N$ in $a,b,\overline{a},\overline{b}$. Since it vanishes identically, all of its coefficients must be zero. In particular, the coefficient of the monomial $\bar a^N b^N$ is precisely
\begin{equation}
\bra{\psi_0}^{\otimes N}A\ket{\psi_1}^{\otimes N},
\end{equation}
because $\bar a^N$ arises only by choosing $\bra{\psi_0}$ from each factor in $\bra{\chi}^{\otimes N}$, and $b^N$ arises only by choosing $\ket{\psi_1}$ from each factor in $\ket{\chi}^{\otimes N}$.

Hence
\begin{equation}
\bra{\psi_0}^{\otimes N}A\ket{\psi_1}^{\otimes N}=0.
\end{equation}
Since $\ket{\psi_0}$ and $\ket{\psi_1}$ were arbitrary, the claim follows.

\end{proof}

\begin{lemma}\label{lemma:alpha}
    Let $B\in\mathcal{L}(\mathbb{C}^d)^{\otimes N}$ be a linear operator such that $\Tr(B \ketbra{\psi}^{\otimes N})=\alpha$ for all pure states $\ket{\psi}$. Then, if $\ket{\phi}\in \mathrm{sym}^N\left(\mathbb{C}^d\right) $ is a normalised vector in the symmetric subspace, we have $\bra{\phi}B\ket{\phi}=\alpha$.
\end{lemma}

\begin{proof}
Let us define $A:=B-\alpha \Pi,$
where $\Pi$ is the projector onto the symmetric subspace $\mathrm{sym}^N(\mathbb{C}^d)$. We first show that, for every vector $\ket{\psi}\in\mathbb{C}^d$, one has
\begin{equation}
\Tr(A \ketbra{\psi}^{\otimes N})=0.
\end{equation}
If $\ket{\psi}\neq 0$, then
\begin{align}
\frac{\Tr(A \ketbra{\psi}^{\otimes N})}{\norm{\ket{\psi}}^{2N}}
&=
\Tr\!\left(B \frac{\ketbra{\psi}^{\otimes N}}{\norm{\ket{\psi}}^{2N}}\right)
-\Tr\!\left(\alpha \Pi \frac{\ketbra{\psi}^{\otimes N}}{\norm{\ket{\psi}}^{2N}}\right) \\
&=
\alpha-\alpha \\
&=0,
\end{align}
and the conclusion is trivial if $\ket{\psi}=0$. Hence, for every $\ket{\psi}\in\mathbb{C}^d$, it holds that
\begin{equation}
\Tr(A \ketbra{\psi}^{\otimes N})=0.
\end{equation}
We can therefore invoke Lemma~\ref{lemma:alpha0}, which implies that
\begin{equation}
\bra{\phi}A\ket{\phi}=0
\end{equation}
for all $\ket{\phi}\in \mathrm{sym}^N\left(\mathbb{C}^d\right)$.

Now let $\ket{\phi}\in\mathrm{sym}^N\left(\mathbb{C}^d\right)$ be a normalized vector. Then
\begin{align}
\bra{\phi}B\ket{\phi}
&= \bra{\phi}A\ket{\phi} + \alpha \bra{\phi}\Pi\ket{\phi} \\
&= 0+\alpha \\
&= \alpha,
\end{align}
since $\Pi\ket{\phi}=\ket{\phi}$.
\end{proof}

\begin{lemma}\label{lemma:unique_wc}
When restricted to channels acting on the symmetric subspace ${\map{T}_\textup{CP}: \mathcal{L}\left(\mathrm{sym}^N\left(\mathbb{C}^d\right)\right)\to\mathcal{L}(\mathbb{C}^{d})^{\otimes K}}$, the optimal solution to the worst-case fidelity problem
\begin{align}
\max_{\mathcal{C}\in\mathrm{CPTP}}
\min_{\ket{\psi}}
F\!\Big(
\mathcal{C}\!\left(\ketbra{\psi}^{\otimes N}\right),
{\ketbra{\psi}^\mathsf{T}}^{\otimes K}
\Big)=  \frac{d_S^{N}}{d_S^{N+K}}
\end{align}
is unique. The unique solution is given by the estimation strategy
\begin{align} \label{eq:T_CP2}
    \mathcal{T}_\textup{CP}(\rho)
    :=& d_S^{N}
    \int \mathrm{d}\psi
    \Tr\!\left(\rho \ketbra{\psi}^{\otimes N}\right)\,{\ketbra{\psi}^\mathsf{T}}^{\otimes K} \\
    =&  \frac{d_S^{N}}{d_S^{N+K}} \Tr_{1\ldots N}\left( \Pi^{(N+K)}_\textup{sym} \, \rho^\mathsf{T} \otimes \id_d^{\otimes K} \right).
\end{align}
\end{lemma}

\begin{proof}
    We start by noticing that if we restrict our analysis to covariant channels, there is a unique covariant channel that maximises the worst-case fidelity. We will now prove uniqueness without the covariance assumption. 

We will split our proof in two cases. The first case is where the fidelity 
$F \of*{
\mathcal{C}(\ketbra{\psi}^{\otimes N}),
{\ketbra{\psi}^\mathsf{T}}^{\otimes K}}$
is constant for every state $\ket{\psi}\in\mathbb{C}^d$ and the case where the fidelity is not constant.

\textbf{The case where the fidelity $F\!\Big(
\mathcal{C}\!\left(\ketbra{\psi}^{\otimes N}\right),
{\ketbra{\psi}^\mathsf{T}}^{\otimes K}
\Big)$ is constant for all $\ket{\psi}$: }

Let $J$ be the Choi operator of a channel that attains the optimal value of the worst-case fidelity, and assume that its performance is constant on every pure state $\ket{\psi}$. Then
\begin{equation}
\Tr\!\left(J \ketbra{\psi}^{\otimes (N+K)}\right)=\alpha,
\qquad
\alpha=\frac{d_S^N}{d_S^{N+K}},
\end{equation}
for every pure state $\ket{\psi}\in\mathbb C^d$. By Lemma~\ref{lemma:alpha}, it follows that for every normalised vector
$\ket{\phi}\in \mathrm{sym}^{N+K}(\mathbb C^d)$,
\begin{equation}
\bra{\phi}J\ket{\phi}=\alpha.
\end{equation}
Hence, the restriction of $J$ to the symmetric subspace is
\begin{equation}
\Pi_{\mathrm{sym}}^{(N+K)} J \Pi_{\mathrm{sym}}^{(N+K)}
=
\alpha \Pi_{\mathrm{sym}}^{(N+K)}.
\end{equation}

With respect to the orthogonal decomposition $(\mathbb C^d)^{\otimes (N+K)}
=
\mathrm{sym}^{N+K}(\mathbb C^d)\oplus \mathrm{sym}^{N+K}(\mathbb C^d)^\perp,$
the operator $J$ can be written as
\begin{equation}
J=
\begin{pmatrix}
\alpha \id_{\mathrm{sym}}^{(N+K)} & X \\
X^\dagger & Y
\end{pmatrix},
\end{equation}
where $\id_{\mathrm{sym}}^{(N+K)}\in\mathcal{L}(\mathrm{sym}^{N+K}(\mathbb C^d))$ is the identity operator on the symmetric subspace.

Since $J\ge 0$, we have in particular $Y\ge 0$.
Now, the trace-preserving condition implies
\begin{equation}
\Tr_{\O}(J)=\Pi_{\mathrm{sym}}^{(N)}.
\end{equation}
Moreover, using $\Tr_{\O}\!\left(\Pi_{\mathrm{sym}}^{(N+K)}\right)
=
\frac{d_S^{N+K}}{d_S^N}\Pi_{\mathrm{sym}}^{(N)},$
together with $\alpha=\frac{d_S^N}{d_S^{N+K}}$, we obtain $\Tr_{\O}(\alpha \Pi_{\mathrm{sym}}^{(N+K)})=\Pi_{\mathrm{sym}}^{(N)}.$
Hence $\Tr_{\O}(J-\alpha \Pi_{\mathrm{sym}}^{(N+K)})=0,$ which further implies 
\begin{equation}
    \Tr(J-\alpha \Pi_{\mathrm{sym}}^{(N+K)})=0.
\end{equation}

On the other hand,
\begin{equation}
J-\alpha \Pi_{\mathrm{sym}}^{(N+K)}=
\begin{pmatrix}
0 & X \\
X^\dagger & Y
\end{pmatrix},
\end{equation}
and therefore
\begin{equation}
\Tr(J-\alpha \Pi_{\mathrm{sym}}^{(N+K)})=\Tr(Y).
\end{equation}
Thus $\Tr(Y)=0.$
Since $Y\ge 0$, it follows that $Y=0.$
Finally, since $J\ge 0$, the off-diagonal block must also vanish when the lower-right block is zero, and therefore $X=0.$
We conclude that
\begin{equation}
J=\alpha \Pi_{\mathrm{sym}}^{(N+K)}.
\end{equation}
Therefore, the optimal channel is unique.

\textbf{The case where the fidelity $F\!\Big(
\mathcal{C}\!\left(\ketbra{\psi}^{\otimes N}\right),
{\ketbra{\psi}^\mathsf{T}}^{\otimes K}
\Big)$ is not constant for all $\ket{\psi}$: } 

Let $F_{\mathcal{C}}(\psi):=
F\!\Big(
\mathcal{C}\!\left(\ketbra{\psi}^{\otimes N}\right),
{\ketbra{\psi}^\mathsf{T}}^{\otimes K}
\Big)$,
and let us define the worst-case fidelity as
$F_{\textup{wc},\mathcal{C}}
:=
\min_{\ket{\psi}} F_{\mathcal{C}}(\psi)$.
Let $\mathcal{C}^{*}$ be an optimal channel for the worst-case fidelity
problem, s.t.
\begin{align}\label{eq:optm_wc}
F_{\textup{wc}, \mathcal{C}^{*}}
=
\max_{\mathcal{C}}
\min_{\ket{\psi}}
F_{\mathcal{C}}(\psi)
=
\alpha
:=
\frac{d_S^N}{d_S^{N+K}}.
\end{align}
Now, suppose $\ket{\psi_0}$ is the state for which the minimum for
$\mathcal{C}^{*}$ is attained, s.t.
$F_{\mathcal{C}^{*}}(\psi_0)=\alpha$. By assumption of this case, the $F_{\mathcal{C}^{*}}$ is not constant as a function of states. This means that, since
$F_{\mathcal{C}^{*}}(\psi)\geq\alpha$ for every $\ket{\psi}$, there
exists a state $\ket{\phi}$ such that
$F_{\mathcal{C}^{*}}(\phi)>\alpha$.
By continuity of $F_{\mathcal{C}^{*}}$, the set of pure states for which
$F_{\mathcal{C}^{*}}(\psi)>\alpha$ has nonzero Haar measure, hence:
\begin{align}
\int \mathrm{d}\psi\,
F_{\mathcal{C}^{*}}(\psi)
>
\alpha.
\end{align}
However, Lemma~\ref{lemma:avg_trans} states that the maximum average fidelity over all quantum channels is $\alpha$,
\begin{align}
\int \mathrm{d}\psi\,
F_{\mathcal{C}^{*}}(\psi)
\leq
\max_{\mathcal{C}}
\int \mathrm{d}\psi\,
F_{\mathcal{C}}(\psi)
=
\alpha,
\end{align}
which is a contradiction.

We conclude that every channel attaining the optimal worst-case
fidelity must have constant fidelity over all pure states
$\ket{\psi}$. By the previous case, this uniquely determines the
optimal channel, which is therefore the estimation strategy
$\mathcal{T}_{\mathrm{CP}}$.
\end{proof}

\TransNK*
\begin{proof}
    The theorem follows directly from Lemma~\ref{lemma:avg_trans} and Lemma~\ref{lemma:unique_wc}.
\end{proof}

\subsection{Single-site output state structure}
\label{app:single_site}

\begin{lemma}[Single-site visibility]\label{thm:single_site}
 Let $\mathcal{C}:\mathcal{L}(\mathbb{C}^{d})^{\otimes N}
\to \mathcal{L}(\mathbb{C}^{d})^{\otimes K}$ be a CPTP map satisfying
the covariance relations given in Eq.~\eqref{eq:cov1_SM} and
Eq.~\eqref{eq:cov2_SM} (see Lemma~\ref{lemma:covariance}).
For every $i\in\{1,\dots,K\}$, the reduced state of the $i-$th output system has the form 
\begin{equation}\label{eq:sigma_psi}
\sigma_{\psi}:= \Tr_{\overline{O_i}}\!\big(
\mathcal{C}(\ketbra{\psi}^{\otimes N})
\big)
=
\eta\,\ketbra{\psi}^\mathsf{T}
+
(1-\eta)\frac{\mathbb{1}}{d},
\end{equation}
where the parameter $\eta$ is given by
\begin{equation}
\eta=\frac{dF^{(1)} -1}{d-1},
\end{equation}
and $F^{(1)}$ is the single-site fidelity
\begin{equation}
F^{(1)}  := \Tr(\sigma_{\psi}\ketbra{\psi}^\mathsf{T}).
\end{equation}

\end{lemma}

\begin{proof}
We begin by defining the reduced single-site output state. 
For every $i \in \{1,\dots,K\}$ let
\begin{equation}
\sigma_\psi^{(i)}
:=
\Tr_{\overline{O_i}}
\Big(
\mathcal{C}\big(\ketbra{\psi}^{\otimes N}\big)
\Big).
\end{equation}

Due to permutation covariance given by Eq.~\eqref{eq:cov2_SM}, all single-site marginals are identical. Therefore the reduced state does not depend on the index $i$, and we simply write $\sigma_\psi := \sigma_\psi^{(i)}$ for all $i$. Let us now define
\begin{equation}
\sigma_{U\psi}
:=
\Tr_{\overline{O_i}}
\Big(
\mathcal{C}\big((U \ketbra{\psi} U^\dagger)^{\otimes N}\big)
\Big).
\end{equation}

Using the $\mathrm{SU}(d)$ covariance relation given in Eq.~\eqref{eq:cov1_SM}, we obtain that for all $U \in \mathrm{SU}(d)$
\begin{equation}
\sigma_{U\psi}
=
\overline{U}\,\sigma_\psi\,U^\mathsf{T}.
\end{equation}

Let us now fix a reference state $\ket{0}$. For every pure state $\ket{\psi}$ there exists $U \in \mathrm{SU}(d)$ such that $\ket{\psi}=U\ket{0}$. Therefore the corresponding reduced state satisfies
\begin{equation}
\sigma_\psi
=
\overline{U}\,\sigma_0\,U^\mathsf{T}.
\end{equation}

We now determine the structure of $\sigma_0$. Consider a unitary operator $W \in \mathrm{SU}(d)$ such that
\begin{equation}
W\ket{0}=e^{i\theta}\ket{0}.
\end{equation}

Such operators form the stabilizer subgroup of $\ket{0}$. Applying the covariance relation to this transformation gives
\begin{equation}
\sigma_0
=
\overline{W}\,\sigma_0\,W^\mathsf{T}.
\end{equation}

Hence $\sigma_0$ commutes with the representation $W \mapsto \overline{W}(\cdot)W^T$ of the stabilizer subgroup. The commutant of this representation is generated by the operators $\ketbra{0}{0}$ and $\mathbb{1}$. Therefore $\sigma_0$ must have the form
\begin{equation}
\sigma_0
=
\eta\,\ketbra{0}{0}
+
\gamma\,\mathbb{1}
\end{equation}
for some coefficients $\eta,\gamma$.

Since $\sigma_0$ is a density operator, it satisfies the normalisation condition $\Tr(\sigma_0)=1,$
which implies $\gamma=\frac{1-\eta}{d}.$ Substituting this expression and transporting the state using the covariance relation yields
\begin{equation}
\sigma_\psi
=
\eta\,\ketbra{\psi}^\mathsf{T}
+
(1-\eta)\frac{\mathbb{1}}{d},
\end{equation}
which proves Eq.~\eqref{eq:sigma_psi}.

Finally we compute the fidelity between the target state $\ketbra{\psi}^\mathsf{T}$ and the reduced state $\sigma_\psi$,
\begin{align}
F^{(1)}
&=
\Tr\big(\sigma_\psi \ketbra{\psi}^\mathsf{T}\big) \\
&=
\eta\,\Tr\big(\ketbra{\psi}^\mathsf{T}\ketbra{\psi}^\mathsf{T}\big)
+
(1-\eta)\frac{1}{d} \\
&=
\eta + \frac{1-\eta}{d}.
\end{align}

Solving for $\eta$ yields
\begin{equation}
\eta=\frac{dF^{(1)} -1}{d-1}.
\end{equation}
\end{proof}

\begin{corollary}[Optimal single-site fidelity and visibility of the transposition channel]
\label{cor:optimal-single-site-visibility}
For pure-state transposition, the optimal channel
$\mathcal{T}_{\mathrm{CP}}^{(N\to K)}$ attains the optimal
$N\to1$ transposition fidelity on each single-site marginal.
Consequently,
\begin{align}
F_{\mathrm{opt}}^{(1)}
&=
F_{\mathrm{opt}}^{(N\to1)}
=
\frac{N+1}{N+d},
\\
\eta_{\mathrm{opt}}
&=
\frac{N}{N+d},
\end{align}
independently of the number of output systems $K$.
\end{corollary}

\begin{proof}
For any $N\to K$ channel $\mathcal{C}$, we can define the single-site marginal channel
\begin{align}
\mathcal{C}_i
:=
\Tr_{\overline{O_i}}
\circ
\mathcal{C},
\end{align}
which is a valid CPTP $N\to1$ channel. If $\mathcal{C}$ is covariant, then its single-site fidelity $F^{(1)}$ is independent of the input state. Therefore we can conclude that
\begin{align}
F^{(1)}
\leq
F_{\mathrm{opt}}^{(N\to1)}.
\end{align}

Using the measure-and-prepare form of the optimal transposition channel,
for every $i\in\{1,\dots,K\}$ we have
\begin{align}
\Tr_{\overline{O_i}}
\circ
\mathcal{T}_{\mathrm{CP}}^{(N\to K)}
=
\mathcal{T}_{\mathrm{CP}}^{(N\to1)}.
\end{align}

Hence, every single-site marginal attains the optimal $N\to1$
transposition fidelity. Applying Thm.~\ref{thm:TransNK} to the case
$K=1$ gives
\begin{align}
F_{\mathrm{opt}}^{(1)}
=
F_{\mathrm{opt}}^{(N\to1)}
=
\frac{d_S^N}{d_S^{N+1}}
=
\frac{N+1}{N+d}.
\end{align}
Then, by Lemma~\ref{thm:single_site},
\begin{align}
\eta_{\mathrm{opt}}
=
\frac{dF_{\mathrm{opt}}^{(1)}-1}{d-1}
=
\frac{N}{N+d}.
\end{align}
\end{proof}

\begin{remark}
   Let us also state the lower bound in the pure-state case follows directly from positivity of the single-site output state. Namely, $\sigma_\psi
=
\eta\,\ketbra{\psi}^{\mathsf T}
+
(1-\eta)\frac{\mathbb{1}}{d} \geq 0$
implies directly,
\begin{align}
\eta\geq -\frac{1}{d-1}.
\end{align}
\end{remark}

\section{Optimal $N\to K$ pure state transposition and $N\to N+K$ pure state cloning are complementary channels}
\label{appendix:C}
\ComplTC*

\begin{proof}

Let $\mathcal{H}_{\I}=\mathcal{H}_{\I'}=(\mathbb{C}^d)^{\otimes N}$ and
$\mathcal{H}_{\O}=\mathcal{H}_{\O'}=(\mathbb{C}^d)^{\otimes K}$, and let us identify $\mathrm{sym}^N(\mathbb{C}^d)$ with its symmetric subspace
of both $\mathcal{H}_{\I}$ and $\mathcal{H}_{\I'}$. Let us start with the Choi operator of the transposition map, represented on the full input space $\mathcal{H}_{\I}$,
\begin{align}
C_{\mathcal T}
&=
\frac{d^{N}_S}{d^{N+K}_S}
\Pi^{(N+K)}_{\mathrm{sym},\,\I\O}.
\end{align}

This operator satisfies
$\Tr_{\O}(C_{\mathcal T})=\Pi_{\mathrm{sym},\,\I}^{(N)}$,
so the corresponding completely positive map is trace preserving on
inputs supported on $\mathrm{sym}^N(\mathbb{C}^d)$.

One possible Stinespring dilation for the restriction to symmetric
inputs can be constructed by the following isometry,
\begin{align}
\dket{V}_{\O'\I'\I\O}
:=
\sqrt{\frac{d^N_S}{d^{N+K}_S}}
\left(
\mathbb{1}_{\O'\I'}
\otimes
\Pi^{(N+K)}_{\mathrm{sym},\,\I\O}
\right)
\dket{1}_{\O'\I'\I\O}.
\end{align}
Here, $\dket{1}_{\O'\I'\I\O}$ denotes the unnormalised maximally
entangled vector between
$\mathcal{H}_{\O'}\otimes\mathcal{H}_{\I'}$ and
$\mathcal{H}_{\O}\otimes\mathcal{H}_{\I}$ under the natural
identifications
$\mathcal{H}_{\I'}\cong\mathcal{H}_{\I}$ and
$\mathcal{H}_{\O'}\cong\mathcal{H}_{\O}$.
We regard $\I$ as the Choi reference register, so
$\dket{V}$ defines an operator
$V:\mathcal{H}_{\I}\to
\mathcal{H}_{\I'}\otimes\mathcal{H}_{\O'}\otimes\mathcal{H}_{\O}$.

To verify the dilation property, one can see that
\begin{align}
\Tr_{\I'\O'}
\left(
\dketbra{V}{V}_{\I\O\I'\O'}
\right)
&=
\frac{d_S^N}{d_S^{N+K}}
\Pi_{\mathrm{sym},\,\I\O}^{(N+K)}
=
C_{\mathcal T}.
\end{align}

Now, using this isometry, for every
$\rho\in\mathcal{L}(\mathrm{sym}^N(\mathbb{C}^d))$, we can obtain
Werner's cloning,
\begin{align}
\mathcal{W}_{\mathrm{CP}}(\rho)
&=
\Tr_{\I\O}
\Big(
\dketbra{V}{V}_{\I\O\I'\O'}
(
\rho^\mathsf{T}_{\I}
\otimes
\mathbb{1}_{\O\I'\O'}
)
\Big)
\\
&=
\frac{d^{N}_S}{d^{N+K}_S}
\Tr_{\I\O}
\Big(
(
\mathbb{1}_{\O'\I'}
\otimes
\Pi^{(N+K)}_{\mathrm{sym},\,\I\O}
)
\dket{1}\dbra{1}_{\I\O\I'\O'}
\notag\\
&\qquad\quad \cdot
(
\mathbb{1}_{\O'\I'}
\otimes
\Pi^{(N+K)}_{\mathrm{sym},\,\I\O}
)
(
\rho^\mathsf{T}_{\I}
\otimes
\mathbb{1}_{\O\I'\O'}
)
\Big)
\\
&=
\frac{d^{N}_S}{d^{N+K}_S}
\Tr_{\I\O}
\Big(
(
\Pi^{(N+K)}_{\mathrm{sym},\,\I'\O'}
\otimes
\mathbb{1}_{\O\I}
)
\dket{1}\dbra{1}_{\I\O\I'\O'}
\notag\\
&\qquad\quad \cdot
(
\Pi^{(N+K)}_{\mathrm{sym},\,\I'\O'}
\otimes
\mathbb{1}_{\O\I}
)
(
\rho^\mathsf{T}_{\I}
\otimes
\mathbb{1}_{\O\I'\O'}
)
\Big)
\\
&=
\frac{d^{N}_S}{d^{N+K}_S}
\Pi^{(N+K)}_{\mathrm{sym},\,\I'\O'}
\Tr_{\I\O}
\Big(
\dket{1}\dbra{1}_{\I\O\I'\O'}
\notag\\
&\qquad\quad \cdot
\big(
\rho^\mathsf{T}_{\I}
\otimes
\mathbb{1}_{\O\I'\O'}
\big)
\Big)
\Pi^{(N+K)}_{\mathrm{sym},\,\I'\O'}.
\end{align}

Now we use
$\dket{1}\dbra{1}_{\I\O\I'\O'}
=
\dket{1}\dbra{1}_{\I\I'}
\dket{1}\dbra{1}_{\O\O'}$
and
\begin{align}
(
\rho^\mathsf{T}_{\I}
\otimes
\mathbb{1}_{\I'}
)
\dket{1}\dbra{1}_{\I\I'}
&=
(
\mathbb{1}_{\I}
\otimes
\rho_{\I'}
)
\dket{1}\dbra{1}_{\I\I'}.
\end{align}
It follows that
\begin{align}\label{eq:cl_map_appendix}
\mathcal{W}_{\mathrm{CP}}(\rho)
&=
\frac{d^{N}_S}{d^{N+K}_S}
\Pi^{(N+K)}_{\mathrm{sym},\,\I'\O'}
(
\rho_{\I'}
\otimes
\mathbb{1}_{\O'}
)
\Pi^{(N+K)}_{\mathrm{sym},\,\I'\O'}.
\end{align}

Here, $\rho_{\I'}$ denotes the same operator supported on
$\mathrm{sym}^N(\mathbb{C}^d)$, acting on the copy
$\mathcal{H}_{\I'}$ of the full $N$-qudit space.

This is exactly Werner's optimal universal symmetric
$N\to N+K$ pure-state cloning channel~\cite{Werner1998_cloning}.
Hence, when restricted to inputs supported on
$\mathrm{sym}^N(\mathbb{C}^d)$,
$\mathcal{W}_{\mathrm{CP}}$ is the complementary channel to
$\mathcal{T}_{\mathrm{CP}}$.

\end{proof}

\begin{remark}
Given a Stinespring isometry of Thm.~\ref{thm:ComplTC}, we have that, for every isometry $H:\mathcal{H}_{\mathrm E}\to\mathcal{H}_{E'}$, the channel $\map{C}_H(\rho):=H\,\map{T}_{\textup{CP}}(\rho)\,H^\dagger$ is complementary to $\map{W}_{\textup{CP}}$. Conversely, if the dilation associated to $\map{T}_{\textup{CP}}$ is minimal as a dilation of $\map{W}_{\textup{CP}}$, then, for every channel $\map{C}$ complementary to $\map{W}_{\textup{CP}}$, there exists an isometry $H:\mathcal{H}_{\mathrm O}\to\mathcal{H}_{E}$ such that $\map{C}(\rho)=H\,\map{T}_{\textup{CP}}(\rho)\,H^\dagger$ for every $\rho$. 
\end{remark}

\paragraph{Degradability and quantum capacity of optimal cloning}

\CloningCapacity*

\begin{proof}

Let us start by emphasizing that, since the image of $\mathcal W_{\mathrm{CP}}$ is supported on the symmetric subspace $\operatorname{sym}^{N+K}(\mathbb C^d)$, the channel $\mathcal T_{\mathrm{CP}}^{(N+K\to K)}$ is well defined on the entire range of $\mathcal W_{\mathrm{CP}}$. Then, for every $\rho\in\mathcal L(\operatorname{sym}^N(\mathbb C^d))$, Eqs.~\eqref{eq:T_CP} and~\eqref{eq:cl_map} give
\begin{align}
&\map{T}_{\textup{CP}}^{(N+K\to K)}
\big(\map{W}_{\textup{CP}}(\rho)\big)\\
&=
d_S^{N+K}
\int \mathrm d\psi\,
\Tr\left[
\map{W}_{\textup{CP}}(\rho)
\ketbra{\psi}^{\otimes(N+K)}
\right]
{\ketbra{\psi}^{\mathsf T}}^{\otimes K}
\\
&=
d_S^N
\int \mathrm d\psi\,
\Tr\left[
\rho\ketbra{\psi}^{\otimes N}
\right]
{\ketbra{\psi}^{\mathsf T}}^{\otimes K}
\\
&=
\map{T}_{\textup{CP}}^{(N\to K)}(\rho),
\end{align}
hence,
\begin{align}\label{eq:cloning_degrading_relation}
\map{T}_{\textup{CP}}^{(N\to K)}
=
\map{T}_{\textup{CP}}^{(N+K\to K)}
\circ
\map{W}_{\textup{CP}}.
\end{align}
By Thm.~\ref{thm:ComplTC}, $\map{T}_{\textup{CP}}^{(N\to K)}$ is complementary to $\map{W}_{\textup{CP}}$. Therefore, Eq.~\eqref{eq:cloning_degrading_relation} proves that $\map{W}_{\textup{CP}}$ is degradable, with degrading map $\map{T}_{\textup{CP}}^{(N+K\to K)}$. Conversely, since $\map{W}_{\textup{CP}}$ is a complementary channel of $\map{T}_{\textup{CP}}^{(N\to K)}$ and $\map{T}_{\textup{CP}}^{(N+K\to K)}$ is CPTP, the same relation shows that $\map{T}_{\textup{CP}}^{(N\to K)}$ is antidegradable. Furthermore, the measure-and-prepare expression in Eq.~\eqref{eq:T_CP} shows that $\map{T}_{\textup{CP}}^{(N\to K)}$ is entanglement breaking, from which it follows that $\map{W}_{\textup{CP}}$ is a Hadamard channel.

Let us now evaluate the quantum capacity of $\map{W}_{\textup{CP}}$. 
For a quantum channel $\mathcal C_{A\to B}$ with complementary channel $\mathcal E_{A\to E}$, the coherent information of an input state $\rho_A$ is defined as
\begin{align}
I_{\mathrm c}(\rho_A,\mathcal C)
&:=
S\bigl(\mathcal C(\rho_A)\bigr)
-
S\bigl(\mathcal E(\rho_A)\bigr),
\end{align}
where $S(\sigma):=-\operatorname{Tr}[\sigma\log_2\sigma]$ is the von Neumann entropy.

Since $\map{W}_{\textup{CP}}$ is degradable, its coherent information is concave and its quantum capacity admits the single-letter formula
\begin{align}
Q(\map{W}_{\textup{CP}})
=
\max_{\rho} I_{\mathrm c}(\rho,\map{W}_{\textup{CP}})
\end{align}
~\cite{DevetakShor2005Capacity}.

Since $\map{W}_{\textup{CP}}$ and its complementary channel $\map{T}_{\textup{CP}}^{(N\to K)}$ are covariant channels, this implies that the coherent information is invariant under unitary conjugation of the input, namely,
\begin{align}
I_{\mathrm c}\!\left(
U^{\otimes N}\rho(U^\dagger)^{\otimes N},
\map{W}_{\textup{CP}}
\right)
=
I_{\mathrm c}\!\left(
\rho,\map{W}_{\textup{CP}}
\right).
\end{align}

Consequently, using Schur's lemma, $\int_{\mathrm{SU}(d)}
\mathrm dU\,
U^{\otimes N}\rho(U^\dagger)^{\otimes N}
=
\frac{\Pi^{N}_{\mathrm{sym}}}{d^{N}_{S}}$
we have,
\begin{align}
I_{\mathrm c}\left(\frac{\Pi^{N}_{\mathrm{sym}}}{d^{N}_{S}},\map{W}_{\textup{CP}}\right)
&\geq
\int_{\mathrm{SU}(d)}
\mathrm dU\,
I_{\mathrm c}\left(
U^{\otimes N}\rho(U^\dagger)^{\otimes N},
\map{W}_{\textup{CP}}
\right)
\\
&=
I_{\mathrm c}\big(\rho,\map{W}_{\textup{CP}}\big),
\end{align}
where the inequality follows from concavity and the equality follows from channel covariance. From this, we conclude that the maximum is attained when $\rho = \frac{\Pi^{N}_{\mathrm{sym}}}{d^{N}_{S}}$. Directly from Eqs.~\eqref{eq:T_CP} and~\eqref{eq:cl_map} we have,
\begin{align}
\map{W}_{\textup{CP}}\left(\frac{\Pi_{\mathrm{sym}}^{(N)}}{d^N_S}\right)
&=
\frac{\Pi_{\mathrm{sym}}^{(N+K)}}{d^{N+K}_S},
&
\map{T}_{\textup{CP}}^{(N\to K)}\left(\frac{\Pi_{\mathrm{sym}}^{(N)}}{d^{N}_S}\right)
&=
\frac{\Pi_{\mathrm{sym}}^{(K)}}{d^{K}_S},
\end{align}
and therefore,
\begin{align}
Q\big(\map{W}_{\textup{CP}}\big)
&=
I_{\mathrm c}\left(\frac{\Pi_{\mathrm{sym}}^{(N)}}{d^{N}_S},\map{W}_{\textup{CP}}\right)
\\
&=
S\left(\frac{\Pi_{\mathrm{sym}}^{(N+K)}}{d^{N+K}_S}\right)-S\left(\frac{\Pi_{\mathrm{sym}}^{(K)}}{d^{K}_S}\right)
\\
&=
\log_2 d_S^{N+K}-\log_2 d_S^K.
\end{align}
Moreover, $\map{T}_{\textup{CP}}^{(N\to K)}$ is entanglement breaking and therefore has zero quantum capacity.

Finally, the equality between the private classical and quantum capacities
follows from the degradability of $\map{W}_{\textup{CP}}$ and the
general result of Ref.~\cite{Smith2008PrivateCapacity}, thus $Q\big(\map{W}_{\textup{CP}}\big)
=
C_p\big(\map{W}_{\textup{CP}}\big)$. Moreover, since $\map{T}_{\textup{CP}}^{(N\to K)}$ is antidegradable, both its quantum and private classical capacities are zero, $Q\big(\map{T}_{\textup{CP}}^{(N\to K)}\big)
=
C_p\big(\map{T}_{\textup{CP}}^{(N\to K)}\big)
=
0.$
\end{proof}

\section{Optimal $N\to 1$ mixed state transposition}
\label{appendix:D}

\begin{lemma}\label{lem:eig_sw}
The minimal and maximal eigenvalues of the operator 
\begin{align}
    T_N:= \sum^{N}_{i = 1} F_{0i}\otimes \mathbb{1}_{\overline{i}}
\end{align}  
are respectively given by
\begin{align}
    \lambda_\textup{min} &=
    \begin{cases}
        - N &\textup{ if } N \leq d-1 \\
        - d + 1 &\textup{ if } N \geq d -1 
    \end{cases},
    \quad \quad  \\ 
    \lambda_\textup{max} & = N .
\end{align}
\end{lemma}

\begin{proof}
Using Lemma \ref{lem:JM-spectrum-Sn} from the Appendix, we know that the spectrum of the Jucys--Murphy element $J_{N+1}$ is given by all possible contents of boxes appearing in Young diagrams $\mu \vdash N+1$. Let us observe that
\begin{align}
    T_N = R(J_{N+1}),
\end{align}
where $R$ is the permutation representation defined by
\begin{equation}
    R(\pi) := \sum_{x \in [d]^{N+1}}
    \ket*{x_{\pi^{-1}(1)},\dotsc,x_{\pi^{-1}(N+1)}}
    \bra{x_1,\dotsc,x_{N+1}}.
\end{equation}

By Schur--Weyl duality, only irreducible representations corresponding to Young diagrams with at most $d$ rows appear in this representation. Therefore the spectrum of $T_N$ is determined by all possible contents of boxes appearing in Young diagrams $\mu \vdash N+1$ with $\ell(\mu)\le d$. Recall that the content of a box in row $r$ and column $c$ is $c-r$. The maximal content is achieved at the last box of the first row, which occurs for the diagram $(N+1)$, yielding
\begin{align}
    \lambda_\textup{max} = N.
\end{align}

The minimal content is achieved at the bottom box of the first column. If the diagram has $r$ rows, this content equals $1-r$. Since admissible diagrams have at most $d$ rows and at most $N+1$ rows, the maximal possible value of $r$ is $\min(d,N+1)$. Hence
\begin{align}
    \lambda_\textup{min} = 1 - \min(d,N+1),
\end{align}
which gives
\begin{align}
    \lambda_\textup{min} =
    \begin{cases}
        -N & \textup{if } N \le d-1,\\
        -d+1 & \textup{if } N \ge d-1.
    \end{cases}
\end{align}
\end{proof}

\TransMix*

\begin{proof}
To determine the admissible range of $\eta$, we characterize when the map $\rho^{\otimes N}\mapsto \eta \rho^\mathsf{T} +(1-\eta)\frac{\mathbb{1}}{d}$ is completely positive. Complete positivity is equivalent to positivity of the corresponding Choi operator.

Since the map is implemented using $N$ copies of the input, we use the notion of CP $N$-copy implementability discussed in \cite{Dong2019positive}. By their result, the Choi operator of the corresponding $N$-copy CP extension is
\begin{align}
     \Lambda^{\textup{CP}}_{N} = \frac{1}{N} \sum^{N}_{i = 1} \Lambda^{\textup{P}}_{0i}\otimes \mathbb{1}_{\overline{i}},
\end{align}
where $\Lambda^{\textup{P}} := \sum_{ij}\ketbra{i}{j} \otimes \mathcal{P}(\ketbra{i}{j})$ denotes the Choi operator of the underlying positive map, and $\mathbb{1}_{\overline{i}}$ is the identity acting on all subsystems except the $i$th one.

We now apply this construction to the noisy transposition map. Its Choi operator is
\begin{align}
    T_{0i,\eta} = \eta F_{0i} + (1-\eta)\frac{\mathbb{1}_{0i}}{d},
\end{align}
where
\begin{align}
    F_{0i} = \sum_{k,l=0}^{d-1} \ketbra{l}{k}_{0}\otimes \ketbra{k}{l}_{i}
\end{align}
is the swap operator between systems $0$ and $i$.
Hence the corresponding $N$-copy Choi operator is
\begin{align}
    T^{\textup{CP}}_{\eta, N}
    = \frac{1}{N} \sum^{N}_{i = 1} T_{0i, \eta} \otimes \mathbb{1}_{\overline{i}}
    = \frac{1}{N}\left(
    \eta \sum^{N}_{i = 1} F_{0i}\otimes \mathbb{1}_{\overline{i}}
    + (1-\eta)\frac{N}{d}\,\mathbb{1}
    \right).
\end{align}

Therefore, if $\lambda$ is an eigenvalue of
\begin{align}
    T_N:=\sum^{N}_{i = 1} F_{0i}\otimes \mathbb{1}_{\overline{i}},
\end{align}
then the corresponding eigenvalue of $T^{\textup{CP}}_{\eta,N}$ is
\begin{align}
    \frac{1}{N}\left(\eta \lambda +(1-\eta)\frac{N}{d}\right).
\end{align}
Thus positivity of $T^{\textup{CP}}_{\eta,N}$ is equivalent to
\begin{align}
    \eta \lambda +(1-\eta)\frac{N}{d}\ge 0
\end{align}
for all eigenvalues $\lambda$ of $T_N$, or equivalently,
\begin{align}
    \eta(d\lambda-N)+N\ge 0
\end{align}
for all $\lambda\in \operatorname{spec}(T_N)$. We can now use Lemma~\ref{lem:eig_sw}, which states that
\begin{align}
    \lambda_{\max}=N,
\end{align}
and
\begin{align}
    \lambda_{\min}=
    \begin{cases}
        -N & \textup{if } N\le d-1,\\
        -d+1 & \textup{if } N\ge d-1.
    \end{cases}
\end{align}

First, using $\lambda=\lambda_{\max}=N$, we obtain
\begin{align}
    \eta N(d-1)+N\ge 0,
\end{align}
which gives
\begin{align}
    \eta \ge -\frac{1}{d-1}.
\end{align}

Next, using $\lambda=\lambda_{\min}$, we obtain
\begin{align}
    \eta(d\lambda_{\min}-N)+N \ge 0.
\end{align}
Since $d\lambda_{\min}-N<0$, this yields
\begin{align}
    \eta \le \frac{N}{N-d\lambda_{\min}}.
\end{align}
Substituting the two cases for $\lambda_{\min}$ gives
\begin{align}
    -\frac{1}{d-1}\le \eta \le \eta_{\max},
\end{align}
where
\begin{align}
\eta_{\max}=
\begin{cases}
\frac{1}{d+1} & \textup{if } N\le d-1,\\[4pt]
\frac{N}{d(d-1)+N} & \textup{if } N\ge d-1.
\end{cases}
\end{align}
\end{proof}

\begin{remark}
Since pure states are a subset of density operators, Thm.~\ref{thm:TransMix} also establishes the attainability of lower bound $\eta=-1/(d-1)$. Along with the attainability of the upper bound $\eta=N/(N+d)$ established in Cor.~\ref{cor:optimal-single-site-visibility}, convexity then implies that every intermediate value of $\eta$ is also attainable.
\end{remark}

\section{Details on the circuit construction}\label{app:circuit}

Here we provide a more detailed explanation on the circuit construction presented in Fig.~\ref{fig:circuit}. For details on the representation theory, the reader can refer to~\ref{app:rep-theory}. Starting from the input state $\ket{\psi}^{\otimes N}$, we apply the Schur transform $U_{\mathrm{Sch}}$, which implements the Schur--Weyl decomposition
\begin{equation}
    (\mathbb C^d)^{\otimes n}
\cong
\bigoplus_{\lambda\vdash n,\,\ell(\lambda)\leq d}
W_\lambda\otimes V_\lambda,
\end{equation}
where $W_\lambda$ carries an irreducible representation of $U(d)$ and $V_\lambda$ the corresponding irreducible representation of $S_n$. For the symmetric input $\ket{\psi}^{\otimes N}$, only the Young-diagram sector $\lambda=(N)$ occurs, so the quantum information is entirely contained in $W_{(N)}\cong\mathrm{sym}^N(\mathbb C^d)$. Since the corresponding $S_N$ irrep $V_{(N)}$ is one-dimensional, its register can be discarded without loss of information.

The middle gate $\mathrm{CG}^\dagger$ is the inverse Clebsch--Gordan transform for $U(d)$ that embeds the input irrep $W_{(N)}$ into $W_{(N+K)}\otimes\overline{W_{(K)}},$
where $W_{(K)}$ is the totally symmetric $K$th tensor-power irreducible representation of $U(d)$, with highest weight $(K,0,\ldots,0)$, while its dual representation $\overline{W_{(K)}}$ has highest weight $(0,\ldots,0,-K)$. 

Let us now argue the uniqueness of intertwining embedding. By duality, the multiplicity of $W_{(N)}$ in $W_{(N+K)}\otimes\overline{W_{(K)}}$ is equal to the multiplicity of $W_{(N+K)}$ in
\begin{equation}
    W_{(N)}\otimes W_{(K)}
\cong
\bigoplus_{j=0}^{\min(N,K)}
W_{(N+K-j,j)}.
\end{equation}
The $j=0$ term is precisely $W_{(N+K)}$ and occurs with multiplicity one. Hence, the $U(d)$-equivariant isometric embedding of $W_{(N)}$ into $W_{(N+K)}\otimes\overline{W_{(K)}}$ is unique up to a phase. The fixed labels $(N+K)$ and $\overline{(K)}$ in the circuit specify the two output irreps, while the total-irrep label $(N)$ and its unique multiplicity state select the corresponding block of $\mathrm{CG}^\dagger$. This identifies the Clebsch--Gordan block with the nontrivial part of the common isometry constructed in the proof of Thm.~\ref{thm:ComplTC}.

In the last step, we have the two inverse Schur-transform blocks, that use fixed auxiliary one-dimensional symmetric-group states $\ket{\mathrm{init}}$ and map the two irrep registers back to physical qudits. The upper block uses the usual inverse Schur transform $U_{\mathrm{Sch}}^\dagger$ and produces the $N+K$ cloning outputs. The lower block uses the corresponding complex-conjugate inverse Schur transform $\overline{U_{\mathrm{Sch}}}^{\dagger}$ associated with the dual irrep $\overline{(K)}$, and produces the $K$ transposition/conjugation outputs. Tracing out the bottom $K$ systems therefore gives the optimal cloning channel, whereas tracing out the upper $N+K$ systems gives the complementary optimal transposition channel,
\begin{align}
\Tr_{E}\!\left(V\rho V^\dagger\right)
&=
\mathcal W_{\mathrm{CP}}(\rho),
&
\Tr_{B}\!\left(V\rho V^\dagger\right)
&=
\mathcal T_{\mathrm{CP}}(\rho).
\end{align}

\section{Representation theory background}
\label{app:rep-theory}

In this section, we review basic facts from the representation theory of the symmetric group \cite{ceccheriniSymmetric}, Schur--Weyl duality \cite{2013Harrow_church}, and Jucys--Murphy elements \cite{jucysSymmetricPolynomialsCenter1974,murphyNewConstructionYoungs1981,okounkovNewApproachRepresentation1996,ceccheriniSymmetric}.

\subsubsection{Symmetric group $S_n$}
Let $S_n$ be the symmetric group on $\{1,\dots,n\}$. A \emph{transposition} is a swap of two letters,
\begin{equation}
(i,j)\in S_n, \qquad 1\le i<j\le n,
\end{equation}
and transpositions generate $S_n$. A particularly convenient generating set is given by the \emph{adjacent transpositions}
\begin{equation}
\sigma_i \;:=\; (i,i+1), \qquad i=1,\dots,n-1,
\end{equation}
which satisfy the Coxeter relations of type $A_{n-1}$:
\begin{align}
\sigma_i^2 = 1,
&\quad
\sigma_i \sigma_{i+1} \sigma_i = \sigma_{i+1} \sigma_i \sigma_{i+1},\\
\sigma_i \sigma_j &= \sigma_j \sigma_i \ \ (|i-j|\ge 2).
\end{align}
It is often convenient to equip a group with a linear structure and consider the corresponding \emph{group algebra}, which for the symmetric group is denoted by $\mathbb{C}[S_n]$.
Equivalently, $\mathbb{C}[S_n]$ is the associative algebra generated by $\{\sigma_i\}_{i=1}^{n-1}$ subject to the above Coxeter relations. The complex irreducible representations of $S_n$ are in bijection with partitions $\lambda\vdash n$, equivalently Young diagrams with $n$ boxes.
We write $V_\lambda$ for the irrep corresponding to the shape $\lambda$.
The group algebra decomposes as
\begin{equation}
\mathbb{C}[S_n] \;\cong\; \bigoplus_{\lambda \vdash n} \mathcal{L}(V_\lambda).
\end{equation}

\subsubsection{Schur--Weyl duality}
Fix $d\in\mathbb{N}$ and consider $(\mathbb{C}^d)^{\otimes n}$. 
There are commuting actions of $U(d)$ and $S_n$, namely the diagonal action $U^{\otimes n}$ and the tensor representation $R : S_n \rightarrow \mathcal{L}((\mathbb{C}^d)^{\otimes n})$ which maps permutations $\pi \in  S_n$ to linear operators that permute the tensor
factors:
\begin{equation}
U^{\otimes n}\,R(\pi) = R(\pi)\,U^{\otimes n}
\qquad \forall\, U\in U(d),\ \pi\in S_n,
\end{equation}
where matrix $R(\pi) \in \mathcal{L}((\mathbb{C}^d)^{\otimes n})$ permutes the $n$ tensor legs:
\begin{equation}
    R(\pi) := \sum_{x \in [d]^n} \ket*{x_{\pi^{-1}(1)},x_{\pi^{-1}(2)},\dotsc,x_{\pi^{-1}(n)}}\bra{x_1,x_2,\dotsc,x_n}
\end{equation}
Schur--Weyl duality yields the multiplicity-free $U(d) \times S_n$-bimodule decomposition
\begin{equation}
\label{eq:sw-duality}
(\mathbb{C}^d)^{\otimes n}
\;\cong\;
\bigoplus_{\lambda \vdash n,\ \ell(\lambda)\le d} W_\lambda \otimes V_\lambda,
\end{equation}
where $V_\lambda$ is the $S_n$-irrep, $W_\lambda$ is the corresponding $U(d)$-irrep (also can be seen as highest-weight $\mathfrak{gl}_d$-module), and $\ell(\lambda)$ is the number of rows of $\lambda$. 
One may use the projectors $\Pi_\lambda$ onto $\lambda$-isotypic components:
\begin{equation}
(\mathbb{C}^d)^{\otimes n} = \bigoplus_{\lambda} \mathrm{Im}(\Pi_\lambda),
\qquad
\mathrm{Im}(\Pi_\lambda) = W_\lambda \otimes V_\lambda.
\end{equation}
The isomorphism (\ref{eq:sw-duality}) is achieved by a unitary called Schur transform $U_{\mathrm{Sch}}$. Schur transform can be constructed from Clebsch--Gordan transforms for the unitary group, see \cite{Bacon2006a,Harrow2005,hdkrovi25burchardt}, and, alternatively, from $S_n$-Quantum Fourier Transform, see \cite{Krovi2019,hdkrovi25burchardt}.

Clebsch--Gordan transforms $\mathrm{CG}$ are defined as unitary transformations which decompose tensor products of two $U(d)$ irreducible representations $\lambda$ and $\mu$ into direct sum of $\nu$ irreps. In other words, $\mathrm{CG}$ transforms implement the following isomorphisms:
\begin{equation}
    W_\lambda \otimes W_\mu \;\cong\; \bigoplus_{\nu}  W_\nu \otimes \C^{c^\nu_{\lambda,\mu}},
\end{equation}
where multiplicity spaces have dimensions $c^\nu_{\lambda,\mu}$, known as Littlewood--Richardson coefficients.

\subsubsection{Young--Yamanouchi (orthogonal) basis}

Let $T$ be an SYT of shape $\lambda$, and let $\ket{T}$ denote the corresponding Young--Yamanouchi (orthonormal) basis vector, labelled by standard Young tableau. 
If $(r_T(i),c_T(i))$ is the row and column of the box containing $i$ in $T$, we define the content
\begin{equation}
\mathrm{cont}_T(i) \;:=\; c_T(i) - r_T(i),
\end{equation}
and we define the axial distance between boxes containing $i$ and $i+1$ as
\begin{equation}
a_T(i) \;:=\; \mathrm{cont}_T(i+1) - \mathrm{cont}_T(i).
\end{equation}
Let $T^{(i)}$ be the tableau obtained by swapping the entries $i$ and $i+1$ in $T$ (if this swap does not yield a standard tableau, we set $\ket*{T^{(i)}} := 0$ by convention). 
Then the adjacent transposition $\sigma_i=(i,i+1)$ acts by
\begin{equation}
\sigma_i \ket{T}
=
\frac{1}{a_T(i)}\,\ket{T}
+
\sqrt{1-\frac{1}{a_T(i)^2}}\,\ket*{T^{(i)}},
\end{equation}
and, whenever $T^{(i)}$ is standard (so that $\ket*{T^{(i)}}\neq 0$),
\begin{equation}
\sigma_i \ket*{T^{(i)}}
=
\sqrt{1-\frac{1}{a_T(i)^2}}\,\ket{T}
-
\frac{1}{a_T(i)}\,\ket*{T^{(i)}}.
\end{equation}
If $T^{(i)}$ is not standard then necessarily $a_T(i)=\pm 1$, hence $\sqrt{1-\frac{1}{a_T(i)^2}} = 0$,
and the above reduces to $\sigma_i\ket{T} = \pm \ket{T}$, consistent with $\sigma_i^2=1$. 

\subsubsection{Jucys--Murphy elements}\label{app:JM}
We define the Jucys--Murphy (JM) elements $J_k\in \mathbb{C}[S_n]$ by the transposition sums
\begin{equation}
J_1 \;:=\; 0,
\qquad
J_k \;:=\; \sum_{r=1}^{k-1} (r,k) \quad (k\ge 2).
\end{equation}
One can see that they commute pairwise:
\begin{equation}
[J_k,J_\ell] = 0 \qquad \forall\, k,\ell,
\end{equation}
and moreover they generate maximal commutative subalgebra of $\mathbb{C}[S_n]$. 
Their sum is central (commutes with every element of $\mathbb{C}[S_n]$) and equals the sum of all transpositions:
\begin{equation}
\sum_{k=1}^n J_k = \sum_{1\le i<j\le n} (i,j).
\end{equation}

Now fix $\lambda\vdash n$. 
The Young--Yamanouchi basis of $V_\lambda$ is indexed by SYT $T$ of shape $\lambda$, with orthonormal vectors $\ket{T}$. 
Then the JM elements are diagonal and their spectrum is known (see, for example, \cite[Lemma B.7]{Grinko2023}):
\begin{lemma}[Jucys--Murphy spectrum for $S_n$ in the Young--Yamanouchi basis]
\label{lem:JM-spectrum-Sn}
    Let $\lambda\vdash n$ and let $\{\ket{T}\}$ be the Young--Yamanouchi (orthonormal) basis of the irrep $V_\lambda$,
    indexed by standard Young tableaux (SYT) $T$ of shape $\lambda$.
    Then each JM element $J_k$ is diagonal in the basis $\{\ket{T}\}$ and
    \begin{equation}
    J_k \ket{T} = \mathrm{cont}_T(k)\,\ket{T}
    \qquad (k=1,\dots,n).
    \end{equation}
\end{lemma}

\begin{proof}
    Set $\sigma_k:=(k,k+1)$ and note the standard recursion
    \begin{equation}
    J_{k+1} = \sigma_k J_k \sigma_k + \sigma_k,
    \qquad k=1,\dots,n-1.
    \end{equation}
    Fix an SYT $T$ and write $T^{(k)}:=\sigma_k T$ (if $\sigma_k T$ is not standard, set $\ket*{T^{(k)}}:=0$).
    Let
    \begin{equation}
    a := \mathrm{cont}_T(k+1)-\mathrm{cont}_T(k).
    \end{equation}
    In the Young--Yamanouchi orthogonal form, $\sigma_k$ acts on $\mathrm{span}\{\ket{T},\ket*{T^{(k)}}\}$ as
    \begin{align}
    \sigma_k \ket{T}
    &=
    \frac{1}{a}\ket{T}
    +
    \sqrt{1-\frac{1}{a^2}}\ket*{T^{(k)}},\\
    \sigma_k \ket*{T^{(k)}}
    &=
    \sqrt{1-\frac{1}{a^2}}\ket{T}
    -
    \frac{1}{a}\ket*{T^{(k)}}.
    \end{align}
    (When $T^{(k)}$ is not standard one has $a=\pm 1$, so the square-root term vanishes automatically.)
    
    We prove the claim by induction on $k$. The base $k=1$ is immediate since $J_1=0$ and $\mathrm{cont}_T(1)=0$.
    Assume $J_k\ket{T}=\mathrm{cont}_T(k)\ket{T}$ for all $T$. On $\mathrm{span}\{\ket{T},\ket*{T^{(k)}}\}$ we have
    \begin{equation}
    J_k = \mathrm{diag}\bigl(\mathrm{cont}_T(k),\,\mathrm{cont}_T(k+1)\bigr),
    \end{equation}
    because swapping $k$ and $k+1$ exchanges the corresponding contents. 
    Writing the representation matrix $S$ of $\sigma_k$ in the ordered basis
    $\{\ket{T},\ket*{T^{(k)}}\}$ as
    \begin{equation}
    S =
    \begin{pmatrix}
    \frac{1}{a} & \sqrt{1-\frac{1}{a^2}} \\
    \sqrt{1-\frac{1}{a^2}} & -\frac{1}{a}
    \end{pmatrix},
    \end{equation}
    the recursion reads $J_{k+1}=S\,J_k\,S+S$ on this subspace. A direct multiplication gives
    \begin{equation}
    (S\,J_k\,S+S)\ket{T} = \mathrm{cont}_T(k+1)\,\ket{T},
    \end{equation}
    Hence $J_{k+1}\ket{T}=\mathrm{cont}_T(k+1)\ket{T}$, completing the induction.
\end{proof}

\addtocontents{toc}{\protect\setcounter{tocdepth}{0}}


\begin{thebibliography}{88}%
\makeatletter
\providecommand \@ifxundefined [1]{%
 \@ifx{#1\undefined}
}%
\providecommand \@ifnum [1]{%
 \ifnum #1\expandafter \@firstoftwo
 \else \expandafter \@secondoftwo
 \fi
}%
\providecommand \@ifx [1]{%
 \ifx #1\expandafter \@firstoftwo
 \else \expandafter \@secondoftwo
 \fi
}%
\providecommand \natexlab [1]{#1}%
\providecommand \enquote  [1]{``#1''}%
\providecommand \bibnamefont  [1]{#1}%
\providecommand \bibfnamefont [1]{#1}%
\providecommand \citenamefont [1]{#1}%
\providecommand \href@noop [0]{\@secondoftwo}%
\providecommand \href [0]{\begingroup \@sanitize@url \@href}%
\providecommand \@href[1]{\@@startlink{#1}\@@href}%
\providecommand \@@href[1]{\endgroup#1\@@endlink}%
\providecommand \@sanitize@url [0]{\catcode `\\12\catcode `\$12\catcode
  `\&12\catcode `\#12\catcode `\^12\catcode `\_12\catcode `\%12\relax}%
\providecommand \@@startlink[1]{}%
\providecommand \@@endlink[0]{}%
\providecommand \url  [0]{\begingroup\@sanitize@url \@url }%
\providecommand \@url [1]{\endgroup\@href {#1}{\urlprefix }}%
\providecommand \urlprefix  [0]{URL }%
\providecommand \Eprint [0]{\href }%
\providecommand \doibase [0]{https://doi.org/}%
\providecommand \selectlanguage [0]{\@gobble}%
\providecommand \bibinfo  [0]{\@secondoftwo}%
\providecommand \bibfield  [0]{\@secondoftwo}%
\providecommand \translation [1]{[#1]}%
\providecommand \BibitemOpen [0]{}%
\providecommand \bibitemStop [0]{}%
\providecommand \bibitemNoStop [0]{.\EOS\space}%
\providecommand \EOS [0]{\spacefactor3000\relax}%
\providecommand \BibitemShut  [1]{\csname bibitem#1\endcsname}%
\let\auto@bib@innerbib\@empty
\bibitem [{\citenamefont {Wootters}\ and\ \citenamefont
  {Zurek}(1982)}]{1982Wotters_NCl}%
  \BibitemOpen
  \bibfield  {author} {\bibinfo {author} {\bibfnamefont {W.~K.}\ \bibnamefont
  {Wootters}}\ and\ \bibinfo {author} {\bibfnamefont {W.~H.}\ \bibnamefont
  {Zurek}},\ }\bibfield  {title} {\bibinfo {title} {A single quantum cannot be
  cloned},\ }\href {https://doi.org/10.1038/299802a0} {\bibfield  {journal}
  {\bibinfo  {journal} {Nature}\ }\textbf {\bibinfo {volume} {299}},\ \bibinfo
  {pages} {802} (\bibinfo {year} {1982})}\BibitemShut {NoStop}%
\bibitem [{\citenamefont {Dieks}(1982)}]{1982Dieks}%
  \BibitemOpen
  \bibfield  {author} {\bibinfo {author} {\bibfnamefont {D.}~\bibnamefont
  {Dieks}},\ }\bibfield  {title} {\bibinfo {title} {Communication by {EPR}
  devices},\ }\href {https://doi.org/10.1016/0375-9601(82)90084-6} {\bibfield
  {journal} {\bibinfo  {journal} {Physics Letters A}\ }\textbf {\bibinfo
  {volume} {92}},\ \bibinfo {pages} {271} (\bibinfo {year} {1982})}\BibitemShut
  {NoStop}%
\bibitem [{\citenamefont {Bennett}\ and\ \citenamefont
  {Brassard}(2014)}]{Bennett_2014}%
  \BibitemOpen
  \bibfield  {author} {\bibinfo {author} {\bibfnamefont {C.~H.}\ \bibnamefont
  {Bennett}}\ and\ \bibinfo {author} {\bibfnamefont {G.}~\bibnamefont
  {Brassard}},\ }\bibfield  {title} {\bibinfo {title} {Quantum cryptography:
  Public key distribution and coin tossing},\ }\href
  {https://doi.org/10.1016/j.tcs.2014.05.025} {\bibfield  {journal} {\bibinfo
  {journal} {Theoretical Computer Science}\ }\textbf {\bibinfo {volume}
  {560}},\ \bibinfo {pages} {7} (\bibinfo {year} {2014})},\ \Eprint
  {https://arxiv.org/abs/2003.06557} {arXiv:2003.06557 [quant-ph]} \BibitemShut
  {NoStop}%
\bibitem [{\citenamefont {{Horodecki}}\ \emph {et~al.}(1996)\citenamefont
  {{Horodecki}}, \citenamefont {{Horodecki}},\ and\ \citenamefont
  {{Horodecki}}}]{Horodecki1996PPT}%
  \BibitemOpen
  \bibfield  {author} {\bibinfo {author} {\bibfnamefont {M.}~\bibnamefont
  {{Horodecki}}}, \bibinfo {author} {\bibfnamefont {P.}~\bibnamefont
  {{Horodecki}}},\ and\ \bibinfo {author} {\bibfnamefont {R.}~\bibnamefont
  {{Horodecki}}},\ }\bibfield  {title} {\bibinfo {title} {{Separability of
  mixed states: necessary and sufficient conditions}},\ }\href
  {https://doi.org/10.1016/S0375-9601(96)00706-2} {\bibfield  {journal}
  {\bibinfo  {journal} {Physics Letters A}\ }\textbf {\bibinfo {volume}
  {223}},\ \bibinfo {pages} {1} (\bibinfo {year} {1996})},\ \Eprint
  {https://arxiv.org/abs/quant-ph/9605038} {arXiv:quant-ph/9605038 [quant-ph]}
  \BibitemShut {NoStop}%
\bibitem [{\citenamefont {{Horodecki}}\ and\ \citenamefont
  {{Ekert}}(2002)}]{Horodecki2000direct}%
  \BibitemOpen
  \bibfield  {author} {\bibinfo {author} {\bibfnamefont {P.}~\bibnamefont
  {{Horodecki}}}\ and\ \bibinfo {author} {\bibfnamefont {A.}~\bibnamefont
  {{Ekert}}},\ }\bibfield  {title} {\bibinfo {title} {{Method for Direct
  Detection of Quantum Entanglement}},\ }\href
  {https://doi.org/10.1103/PhysRevLett.89.127902} {\bibfield  {journal}
  {\bibinfo  {journal} {Phys. Rev. Lett.}\ }\textbf {\bibinfo {volume} {89}},\
  \bibinfo {eid} {127902} (\bibinfo {year} {2002})},\ \Eprint
  {https://arxiv.org/abs/quant-ph/0111064} {arXiv:quant-ph/0111064 [quant-ph]}
  \BibitemShut {NoStop}%
\bibitem [{\citenamefont {Gisin}\ and\ \citenamefont
  {Popescu}(1999)}]{Gisin1999Antiparallel}%
  \BibitemOpen
  \bibfield  {author} {\bibinfo {author} {\bibfnamefont {N.}~\bibnamefont
  {Gisin}}\ and\ \bibinfo {author} {\bibfnamefont {S.}~\bibnamefont
  {Popescu}},\ }\bibfield  {title} {\bibinfo {title} {Spin flips and quantum
  information for antiparallel spins},\ }\href
  {https://doi.org/10.1103/PhysRevLett.83.432} {\bibfield  {journal} {\bibinfo
  {journal} {Phys. Rev. Lett.}\ }\textbf {\bibinfo {volume} {83}},\ \bibinfo
  {pages} {432} (\bibinfo {year} {1999})},\ \Eprint
  {https://arxiv.org/abs/quant-ph/9901072} {arXiv:quant-ph/9901072}
  \BibitemShut {NoStop}%
\bibitem [{\citenamefont {Demkowicz-Dobrza{\'n}ski}(2005)}]{Demkowicz2005Estm}%
  \BibitemOpen
  \bibfield  {author} {\bibinfo {author} {\bibfnamefont {R.}~\bibnamefont
  {Demkowicz-Dobrza{\'n}ski}},\ }\bibfield  {title} {\bibinfo {title} {State
  estimation on correlated copies},\ }\href
  {https://doi.org/10.1103/PhysRevA.71.062321} {\bibfield  {journal} {\bibinfo
  {journal} {Physical Review A}\ }\textbf {\bibinfo {volume} {71}},\ \bibinfo
  {pages} {062321} (\bibinfo {year} {2005})},\ \Eprint
  {https://arxiv.org/abs/quant-ph/0412155} {arXiv:quant-ph/0412155}
  \BibitemShut {NoStop}%
\bibitem [{\citenamefont {Miyazaki}(2020)}]{Miyazaki_2020}%
  \BibitemOpen
  \bibfield  {author} {\bibinfo {author} {\bibfnamefont {J.}~\bibnamefont
  {Miyazaki}},\ }\bibfield  {title} {\bibinfo {title} {Strongly
  non-quantitative classical information in quantum carriers},\ }\href@noop {}
  {\bibfield  {journal} {\bibinfo  {journal} {arXiv e-prints}\ } (\bibinfo
  {year} {2020})},\ \Eprint {https://arxiv.org/abs/2005.06685}
  {arXiv:2005.06685 [quant-ph]} \BibitemShut {NoStop}%
\bibitem [{\citenamefont {King}\ \emph {et~al.}(2024)\citenamefont {King},
  \citenamefont {Wan},\ and\ \citenamefont {McClean}}]{King2024Learning_adv}%
  \BibitemOpen
  \bibfield  {author} {\bibinfo {author} {\bibfnamefont {R.}~\bibnamefont
  {King}}, \bibinfo {author} {\bibfnamefont {K.}~\bibnamefont {Wan}},\ and\
  \bibinfo {author} {\bibfnamefont {J.~R.}\ \bibnamefont {McClean}},\
  }\bibfield  {title} {\bibinfo {title} {Exponential learning advantages with
  conjugate states and minimal quantum memory},\ }\href
  {https://doi.org/10.1103/PRXQuantum.5.040301} {\bibfield  {journal} {\bibinfo
   {journal} {PRX Quantum}\ }\textbf {\bibinfo {volume} {5}},\ \bibinfo {eid}
  {040301} (\bibinfo {year} {2024})},\ \Eprint
  {https://arxiv.org/abs/2403.03469} {arXiv:2403.03469 [quant-ph]} \BibitemShut
  {NoStop}%
\bibitem [{\citenamefont {Patra}\ \emph {et~al.}(2026)\citenamefont {Patra},
  \citenamefont {Agarwal}, \citenamefont {Paul}, \citenamefont {Roy~Chowdhury},
  \citenamefont {Naik},\ and\ \citenamefont
  {Banik}}]{Patra2026QuantumIncompatibility}%
  \BibitemOpen
  \bibfield  {author} {\bibinfo {author} {\bibfnamefont {R.~K.}\ \bibnamefont
  {Patra}}, \bibinfo {author} {\bibfnamefont {K.}~\bibnamefont {Agarwal}},
  \bibinfo {author} {\bibfnamefont {B.}~\bibnamefont {Paul}}, \bibinfo {author}
  {\bibfnamefont {S.}~\bibnamefont {Roy~Chowdhury}}, \bibinfo {author}
  {\bibfnamefont {S.~G.}\ \bibnamefont {Naik}},\ and\ \bibinfo {author}
  {\bibfnamefont {M.}~\bibnamefont {Banik}},\ }\bibfield  {title} {\bibinfo
  {title} {Quantum incompatibility in parallel versus antiparallel spins},\
  }\href {https://doi.org/10.1103/tqrb-4m9p} {\bibfield  {journal} {\bibinfo
  {journal} {Phys. Rev. Lett.}\ }\textbf {\bibinfo {volume} {136}},\ \bibinfo
  {pages} {110402} (\bibinfo {year} {2026})},\ \Eprint
  {https://arxiv.org/abs/2504.04723} {arXiv:2504.04723 [quant-ph]} \BibitemShut
  {NoStop}%
\bibitem [{\citenamefont {{Peres}}(1996)}]{1996Peres_Sep}%
  \BibitemOpen
  \bibfield  {author} {\bibinfo {author} {\bibfnamefont {A.}~\bibnamefont
  {{Peres}}},\ }\bibfield  {title} {\bibinfo {title} {{Separability Criterion
  for Density Matrices}},\ }\href {https://doi.org/10.1103/PhysRevLett.77.1413}
  {\bibfield  {journal} {\bibinfo  {journal} {Phys.Rev.Lett.}\ }\textbf
  {\bibinfo {volume} {77}},\ \bibinfo {pages} {1413} (\bibinfo {year}
  {1996})},\ \Eprint {https://arxiv.org/abs/quant-ph/9604005}
  {arXiv:quant-ph/9604005 [quant-ph]} \BibitemShut {NoStop}%
\bibitem [{\citenamefont {Bu{\v{z}}ek}\ and\ \citenamefont
  {Hillery}(1996)}]{1996Buzek_qubit12_cl}%
  \BibitemOpen
  \bibfield  {author} {\bibinfo {author} {\bibfnamefont {V.}~\bibnamefont
  {Bu{\v{z}}ek}}\ and\ \bibinfo {author} {\bibfnamefont {M.}~\bibnamefont
  {Hillery}},\ }\bibfield  {title} {\bibinfo {title} {Quantum copying: Beyond
  the no-cloning theorem},\ }\href {https://doi.org/10.1103/PhysRevA.54.1844}
  {\bibfield  {journal} {\bibinfo  {journal} {Phys. Rev. A}\ }\textbf {\bibinfo
  {volume} {54}},\ \bibinfo {pages} {1844} (\bibinfo {year} {1996})},\ \Eprint
  {https://arxiv.org/abs/quant-ph/9607018} {arXiv:quant-ph/9607018 [quant-ph]}
  \BibitemShut {NoStop}%
\bibitem [{\citenamefont {Gisin}\ and\ \citenamefont
  {Massar}(1997)}]{1997Gisin_qubitNM_cl}%
  \BibitemOpen
  \bibfield  {author} {\bibinfo {author} {\bibfnamefont {N.}~\bibnamefont
  {Gisin}}\ and\ \bibinfo {author} {\bibfnamefont {S.}~\bibnamefont {Massar}},\
  }\bibfield  {title} {\bibinfo {title} {Optimal quantum cloning machines},\
  }\href {https://doi.org/10.1103/PhysRevLett.79.2153} {\bibfield  {journal}
  {\bibinfo  {journal} {Phys. Rev. Lett.}\ }\textbf {\bibinfo {volume} {79}},\
  \bibinfo {pages} {2153} (\bibinfo {year} {1997})},\ \Eprint
  {https://arxiv.org/abs/quant-ph/9705046} {arXiv:quant-ph/9705046 [quant-ph]}
  \BibitemShut {NoStop}%
\bibitem [{\citenamefont {Werner}(1998)}]{Werner1998_cloning}%
  \BibitemOpen
  \bibfield  {author} {\bibinfo {author} {\bibfnamefont {R.~F.}\ \bibnamefont
  {Werner}},\ }\bibfield  {title} {\bibinfo {title} {Optimal cloning of pure
  states},\ }\href {https://doi.org/10.1103/PhysRevA.58.1827} {\bibfield
  {journal} {\bibinfo  {journal} {Physical Review A}\ }\textbf {\bibinfo
  {volume} {58}},\ \bibinfo {pages} {1827} (\bibinfo {year} {1998})},\ \Eprint
  {https://arxiv.org/abs/quant-ph/9804001} {arXiv:quant-ph/9804001 [quant-ph]}
  \BibitemShut {NoStop}%
\bibitem [{\citenamefont {Keyl}\ and\ \citenamefont
  {Werner}(1999)}]{1999Keyl_cl}%
  \BibitemOpen
  \bibfield  {author} {\bibinfo {author} {\bibfnamefont {M.}~\bibnamefont
  {Keyl}}\ and\ \bibinfo {author} {\bibfnamefont {R.~F.}\ \bibnamefont
  {Werner}},\ }\bibfield  {title} {\bibinfo {title} {Optimal cloning of pure
  states, testing single clones},\ }\href {https://doi.org/10.1063/1.532887}
  {\bibfield  {journal} {\bibinfo  {journal} {J. Math. Phys.}\ }\textbf
  {\bibinfo {volume} {40}},\ \bibinfo {pages} {3283} (\bibinfo {year}
  {1999})},\ \Eprint {https://arxiv.org/abs/quant-ph/9807010}
  {arXiv:quant-ph/9807010} \BibitemShut {NoStop}%
\bibitem [{\citenamefont {Scarani}\ \emph {et~al.}(2005)\citenamefont
  {Scarani}, \citenamefont {Iblisdir}, \citenamefont {Gisin},\ and\
  \citenamefont {Ac{\'\i}n}}]{Scarani2005Clon}%
  \BibitemOpen
  \bibfield  {author} {\bibinfo {author} {\bibfnamefont {V.}~\bibnamefont
  {Scarani}}, \bibinfo {author} {\bibfnamefont {S.}~\bibnamefont {Iblisdir}},
  \bibinfo {author} {\bibfnamefont {N.}~\bibnamefont {Gisin}},\ and\ \bibinfo
  {author} {\bibfnamefont {A.}~\bibnamefont {Ac{\'\i}n}},\ }\bibfield  {title}
  {\bibinfo {title} {Quantum cloning},\ }\href
  {https://doi.org/10.1103/RevModPhys.77.1225} {\bibfield  {journal} {\bibinfo
  {journal} {Rev. Mod. Phys.}\ }\textbf {\bibinfo {volume} {77}},\ \bibinfo
  {pages} {1225} (\bibinfo {year} {2005})},\ \Eprint
  {https://arxiv.org/abs/quant-ph/0511088} {arXiv:quant-ph/0511088}
  \BibitemShut {NoStop}%
\bibitem [{\citenamefont {Dong}\ \emph {et~al.}(2019)\citenamefont {Dong},
  \citenamefont {Quintino}, \citenamefont {Soeda},\ and\ \citenamefont
  {Murao}}]{Dong2019positive}%
  \BibitemOpen
  \bibfield  {author} {\bibinfo {author} {\bibfnamefont {Q.}~\bibnamefont
  {Dong}}, \bibinfo {author} {\bibfnamefont {M.~T.}\ \bibnamefont {Quintino}},
  \bibinfo {author} {\bibfnamefont {A.}~\bibnamefont {Soeda}},\ and\ \bibinfo
  {author} {\bibfnamefont {M.}~\bibnamefont {Murao}},\ }\bibfield  {title}
  {\bibinfo {title} {Implementing positive maps with multiple copies of an
  input state},\ }\href {https://doi.org/10.1103/PhysRevA.99.052352} {\bibfield
   {journal} {\bibinfo  {journal} {Phys. Rev. A}\ }\textbf {\bibinfo {volume}
  {99}},\ \bibinfo {pages} {052352} (\bibinfo {year} {2019})},\ \Eprint
  {https://arxiv.org/abs/1808.05788} {arXiv:1808.05788} \BibitemShut {NoStop}%
\bibitem [{\citenamefont {Fan}\ \emph {et~al.}(2014)\citenamefont {Fan},
  \citenamefont {Wang}, \citenamefont {Jing}, \citenamefont {Yue},
  \citenamefont {Shi}, \citenamefont {Zhang},\ and\ \citenamefont
  {Mu}}]{Fan2014Cloning}%
  \BibitemOpen
  \bibfield  {author} {\bibinfo {author} {\bibfnamefont {H.}~\bibnamefont
  {Fan}}, \bibinfo {author} {\bibfnamefont {Y.-N.}\ \bibnamefont {Wang}},
  \bibinfo {author} {\bibfnamefont {L.}~\bibnamefont {Jing}}, \bibinfo {author}
  {\bibfnamefont {J.-D.}\ \bibnamefont {Yue}}, \bibinfo {author} {\bibfnamefont
  {H.-D.}\ \bibnamefont {Shi}}, \bibinfo {author} {\bibfnamefont {Y.-L.}\
  \bibnamefont {Zhang}},\ and\ \bibinfo {author} {\bibfnamefont {L.-Z.}\
  \bibnamefont {Mu}},\ }\bibfield  {title} {\bibinfo {title} {Quantum cloning
  machines and the applications},\ }\href
  {https://doi.org/10.1016/j.physrep.2014.06.004} {\bibfield  {journal}
  {\bibinfo  {journal} {Phys. Rep.}\ }\textbf {\bibinfo {volume} {544}},\
  \bibinfo {pages} {241} (\bibinfo {year} {2014})},\ \Eprint
  {https://arxiv.org/abs/1301.2956} {arXiv:1301.2956 [quant-ph]} \BibitemShut
  {NoStop}%
\bibitem [{\citenamefont {Horodecki}(2003)}]{Horodecki2003multicopy}%
  \BibitemOpen
  \bibfield  {author} {\bibinfo {author} {\bibfnamefont {P.}~\bibnamefont
  {Horodecki}},\ }\bibfield  {title} {\bibinfo {title} {From limits of quantum
  operations to multicopy entanglement witnesses and state-spectrum
  estimation},\ }\href {https://doi.org/10.1103/PhysRevA.68.052101} {\bibfield
  {journal} {\bibinfo  {journal} {Phys. Rev. A}\ }\textbf {\bibinfo {volume}
  {68}},\ \bibinfo {pages} {052101} (\bibinfo {year} {2003})},\ \Eprint
  {https://arxiv.org/abs/quant-ph/0111036} {arXiv:quant-ph/0111036}
  \BibitemShut {NoStop}%
\bibitem [{\citenamefont {Bae}(2017)}]{Bae2017}%
  \BibitemOpen
  \bibfield  {author} {\bibinfo {author} {\bibfnamefont {J.}~\bibnamefont
  {Bae}},\ }\bibfield  {title} {\bibinfo {title} {Designing quantum information
  processing via structural physical approximation},\ }\href
  {https://doi.org/10.1088/1361-6633/aa7d45} {\bibfield  {journal} {\bibinfo
  {journal} {Reports on Progress in Physics}\ }\textbf {\bibinfo {volume}
  {80}},\ \bibinfo {pages} {104001} (\bibinfo {year} {2017})},\ \Eprint
  {https://arxiv.org/abs/1707.02583} {arXiv:1707.02583 [quant-ph]} \BibitemShut
  {NoStop}%
\bibitem [{\citenamefont {Bu{\v{z}}ek}\ \emph {et~al.}(1999)\citenamefont
  {Bu{\v{z}}ek}, \citenamefont {Hillery},\ and\ \citenamefont
  {Werner}}]{Buzek1999UniNOT}%
  \BibitemOpen
  \bibfield  {author} {\bibinfo {author} {\bibfnamefont {V.}~\bibnamefont
  {Bu{\v{z}}ek}}, \bibinfo {author} {\bibfnamefont {M.}~\bibnamefont
  {Hillery}},\ and\ \bibinfo {author} {\bibfnamefont {R.~F.}\ \bibnamefont
  {Werner}},\ }\bibfield  {title} {\bibinfo {title} {Optimal manipulations with
  qubits: Universal-{NOT} gate},\ }\href
  {https://doi.org/10.1103/physreva.60.r2626} {\bibfield  {journal} {\bibinfo
  {journal} {Physical Review A}\ }\textbf {\bibinfo {volume} {60}},\ \bibinfo
  {pages} {R2626} (\bibinfo {year} {1999})},\ \Eprint
  {https://arxiv.org/abs/quant-ph/9901053} {arXiv:quant-ph/9901053}
  \BibitemShut {NoStop}%
\bibitem [{\citenamefont {{Rungta}}\ \emph {et~al.}(2001)\citenamefont
  {{Rungta}}, \citenamefont {{Bu{\v{z}}ek}}, \citenamefont {{Caves}},
  \citenamefont {{Hillery}},\ and\ \citenamefont {{Milburn}}}]{Rungta2001NOT}%
  \BibitemOpen
  \bibfield  {author} {\bibinfo {author} {\bibfnamefont {P.}~\bibnamefont
  {{Rungta}}}, \bibinfo {author} {\bibfnamefont {V.}~\bibnamefont
  {{Bu{\v{z}}ek}}}, \bibinfo {author} {\bibfnamefont {C.~M.}\ \bibnamefont
  {{Caves}}}, \bibinfo {author} {\bibfnamefont {M.}~\bibnamefont {{Hillery}}},\
  and\ \bibinfo {author} {\bibfnamefont {G.~J.}\ \bibnamefont {{Milburn}}},\
  }\bibfield  {title} {\bibinfo {title} {{Universal state inversion and
  concurrence in arbitrary dimensions}},\ }\href
  {https://doi.org/10.1103/PhysRevA.64.042315} {\bibfield  {journal} {\bibinfo
  {journal} {Phys. Rev. A}\ }\textbf {\bibinfo {volume} {64}},\ \bibinfo {eid}
  {042315} (\bibinfo {year} {2001})},\ \Eprint
  {https://arxiv.org/abs/quant-ph/0102040} {arXiv:quant-ph/0102040 [quant-ph]}
  \BibitemShut {NoStop}%
\bibitem [{\citenamefont {Bruss}\ \emph {et~al.}(1998)\citenamefont {Bruss},
  \citenamefont {Ekert},\ and\ \citenamefont {Macchiavello}}]{1998Bruss}%
  \BibitemOpen
  \bibfield  {author} {\bibinfo {author} {\bibfnamefont {D.}~\bibnamefont
  {Bruss}}, \bibinfo {author} {\bibfnamefont {A.}~\bibnamefont {Ekert}},\ and\
  \bibinfo {author} {\bibfnamefont {C.}~\bibnamefont {Macchiavello}},\
  }\bibfield  {title} {\bibinfo {title} {Optimal universal quantum cloning and
  state estimation},\ }\href {https://doi.org/10.1103/PhysRevLett.81.2598}
  {\bibfield  {journal} {\bibinfo  {journal} {Physical Review Letters}\
  }\textbf {\bibinfo {volume} {81}},\ \bibinfo {pages} {2598} (\bibinfo {year}
  {1998})},\ \Eprint {https://arxiv.org/abs/quant-ph/9712019}
  {arXiv:quant-ph/9712019 [quant-ph]} \BibitemShut {NoStop}%
\bibitem [{\citenamefont {{Bruss}}\ \emph {et~al.}(1998)\citenamefont
  {{Bruss}}, \citenamefont {{Divincenzo}}, \citenamefont {{Ekert}},
  \citenamefont {{Fuchs}}, \citenamefont {{Macchiavello}},\ and\ \citenamefont
  {{Smolin}}}]{Bruss1998state_dependent}%
  \BibitemOpen
  \bibfield  {author} {\bibinfo {author} {\bibfnamefont {D.}~\bibnamefont
  {{Bruss}}}, \bibinfo {author} {\bibfnamefont {D.~P.}\ \bibnamefont
  {{Divincenzo}}}, \bibinfo {author} {\bibfnamefont {A.}~\bibnamefont
  {{Ekert}}}, \bibinfo {author} {\bibfnamefont {C.~A.}\ \bibnamefont
  {{Fuchs}}}, \bibinfo {author} {\bibfnamefont {C.}~\bibnamefont
  {{Macchiavello}}},\ and\ \bibinfo {author} {\bibfnamefont {J.~A.}\
  \bibnamefont {{Smolin}}},\ }\bibfield  {title} {\bibinfo {title} {{Optimal
  universal and state-dependent quantum cloning}},\ }\href
  {https://doi.org/10.1103/PhysRevA.57.2368} {\bibfield  {journal} {\bibinfo
  {journal} {Phys. Rev. A}\ }\textbf {\bibinfo {volume} {57}},\ \bibinfo
  {pages} {2368} (\bibinfo {year} {1998})},\ \Eprint
  {https://arxiv.org/abs/quant-ph/9705038} {arXiv:quant-ph/9705038 [quant-ph]}
  \BibitemShut {NoStop}%
\bibitem [{\citenamefont {Audenaert}\ and\ \citenamefont
  {De~Moor}(2002)}]{AudenaertDeMoor2002}%
  \BibitemOpen
  \bibfield  {author} {\bibinfo {author} {\bibfnamefont {K.}~\bibnamefont
  {Audenaert}}\ and\ \bibinfo {author} {\bibfnamefont {B.}~\bibnamefont
  {De~Moor}},\ }\bibfield  {title} {\bibinfo {title} {Optimizing completely
  positive maps using semidefinite programming},\ }\href
  {https://doi.org/10.1103/PhysRevA.65.030302} {\bibfield  {journal} {\bibinfo
  {journal} {Physical Review A}\ }\textbf {\bibinfo {volume} {65}},\ \bibinfo
  {pages} {030302} (\bibinfo {year} {2002})},\ \Eprint
  {https://arxiv.org/abs/quant-ph/0109155} {arXiv:quant-ph/0109155 [quant-ph]}
  \BibitemShut {NoStop}%
\bibitem [{\citenamefont {Skrzypczyk}\ and\ \citenamefont
  {Cavalcanti}(2023)}]{Skrzypczyk2023SDP}%
  \BibitemOpen
  \bibfield  {author} {\bibinfo {author} {\bibfnamefont {P.}~\bibnamefont
  {Skrzypczyk}}\ and\ \bibinfo {author} {\bibfnamefont {D.}~\bibnamefont
  {Cavalcanti}},\ }\href {https://doi.org/10.1088/978-0-7503-3343-6} {\emph
  {\bibinfo {title} {Semidefinite Programming in Quantum Information
  Science}}},\ 2053-2563\ (\bibinfo  {publisher} {IOP Publishing},\ \bibinfo
  {year} {2023})\ \bibinfo {note} {the arXiv version contains Chapter 3 only},\
  \Eprint {https://arxiv.org/abs/2306.11637} {arXiv:2306.11637 [quant-ph]}
  \BibitemShut {NoStop}%
\bibitem [{\citenamefont {Buscemi}(2006)}]{Buscemi2006PhD}%
  \BibitemOpen
  \bibfield  {author} {\bibinfo {author} {\bibfnamefont {F.}~\bibnamefont
  {Buscemi}},\ }\emph {\bibinfo {title} {Optimization and Realization of
  Quantum Devices}},\ \href {https://arxiv.org/abs/quant-ph/0602126} {Ph.D.
  thesis},\ \bibinfo  {school} {University of Pavia} (\bibinfo {year} {2006}),\
  \Eprint {https://arxiv.org/abs/quant-ph/0602126} {arXiv:quant-ph/0602126}
  \BibitemShut {NoStop}%
\bibitem [{\citenamefont {Barnum}\ \emph {et~al.}(1996)\citenamefont {Barnum},
  \citenamefont {Caves}, \citenamefont {Fuchs}, \citenamefont {Jozsa},\ and\
  \citenamefont {Schumacher}}]{Barnum1996Broadcast}%
  \BibitemOpen
  \bibfield  {author} {\bibinfo {author} {\bibfnamefont {H.}~\bibnamefont
  {Barnum}}, \bibinfo {author} {\bibfnamefont {C.~M.}\ \bibnamefont {Caves}},
  \bibinfo {author} {\bibfnamefont {C.~A.}\ \bibnamefont {Fuchs}}, \bibinfo
  {author} {\bibfnamefont {R.}~\bibnamefont {Jozsa}},\ and\ \bibinfo {author}
  {\bibfnamefont {B.}~\bibnamefont {Schumacher}},\ }\bibfield  {title}
  {\bibinfo {title} {Noncommuting mixed states cannot be broadcast},\ }\href
  {https://doi.org/10.1103/PhysRevLett.76.2818} {\bibfield  {journal} {\bibinfo
   {journal} {Phys. Rev. Lett.}\ }\textbf {\bibinfo {volume} {76}},\ \bibinfo
  {pages} {2818} (\bibinfo {year} {1996})},\ \Eprint
  {https://arxiv.org/abs/quant-ph/9511010} {arXiv:quant-ph/9511010 [quant-ph]}
  \BibitemShut {NoStop}%
\bibitem [{\citenamefont {{D'Ariano}}\ \emph {et~al.}(2005)\citenamefont
  {{D'Ariano}}, \citenamefont {{Macchiavello}},\ and\ \citenamefont
  {{Perinotti}}}]{DAriano2005Superbroadcasting}%
  \BibitemOpen
  \bibfield  {author} {\bibinfo {author} {\bibfnamefont {G.~M.}\ \bibnamefont
  {{D'Ariano}}}, \bibinfo {author} {\bibfnamefont {C.}~\bibnamefont
  {{Macchiavello}}},\ and\ \bibinfo {author} {\bibfnamefont {P.}~\bibnamefont
  {{Perinotti}}},\ }\bibfield  {title} {\bibinfo {title} {Superbroadcasting of
  mixed states},\ }\href {https://doi.org/10.1103/PhysRevLett.95.060503}
  {\bibfield  {journal} {\bibinfo  {journal} {Phys. Rev. Lett.}\ }\textbf
  {\bibinfo {volume} {95}},\ \bibinfo {pages} {060503} (\bibinfo {year}
  {2005})},\ \Eprint {https://arxiv.org/abs/quant-ph/0506251}
  {arXiv:quant-ph/0506251 [quant-ph]} \BibitemShut {NoStop}%
\bibitem [{\citenamefont {Buscemi}\ \emph {et~al.}(2006)\citenamefont
  {Buscemi}, \citenamefont {{D'Ariano}}, \citenamefont {Macchiavello},\ and\
  \citenamefont {Perinotti}}]{Buscemi2006Superbroadcasting}%
  \BibitemOpen
  \bibfield  {author} {\bibinfo {author} {\bibfnamefont {F.}~\bibnamefont
  {Buscemi}}, \bibinfo {author} {\bibfnamefont {G.~M.}\ \bibnamefont
  {{D'Ariano}}}, \bibinfo {author} {\bibfnamefont {C.}~\bibnamefont
  {Macchiavello}},\ and\ \bibinfo {author} {\bibfnamefont {P.}~\bibnamefont
  {Perinotti}},\ }\bibfield  {title} {\bibinfo {title} {Universal and
  phase-covariant superbroadcasting for mixed qubit states},\ }\href
  {https://doi.org/10.1103/PhysRevA.74.042309} {\bibfield  {journal} {\bibinfo
  {journal} {Phys. Rev. A}\ }\textbf {\bibinfo {volume} {74}},\ \bibinfo
  {pages} {042309} (\bibinfo {year} {2006})},\ \Eprint
  {https://arxiv.org/abs/quant-ph/0602125} {arXiv:quant-ph/0602125 [quant-ph]}
  \BibitemShut {NoStop}%
\bibitem [{\citenamefont {Chiribella}\ \emph {et~al.}(2007)\citenamefont
  {Chiribella}, \citenamefont {D'Ariano}, \citenamefont {Macchiavello},
  \citenamefont {Perinotti},\ and\ \citenamefont
  {Buscemi}}]{Chiribella_2007Superbroad}%
  \BibitemOpen
  \bibfield  {author} {\bibinfo {author} {\bibfnamefont {G.}~\bibnamefont
  {Chiribella}}, \bibinfo {author} {\bibfnamefont {G.~M.}\ \bibnamefont
  {D'Ariano}}, \bibinfo {author} {\bibfnamefont {C.}~\bibnamefont
  {Macchiavello}}, \bibinfo {author} {\bibfnamefont {P.}~\bibnamefont
  {Perinotti}},\ and\ \bibinfo {author} {\bibfnamefont {F.}~\bibnamefont
  {Buscemi}},\ }\bibfield  {title} {\bibinfo {title} {Superbroadcasting and
  classical information},\ }\href {https://doi.org/10.1103/PhysRevA.75.012315}
  {\bibfield  {journal} {\bibinfo  {journal} {Physical Review A}\ }\textbf
  {\bibinfo {volume} {75}},\ \bibinfo {pages} {012315} (\bibinfo {year}
  {2007})},\ \Eprint {https://arxiv.org/abs/quant-ph/0608153}
  {arXiv:quant-ph/0608153} \BibitemShut {NoStop}%
\bibitem [{\citenamefont {Bu{\v{z}}ek}\ and\ \citenamefont
  {Hillery}(1998)}]{Buzek1998cl}%
  \BibitemOpen
  \bibfield  {author} {\bibinfo {author} {\bibfnamefont {V.}~\bibnamefont
  {Bu{\v{z}}ek}}\ and\ \bibinfo {author} {\bibfnamefont {M.}~\bibnamefont
  {Hillery}},\ }\bibfield  {title} {\bibinfo {title} {Universal optimal cloning
  of arbitrary quantum states: From qubits to quantum registers},\ }\href
  {https://doi.org/10.1103/PhysRevLett.81.5003} {\bibfield  {journal} {\bibinfo
   {journal} {Phys. Rev. Lett.}\ }\textbf {\bibinfo {volume} {81}},\ \bibinfo
  {pages} {5003} (\bibinfo {year} {1998})},\ \Eprint
  {https://arxiv.org/abs/quant-ph/9801009} {arXiv:quant-ph/9801009}
  \BibitemShut {NoStop}%
\bibitem [{\citenamefont {Buscemi}\ \emph {et~al.}(2003)\citenamefont
  {Buscemi}, \citenamefont {D'Ariano}, \citenamefont {Perinotti},\ and\
  \citenamefont {Sacchi}}]{2003Buscemi_T}%
  \BibitemOpen
  \bibfield  {author} {\bibinfo {author} {\bibfnamefont {F.}~\bibnamefont
  {Buscemi}}, \bibinfo {author} {\bibfnamefont {G.~M.}\ \bibnamefont
  {D'Ariano}}, \bibinfo {author} {\bibfnamefont {P.}~\bibnamefont
  {Perinotti}},\ and\ \bibinfo {author} {\bibfnamefont {M.~F.}\ \bibnamefont
  {Sacchi}},\ }\bibfield  {title} {\bibinfo {title} {Optimal realization of the
  transposition maps},\ }\href {https://doi.org/10.1016/S0375-9601(03)00954-X}
  {\bibfield  {journal} {\bibinfo  {journal} {Physics Letters A}\ }\textbf
  {\bibinfo {volume} {314}},\ \bibinfo {pages} {374} (\bibinfo {year}
  {2003})},\ \Eprint {https://arxiv.org/abs/quant-ph/0304175}
  {arXiv:quant-ph/0304175 [quant-ph]} \BibitemShut {NoStop}%
\bibitem [{\citenamefont {Holevo}(2007)}]{Holevo2005ComplChannels}%
  \BibitemOpen
  \bibfield  {author} {\bibinfo {author} {\bibfnamefont {A.~S.}\ \bibnamefont
  {Holevo}},\ }\bibfield  {title} {\bibinfo {title} {Complementary channels and
  the additivity problem},\ }\href {https://doi.org/10.1137/S0040585X97982244}
  {\bibfield  {journal} {\bibinfo  {journal} {Theory of Probability \& Its
  Applications}\ }\textbf {\bibinfo {volume} {51}},\ \bibinfo {pages} {92}
  (\bibinfo {year} {2007})},\ \Eprint {https://arxiv.org/abs/quant-ph/0509101}
  {arXiv:quant-ph/0509101 [quant-ph]} \BibitemShut {NoStop}%
\bibitem [{\citenamefont {Schumacher}\ and\ \citenamefont
  {Westmoreland}(1998)}]{SchumacherWestmoreland1998Privacy}%
  \BibitemOpen
  \bibfield  {author} {\bibinfo {author} {\bibfnamefont {B.}~\bibnamefont
  {Schumacher}}\ and\ \bibinfo {author} {\bibfnamefont {M.~D.}\ \bibnamefont
  {Westmoreland}},\ }\bibfield  {title} {\bibinfo {title} {Quantum privacy and
  quantum coherence},\ }\href {https://doi.org/10.1103/PhysRevLett.80.5695}
  {\bibfield  {journal} {\bibinfo  {journal} {Physical Review Letters}\
  }\textbf {\bibinfo {volume} {80}},\ \bibinfo {pages} {5695} (\bibinfo {year}
  {1998})},\ \Eprint {https://arxiv.org/abs/quant-ph/9709058}
  {arXiv:quant-ph/9709058 [quant-ph]} \BibitemShut {NoStop}%
\bibitem [{\citenamefont {Devetak}(2005)}]{Devetak2005PrivateCapacity}%
  \BibitemOpen
  \bibfield  {author} {\bibinfo {author} {\bibfnamefont {I.}~\bibnamefont
  {Devetak}},\ }\bibfield  {title} {\bibinfo {title} {The private classical
  capacity and quantum capacity of a quantum channel},\ }\href
  {https://doi.org/10.1109/TIT.2004.839515} {\bibfield  {journal} {\bibinfo
  {journal} {IEEE Transactions on Information Theory}\ }\textbf {\bibinfo
  {volume} {51}},\ \bibinfo {pages} {44} (\bibinfo {year} {2005})},\ \Eprint
  {https://arxiv.org/abs/quant-ph/0304127} {arXiv:quant-ph/0304127 [quant-ph]}
  \BibitemShut {NoStop}%
\bibitem [{\citenamefont {Devetak}\ and\ \citenamefont
  {Shor}(2005)}]{DevetakShor2005Capacity}%
  \BibitemOpen
  \bibfield  {author} {\bibinfo {author} {\bibfnamefont {I.}~\bibnamefont
  {Devetak}}\ and\ \bibinfo {author} {\bibfnamefont {P.~W.}\ \bibnamefont
  {Shor}},\ }\bibfield  {title} {\bibinfo {title} {The capacity of a quantum
  channel for simultaneous transmission of classical and quantum information},\
  }\href {https://doi.org/10.1007/s00220-005-1317-6} {\bibfield  {journal}
  {\bibinfo  {journal} {Commun. Math. Phys.}\ }\textbf {\bibinfo {volume}
  {256}},\ \bibinfo {pages} {287} (\bibinfo {year} {2005})},\ \Eprint
  {https://arxiv.org/abs/quant-ph/0311131} {arXiv:quant-ph/0311131}
  \BibitemShut {NoStop}%
\bibitem [{\citenamefont {Mančinska}\ and\ \citenamefont
  {Theil}(2025)}]{Mancinska2025classification}%
  \BibitemOpen
  \bibfield  {author} {\bibinfo {author} {\bibfnamefont {L.}~\bibnamefont
  {Mančinska}}\ and\ \bibinfo {author} {\bibfnamefont {E.}~\bibnamefont
  {Theil}},\ }\bibfield  {title} {\bibinfo {title} {Classification and
  implementation of unitary-equivariant and permutation-invariant quantum
  channels},\ }\href@noop {} {\bibfield  {journal} {\bibinfo  {journal} {arXiv
  e-prints}\ } (\bibinfo {year} {2025})},\ \Eprint
  {https://arxiv.org/abs/2510.08154} {arXiv:2510.08154 [quant-ph]} \BibitemShut
  {NoStop}%
\bibitem [{\citenamefont {Holevo}(2011)}]{holevoBook}%
  \BibitemOpen
  \bibfield  {author} {\bibinfo {author} {\bibfnamefont {A.~S.}\ \bibnamefont
  {Holevo}},\ }\href {https://doi.org/10.1007/978-88-7642-378-9} {\emph
  {\bibinfo {title} {Probabilistic and statistical aspects of quantum
  theory}}},\ Vol.~\bibinfo {volume} {1}\ (\bibinfo  {publisher} {Springer
  Science \& Business Media},\ \bibinfo {year} {2011})\BibitemShut {NoStop}%
\bibitem [{\citenamefont {{Chiribella}}\ \emph {et~al.}(2005)\citenamefont
  {{Chiribella}}, \citenamefont {{D'Ariano}},\ and\ \citenamefont
  {{Sacchi}}}]{chiribella2005estimation}%
  \BibitemOpen
  \bibfield  {author} {\bibinfo {author} {\bibfnamefont {G.}~\bibnamefont
  {{Chiribella}}}, \bibinfo {author} {\bibfnamefont {G.~M.}\ \bibnamefont
  {{D'Ariano}}},\ and\ \bibinfo {author} {\bibfnamefont {M.~F.}\ \bibnamefont
  {{Sacchi}}},\ }\bibfield  {title} {\bibinfo {title} {{Optimal estimation of
  group transformations using entanglement}},\ }\href
  {https://doi.org/10.1103/PhysRevA.72.042338} {\bibfield  {journal} {\bibinfo
  {journal} {Phys. Rev. A}\ }\textbf {\bibinfo {volume} {72}},\ \bibinfo {eid}
  {042338} (\bibinfo {year} {2005})},\ \Eprint
  {https://arxiv.org/abs/quant-ph/0506267} {arXiv:quant-ph/0506267 [quant-ph]}
  \BibitemShut {NoStop}%
\bibitem [{\citenamefont {{Grinko}}\ and\ \citenamefont
  {{Ozols}}(2024)}]{grinko2024linear}%
  \BibitemOpen
  \bibfield  {author} {\bibinfo {author} {\bibfnamefont {D.}~\bibnamefont
  {{Grinko}}}\ and\ \bibinfo {author} {\bibfnamefont {M.}~\bibnamefont
  {{Ozols}}},\ }\bibfield  {title} {\bibinfo {title} {{Linear Programming with
  Unitary-Equivariant Constraints}},\ }\href
  {https://doi.org/10.1007/s00220-024-05108-1} {\bibfield  {journal} {\bibinfo
  {journal} {Communications in Mathematical Physics}\ }\textbf {\bibinfo
  {volume} {405}},\ \bibinfo {eid} {278} (\bibinfo {year} {2024})},\ \Eprint
  {https://arxiv.org/abs/2207.05713} {arXiv:2207.05713 [quant-ph]} \BibitemShut
  {NoStop}%
\bibitem [{\citenamefont {Grinko}\ \emph {et~al.}(2023)\citenamefont {Grinko},
  \citenamefont {Burchardt},\ and\ \citenamefont {Ozols}}]{Grinko2023}%
  \BibitemOpen
  \bibfield  {author} {\bibinfo {author} {\bibfnamefont {D.}~\bibnamefont
  {Grinko}}, \bibinfo {author} {\bibfnamefont {A.}~\bibnamefont {Burchardt}},\
  and\ \bibinfo {author} {\bibfnamefont {M.}~\bibnamefont {Ozols}},\ }\bibfield
   {title} {\bibinfo {title} {Gelfand--tsetlin basis for partially transposed
  permutations, with applications to quantum information},\ }\href@noop {}
  {\bibfield  {journal} {\bibinfo  {journal} {arXiv e-prints}\ } (\bibinfo
  {year} {2023})},\ \Eprint {https://arxiv.org/abs/2310.02252}
  {arXiv:2310.02252 [quant-ph]} \BibitemShut {NoStop}%
\bibitem [{\citenamefont {{Quintino}}\ and\ \citenamefont
  {{Ebler}}(2022)}]{Quintino2022deterministic}%
  \BibitemOpen
  \bibfield  {author} {\bibinfo {author} {\bibfnamefont {M.~T.}\ \bibnamefont
  {{Quintino}}}\ and\ \bibinfo {author} {\bibfnamefont {D.}~\bibnamefont
  {{Ebler}}},\ }\bibfield  {title} {\bibinfo {title} {{Deterministic
  transformations between unitary operations: Exponential advantage with
  adaptive quantum circuits and the power of indefinite causality}},\ }\href
  {https://doi.org/10.22331/q-2022-03-31-679} {\bibfield  {journal} {\bibinfo
  {journal} {Quantum}\ }\textbf {\bibinfo {volume} {6}},\ \bibinfo {pages}
  {679} (\bibinfo {year} {2022})},\ \Eprint {https://arxiv.org/abs/2109.08202}
  {arXiv:2109.08202 [quant-ph]} \BibitemShut {NoStop}%
\bibitem [{\citenamefont {{Brzi{\'c}}}\ \emph {et~al.}(2025)\citenamefont
  {{Brzi{\'c}}}, \citenamefont {{Yoshida}}, \citenamefont {{Murao}},\ and\
  \citenamefont {{Quintino}}}]{brzic2025HOQC_KN}%
  \BibitemOpen
  \bibfield  {author} {\bibinfo {author} {\bibfnamefont {V.}~\bibnamefont
  {{Brzi{\'c}}}}, \bibinfo {author} {\bibfnamefont {S.}~\bibnamefont
  {{Yoshida}}}, \bibinfo {author} {\bibfnamefont {M.}~\bibnamefont {{Murao}}},\
  and\ \bibinfo {author} {\bibfnamefont {M.~T.}\ \bibnamefont {{Quintino}}},\
  }\bibfield  {title} {\bibinfo {title} {{Higher-order quantum computing with
  known input states}},\ }\href@noop {} {\bibfield  {journal} {\bibinfo
  {journal} {arXiv e-prints}\ } (\bibinfo {year} {2025})},\ \Eprint
  {https://arxiv.org/abs/2510.20530} {arXiv:2510.20530 [quant-ph]} \BibitemShut
  {NoStop}%
\bibitem [{\citenamefont {{Ishizaka}}\ and\ \citenamefont
  {{Hiroshima}}(2009)}]{ishizaka2009quantum}%
  \BibitemOpen
  \bibfield  {author} {\bibinfo {author} {\bibfnamefont {S.}~\bibnamefont
  {{Ishizaka}}}\ and\ \bibinfo {author} {\bibfnamefont {T.}~\bibnamefont
  {{Hiroshima}}},\ }\bibfield  {title} {\bibinfo {title} {{Quantum
  teleportation scheme by selecting one of multiple output ports}},\ }\href
  {https://doi.org/10.1103/PhysRevA.79.042306} {\bibfield  {journal} {\bibinfo
  {journal} {Phys. Rev. A}\ }\textbf {\bibinfo {volume} {79}},\ \bibinfo {eid}
  {042306} (\bibinfo {year} {2009})},\ \Eprint
  {https://arxiv.org/abs/0901.2975} {arXiv:0901.2975 [quant-ph]} \BibitemShut
  {NoStop}%
\bibitem [{\citenamefont {Ishizaka}\ and\ \citenamefont
  {Hiroshima}(2008)}]{IshizakaHiroshima}%
  \BibitemOpen
  \bibfield  {author} {\bibinfo {author} {\bibfnamefont {S.}~\bibnamefont
  {Ishizaka}}\ and\ \bibinfo {author} {\bibfnamefont {T.}~\bibnamefont
  {Hiroshima}},\ }\bibfield  {title} {\bibinfo {title} {Asymptotic
  teleportation scheme as a universal programmable quantum processor},\ }\href
  {https://doi.org/10.1103/PhysRevLett.101.240501} {\bibfield  {journal}
  {\bibinfo  {journal} {Phys. Rev. Lett.}\ }\textbf {\bibinfo {volume} {101}},\
  \bibinfo {pages} {240501} (\bibinfo {year} {2008})},\ \Eprint
  {https://arxiv.org/abs/0807.4568} {arXiv:0807.4568} \BibitemShut {NoStop}%
\bibitem [{\citenamefont {Studzi{\'n}ski}\ \emph {et~al.}(2017)\citenamefont
  {Studzi{\'n}ski}, \citenamefont {Strelchuk}, \citenamefont {Mozrzymas},\ and\
  \citenamefont {Horodecki}}]{studzinski2017port}%
  \BibitemOpen
  \bibfield  {author} {\bibinfo {author} {\bibfnamefont {M.}~\bibnamefont
  {Studzi{\'n}ski}}, \bibinfo {author} {\bibfnamefont {S.}~\bibnamefont
  {Strelchuk}}, \bibinfo {author} {\bibfnamefont {M.}~\bibnamefont
  {Mozrzymas}},\ and\ \bibinfo {author} {\bibfnamefont {M.}~\bibnamefont
  {Horodecki}},\ }\bibfield  {title} {\bibinfo {title} {Port-based
  teleportation in arbitrary dimension},\ }\href
  {https://doi.org/10.1038/s41598-017-10051-4} {\bibfield  {journal} {\bibinfo
  {journal} {Scientific reports}\ }\textbf {\bibinfo {volume} {7}},\ \bibinfo
  {pages} {10871} (\bibinfo {year} {2017})},\ \Eprint
  {https://arxiv.org/abs/1612.09260} {arXiv:1612.09260} \BibitemShut {NoStop}%
\bibitem [{\citenamefont {Mozrzymas}\ \emph {et~al.}(2018)\citenamefont
  {Mozrzymas}, \citenamefont {Studzi{\'n}ski}, \citenamefont {Strelchuk},\ and\
  \citenamefont {Horodecki}}]{mozrzymas2018optimal}%
  \BibitemOpen
  \bibfield  {author} {\bibinfo {author} {\bibfnamefont {M.}~\bibnamefont
  {Mozrzymas}}, \bibinfo {author} {\bibfnamefont {M.}~\bibnamefont
  {Studzi{\'n}ski}}, \bibinfo {author} {\bibfnamefont {S.}~\bibnamefont
  {Strelchuk}},\ and\ \bibinfo {author} {\bibfnamefont {M.}~\bibnamefont
  {Horodecki}},\ }\bibfield  {title} {\bibinfo {title} {Optimal port-based
  teleportation},\ }\href {https://doi.org/10.1088/1367-2630/aab8e7} {\bibfield
   {journal} {\bibinfo  {journal} {New Journal of Physics}\ }\textbf {\bibinfo
  {volume} {20}},\ \bibinfo {pages} {053006} (\bibinfo {year} {2018})},\
  \Eprint {https://arxiv.org/abs/1707.08456} {arXiv:1707.08456} \BibitemShut
  {NoStop}%
\bibitem [{\citenamefont {{Studzi{\'n}ski}}\ \emph {et~al.}(2022)\citenamefont
  {{Studzi{\'n}ski}}, \citenamefont {{Mozrzymas}},\ and\ \citenamefont
  {{Kopszak}}}]{studzinski2021degradation}%
  \BibitemOpen
  \bibfield  {author} {\bibinfo {author} {\bibfnamefont {M.}~\bibnamefont
  {{Studzi{\'n}ski}}}, \bibinfo {author} {\bibfnamefont {M.}~\bibnamefont
  {{Mozrzymas}}},\ and\ \bibinfo {author} {\bibfnamefont {P.}~\bibnamefont
  {{Kopszak}}},\ }\bibfield  {title} {\bibinfo {title} {{Square-root
  measurements and degradation of the resource state in port-based
  teleportation scheme}},\ }\href {https://doi.org/10.1088/1751-8121/ac8530}
  {\bibfield  {journal} {\bibinfo  {journal} {Journal of Physics A Mathematical
  General}\ }\textbf {\bibinfo {volume} {55}},\ \bibinfo {eid} {375302}
  (\bibinfo {year} {2022})},\ \Eprint {https://arxiv.org/abs/2105.14886}
  {arXiv:2105.14886 [quant-ph]} \BibitemShut {NoStop}%
\bibitem [{\citenamefont {Studzi{\'n}ski}\ \emph {et~al.}(2022)\citenamefont
  {Studzi{\'n}ski}, \citenamefont {Mozrzymas}, \citenamefont {Kopszak},\ and\
  \citenamefont {Horodecki}}]{studzinski2020efficient}%
  \BibitemOpen
  \bibfield  {author} {\bibinfo {author} {\bibfnamefont {M.}~\bibnamefont
  {Studzi{\'n}ski}}, \bibinfo {author} {\bibfnamefont {M.}~\bibnamefont
  {Mozrzymas}}, \bibinfo {author} {\bibfnamefont {P.}~\bibnamefont {Kopszak}},\
  and\ \bibinfo {author} {\bibfnamefont {M.}~\bibnamefont {Horodecki}},\
  }\bibfield  {title} {\bibinfo {title} {Efficient multi port-based
  teleportation schemes},\ }\href {https://doi.org/10.1109/TIT.2022.3187852}
  {\bibfield  {journal} {\bibinfo  {journal} {IEEE Transactions on Information
  Theory}\ }\textbf {\bibinfo {volume} {68}},\ \bibinfo {pages} {7892}
  (\bibinfo {year} {2022})},\ \Eprint {https://arxiv.org/abs/2008.00984}
  {arXiv:2008.00984} \BibitemShut {NoStop}%
\bibitem [{\citenamefont {Kopszak}\ \emph {et~al.}(2021)\citenamefont
  {Kopszak}, \citenamefont {Mozrzymas}, \citenamefont {Studzi{\'{n}}ski},\ and\
  \citenamefont {Horodecki}}]{kopszak2020multiport}%
  \BibitemOpen
  \bibfield  {author} {\bibinfo {author} {\bibfnamefont {P.}~\bibnamefont
  {Kopszak}}, \bibinfo {author} {\bibfnamefont {M.}~\bibnamefont {Mozrzymas}},
  \bibinfo {author} {\bibfnamefont {M.}~\bibnamefont {Studzi{\'{n}}ski}},\ and\
  \bibinfo {author} {\bibfnamefont {M.}~\bibnamefont {Horodecki}},\ }\bibfield
  {title} {\bibinfo {title} {Multiport based teleportation – transmission of
  a large amount of quantum information},\ }\href
  {https://doi.org/10.22331/q-2021-11-11-576} {\bibfield  {journal} {\bibinfo
  {journal} {Quantum}\ }\textbf {\bibinfo {volume} {5}},\ \bibinfo {pages}
  {576} (\bibinfo {year} {2021})},\ \Eprint {https://arxiv.org/abs/2008.00856}
  {arXiv:2008.00856} \BibitemShut {NoStop}%
\bibitem [{\citenamefont {Mozrzymas}\ \emph {et~al.}(2021)\citenamefont
  {Mozrzymas}, \citenamefont {Studzi{\'n}ski},\ and\ \citenamefont
  {Kopszak}}]{mozrzymas2021optimal}%
  \BibitemOpen
  \bibfield  {author} {\bibinfo {author} {\bibfnamefont {M.}~\bibnamefont
  {Mozrzymas}}, \bibinfo {author} {\bibfnamefont {M.}~\bibnamefont
  {Studzi{\'n}ski}},\ and\ \bibinfo {author} {\bibfnamefont {P.}~\bibnamefont
  {Kopszak}},\ }\bibfield  {title} {\bibinfo {title} {Optimal multi-port-based
  teleportation schemes},\ }\href {https://doi.org/10.22331/q-2021-06-17-477}
  {\bibfield  {journal} {\bibinfo  {journal} {Quantum}\ }\textbf {\bibinfo
  {volume} {5}},\ \bibinfo {pages} {477} (\bibinfo {year} {2021})},\ \Eprint
  {https://arxiv.org/abs/2011.09256} {arXiv:2011.09256} \BibitemShut {NoStop}%
\bibitem [{\citenamefont {Horodecki}\ \emph {et~al.}(2003)\citenamefont
  {Horodecki}, \citenamefont {Shor},\ and\ \citenamefont
  {Ruskai}}]{Horodecki_2003}%
  \BibitemOpen
  \bibfield  {author} {\bibinfo {author} {\bibfnamefont {M.}~\bibnamefont
  {Horodecki}}, \bibinfo {author} {\bibfnamefont {P.~W.}\ \bibnamefont
  {Shor}},\ and\ \bibinfo {author} {\bibfnamefont {M.~B.}\ \bibnamefont
  {Ruskai}},\ }\bibfield  {title} {\bibinfo {title} {Entanglement breaking
  channels},\ }\href {https://doi.org/10.1142/S0129055X03001709} {\bibfield
  {journal} {\bibinfo  {journal} {Reviews in Mathematical Physics}\ }\textbf
  {\bibinfo {volume} {15}},\ \bibinfo {pages} {629} (\bibinfo {year} {2003})},\
  \Eprint {https://arxiv.org/abs/quant-ph/0302031} {arXiv:quant-ph/0302031}
  \BibitemShut {NoStop}%
\bibitem [{\citenamefont {{Hayashi}}(1998)}]{Hayashi1998estimation}%
  \BibitemOpen
  \bibfield  {author} {\bibinfo {author} {\bibfnamefont {M.}~\bibnamefont
  {{Hayashi}}},\ }\bibfield  {title} {\bibinfo {title} {{Asymptotic estimation
  theory for a finite-dimensional pure state model}},\ }\href
  {https://doi.org/10.1088/0305-4470/31/20/006} {\bibfield  {journal} {\bibinfo
   {journal} {Journal of Physics A Mathematical General}\ }\textbf {\bibinfo
  {volume} {31}},\ \bibinfo {pages} {4633} (\bibinfo {year} {1998})},\ \Eprint
  {https://arxiv.org/abs/quant-ph/9704041} {arXiv:quant-ph/9704041 [quant-ph]}
  \BibitemShut {NoStop}%
\bibitem [{\citenamefont {Harrow}(2013)}]{2013Harrow_church}%
  \BibitemOpen
  \bibfield  {author} {\bibinfo {author} {\bibfnamefont {A.~W.}\ \bibnamefont
  {Harrow}},\ }\bibfield  {title} {\bibinfo {title} {The church of the
  symmetric subspace},\ }\href@noop {} {\bibfield  {journal} {\bibinfo
  {journal} {arXiv e-prints}\ } (\bibinfo {year} {2013})},\ \Eprint
  {https://arxiv.org/abs/1308.6595} {arXiv:1308.6595} \BibitemShut {NoStop}%
\bibitem [{\citenamefont {Stinespring}(1955)}]{Stinespring1955Dilation}%
  \BibitemOpen
  \bibfield  {author} {\bibinfo {author} {\bibfnamefont {W.~F.}\ \bibnamefont
  {Stinespring}},\ }\bibfield  {title} {\bibinfo {title} {Positive functions on
  {$C^*$}-algebras},\ }\href {https://doi.org/10.2307/2032342} {\bibfield
  {journal} {\bibinfo  {journal} {Proceedings of the American Mathematical
  Society}\ }\textbf {\bibinfo {volume} {6}},\ \bibinfo {pages} {211} (\bibinfo
  {year} {1955})}\BibitemShut {NoStop}%
\bibitem [{\citenamefont {Cubitt}\ \emph {et~al.}(2008)\citenamefont {Cubitt},
  \citenamefont {Ruskai},\ and\ \citenamefont {Smith}}]{Cubitt2008Deg}%
  \BibitemOpen
  \bibfield  {author} {\bibinfo {author} {\bibfnamefont {T.~S.}\ \bibnamefont
  {Cubitt}}, \bibinfo {author} {\bibfnamefont {M.~B.}\ \bibnamefont {Ruskai}},\
  and\ \bibinfo {author} {\bibfnamefont {G.}~\bibnamefont {Smith}},\ }\bibfield
   {title} {\bibinfo {title} {The structure of degradable quantum channels},\
  }\href {https://doi.org/10.1063/1.2953685} {\bibfield  {journal} {\bibinfo
  {journal} {Journal of Mathematical Physics}\ }\textbf {\bibinfo {volume}
  {49}},\ \bibinfo {eid} {102104} (\bibinfo {year} {2008})},\ \Eprint
  {https://arxiv.org/abs/0802.1360} {arXiv:0802.1360 [quant-ph]} \BibitemShut
  {NoStop}%
\bibitem [{\citenamefont {Caruso}\ and\ \citenamefont
  {Giovannetti}(2006)}]{Caruso2006Antideg}%
  \BibitemOpen
  \bibfield  {author} {\bibinfo {author} {\bibfnamefont {F.}~\bibnamefont
  {Caruso}}\ and\ \bibinfo {author} {\bibfnamefont {V.}~\bibnamefont
  {Giovannetti}},\ }\bibfield  {title} {\bibinfo {title} {Degradability of
  bosonic gaussian channels},\ }\href
  {https://doi.org/10.1103/PhysRevA.74.062307} {\bibfield  {journal} {\bibinfo
  {journal} {Physical Review A}\ }\textbf {\bibinfo {volume} {74}},\ \bibinfo
  {eid} {062307} (\bibinfo {year} {2006})},\ \Eprint
  {https://arxiv.org/abs/quant-ph/0603257} {arXiv:quant-ph/0603257 [quant-ph]}
  \BibitemShut {NoStop}%
\bibitem [{\citenamefont {Kretschmann}\ \emph
  {et~al.}(2008{\natexlab{a}})\citenamefont {Kretschmann}, \citenamefont
  {Kribs},\ and\ \citenamefont {Spekkens}}]{Kretschmann2008PrivateCorrectable}%
  \BibitemOpen
  \bibfield  {author} {\bibinfo {author} {\bibfnamefont {D.}~\bibnamefont
  {Kretschmann}}, \bibinfo {author} {\bibfnamefont {D.~W.}\ \bibnamefont
  {Kribs}},\ and\ \bibinfo {author} {\bibfnamefont {R.~W.}\ \bibnamefont
  {Spekkens}},\ }\bibfield  {title} {\bibinfo {title} {Complementarity of
  private and correctable subsystems in quantum cryptography and error
  correction},\ }\href {https://doi.org/10.1103/PhysRevA.78.032330} {\bibfield
  {journal} {\bibinfo  {journal} {Physical Review A}\ }\textbf {\bibinfo
  {volume} {78}},\ \bibinfo {eid} {032330} (\bibinfo {year}
  {2008}{\natexlab{a}})},\ \Eprint {https://arxiv.org/abs/0711.3438}
  {arXiv:0711.3438 [quant-ph]} \BibitemShut {NoStop}%
\bibitem [{\citenamefont {Kretschmann}\ \emph
  {et~al.}(2008{\natexlab{b}})\citenamefont {Kretschmann}, \citenamefont
  {Schlingemann},\ and\ \citenamefont {Werner}}]{Kretschmann2008Continuity}%
  \BibitemOpen
  \bibfield  {author} {\bibinfo {author} {\bibfnamefont {D.}~\bibnamefont
  {Kretschmann}}, \bibinfo {author} {\bibfnamefont {D.}~\bibnamefont
  {Schlingemann}},\ and\ \bibinfo {author} {\bibfnamefont {R.~F.}\ \bibnamefont
  {Werner}},\ }\bibfield  {title} {\bibinfo {title} {The
  information-disturbance tradeoff and the continuity of stinespring's
  representation},\ }\href {https://doi.org/10.1109/TIT.2008.917696} {\bibfield
   {journal} {\bibinfo  {journal} {IEEE Transactions on Information Theory}\
  }\textbf {\bibinfo {volume} {54}},\ \bibinfo {pages} {1708} (\bibinfo {year}
  {2008}{\natexlab{b}})},\ \Eprint {https://arxiv.org/abs/quant-ph/0605009}
  {arXiv:quant-ph/0605009 [quant-ph]} \BibitemShut {NoStop}%
\bibitem [{\citenamefont {Yard}\ \emph {et~al.}(2008)\citenamefont {Yard},
  \citenamefont {Hayden},\ and\ \citenamefont
  {Devetak}}]{Yard2005CapacityDegrb}%
  \BibitemOpen
  \bibfield  {author} {\bibinfo {author} {\bibfnamefont {J.}~\bibnamefont
  {Yard}}, \bibinfo {author} {\bibfnamefont {P.}~\bibnamefont {Hayden}},\ and\
  \bibinfo {author} {\bibfnamefont {I.}~\bibnamefont {Devetak}},\ }\bibfield
  {title} {\bibinfo {title} {Capacity theorems for quantum multiple-access
  channels: classical-quantum and quantum-quantum capacity regions},\ }\href
  {https://doi.org/10.1109/tit.2008.924665} {\bibfield  {journal} {\bibinfo
  {journal} {IEEE Transactions on Information Theory}\ }\textbf {\bibinfo
  {volume} {54}},\ \bibinfo {pages} {3091} (\bibinfo {year} {2008})},\ \Eprint
  {https://arxiv.org/abs/quant-ph/0501045} {arXiv:quant-ph/0501045 [quant-ph]}
  \BibitemShut {NoStop}%
\bibitem [{\citenamefont {King}\ \emph {et~al.}(2007)\citenamefont {King},
  \citenamefont {Matsumoto}, \citenamefont {Nathanson},\ and\ \citenamefont
  {Ruskai}}]{King2005Hadm_ch}%
  \BibitemOpen
  \bibfield  {author} {\bibinfo {author} {\bibfnamefont {C.}~\bibnamefont
  {King}}, \bibinfo {author} {\bibfnamefont {K.}~\bibnamefont {Matsumoto}},
  \bibinfo {author} {\bibfnamefont {M.}~\bibnamefont {Nathanson}},\ and\
  \bibinfo {author} {\bibfnamefont {M.~B.}\ \bibnamefont {Ruskai}},\ }\bibfield
   {title} {\bibinfo {title} {Properties of conjugate channels with
  applications to additivity and multiplicativity},\ }\href
  {https://math-mprf.org/journal/articles/id1123/} {\bibfield  {journal}
  {\bibinfo  {journal} {Markov Processes and Related Fields}\ }\textbf
  {\bibinfo {volume} {13}},\ \bibinfo {pages} {391} (\bibinfo {year} {2007})},\
  \Eprint {https://arxiv.org/abs/quant-ph/0509126} {arXiv:quant-ph/0509126
  [quant-ph]} \BibitemShut {NoStop}%
\bibitem [{\citenamefont {Br{\'a}dler}\ \emph
  {et~al.}(2010{\natexlab{a}})\citenamefont {Br{\'a}dler}, \citenamefont
  {Hayden}, \citenamefont {Touchette},\ and\ \citenamefont
  {Wilde}}]{Bradler2010Hadm_Entgbreak}%
  \BibitemOpen
  \bibfield  {author} {\bibinfo {author} {\bibfnamefont {K.}~\bibnamefont
  {Br{\'a}dler}}, \bibinfo {author} {\bibfnamefont {P.}~\bibnamefont {Hayden}},
  \bibinfo {author} {\bibfnamefont {D.}~\bibnamefont {Touchette}},\ and\
  \bibinfo {author} {\bibfnamefont {M.~M.}\ \bibnamefont {Wilde}},\ }\bibfield
  {title} {\bibinfo {title} {Trade-off capacities of the quantum hadamard
  channels},\ }\href {https://doi.org/10.1103/PhysRevA.81.062312} {\bibfield
  {journal} {\bibinfo  {journal} {Physical Review A}\ }\textbf {\bibinfo
  {volume} {81}},\ \bibinfo {eid} {062312} (\bibinfo {year}
  {2010}{\natexlab{a}})},\ \Eprint {https://arxiv.org/abs/1001.1732}
  {arXiv:1001.1732 [quant-ph]} \BibitemShut {NoStop}%
\bibitem [{\citenamefont {Br{\'a}dler}\ \emph
  {et~al.}(2010{\natexlab{b}})\citenamefont {Br{\'a}dler}, \citenamefont
  {Dutil}, \citenamefont {Hayden},\ and\ \citenamefont
  {Muhammad}}]{BradlerDutilHaydenMuhammad2010}%
  \BibitemOpen
  \bibfield  {author} {\bibinfo {author} {\bibfnamefont {K.}~\bibnamefont
  {Br{\'a}dler}}, \bibinfo {author} {\bibfnamefont {N.}~\bibnamefont {Dutil}},
  \bibinfo {author} {\bibfnamefont {P.}~\bibnamefont {Hayden}},\ and\ \bibinfo
  {author} {\bibfnamefont {A.}~\bibnamefont {Muhammad}},\ }\bibfield  {title}
  {\bibinfo {title} {Conjugate degradability and the quantum capacity of
  cloning channels},\ }\href {https://doi.org/10.1063/1.3449555} {\bibfield
  {journal} {\bibinfo  {journal} {Journal of Mathematical Physics}\ }\textbf
  {\bibinfo {volume} {51}},\ \bibinfo {pages} {072201} (\bibinfo {year}
  {2010}{\natexlab{b}})},\ \Eprint {https://arxiv.org/abs/0909.3297}
  {arXiv:0909.3297 [quant-ph]} \BibitemShut {NoStop}%
\bibitem [{\citenamefont {Br{\'a}dler}(2011)}]{Bradler2009Adit_additiveCh}%
  \BibitemOpen
  \bibfield  {author} {\bibinfo {author} {\bibfnamefont {K.}~\bibnamefont
  {Br{\'a}dler}},\ }\bibfield  {title} {\bibinfo {title} {An infinite sequence
  of additive channels: The classical capacity of cloning channels},\ }\href
  {https://doi.org/10.1109/tit.2011.2158896} {\bibfield  {journal} {\bibinfo
  {journal} {IEEE Transactions on Information Theory}\ }\textbf {\bibinfo
  {volume} {57}},\ \bibinfo {pages} {5497} (\bibinfo {year} {2011})},\ \Eprint
  {https://arxiv.org/abs/0903.1638} {arXiv:0903.1638 [quant-ph]} \BibitemShut
  {NoStop}%
\bibitem [{\citenamefont {{Jochym-O'Connor}}\ \emph {et~al.}(2011)\citenamefont
  {{Jochym-O'Connor}}, \citenamefont {Br{\'a}dler},\ and\ \citenamefont
  {Wilde}}]{JochymOConnor2011}%
  \BibitemOpen
  \bibfield  {author} {\bibinfo {author} {\bibfnamefont {T.}~\bibnamefont
  {{Jochym-O'Connor}}}, \bibinfo {author} {\bibfnamefont {K.}~\bibnamefont
  {Br{\'a}dler}},\ and\ \bibinfo {author} {\bibfnamefont {M.~M.}\ \bibnamefont
  {Wilde}},\ }\bibfield  {title} {\bibinfo {title} {Trade-off coding for
  universal qudit cloners motivated by the unruh effect},\ }\href
  {https://doi.org/10.1088/1751-8113/44/41/415306} {\bibfield  {journal}
  {\bibinfo  {journal} {Journal of Physics A: Mathematical and Theoretical}\
  }\textbf {\bibinfo {volume} {44}},\ \bibinfo {pages} {415306} (\bibinfo
  {year} {2011})},\ \Eprint {https://arxiv.org/abs/1103.0286} {arXiv:1103.0286
  [quant-ph]} \BibitemShut {NoStop}%
\bibitem [{\citenamefont {Wang}\ \emph {et~al.}(2017)\citenamefont {Wang},
  \citenamefont {Das},\ and\ \citenamefont {Wilde}}]{Wang2017HdmQmBroadCh}%
  \BibitemOpen
  \bibfield  {author} {\bibinfo {author} {\bibfnamefont {Q.}~\bibnamefont
  {Wang}}, \bibinfo {author} {\bibfnamefont {S.}~\bibnamefont {Das}},\ and\
  \bibinfo {author} {\bibfnamefont {M.~M.}\ \bibnamefont {Wilde}},\ }\bibfield
  {title} {\bibinfo {title} {Hadamard quantum broadcast channels},\ }\href
  {https://doi.org/10.1007/s11128-017-1697-5} {\bibfield  {journal} {\bibinfo
  {journal} {Quantum Information Processing}\ }\textbf {\bibinfo {volume}
  {16}},\ \bibinfo {eid} {248} (\bibinfo {year} {2017})},\ \Eprint
  {https://arxiv.org/abs/1611.07651} {arXiv:1611.07651 [quant-ph]} \BibitemShut
  {NoStop}%
\bibitem [{\citenamefont {{Burchardt}}\ \emph {et~al.}(2025)\citenamefont
  {{Burchardt}}, \citenamefont {{Fei}}, \citenamefont {{Grinko}}, \citenamefont
  {{Larocca}}, \citenamefont {{Ozols}}, \citenamefont {{Timmerman}},\ and\
  \citenamefont {{Visnevskyi}}}]{hdkrovi25burchardt}%
  \BibitemOpen
  \bibfield  {author} {\bibinfo {author} {\bibfnamefont {A.}~\bibnamefont
  {{Burchardt}}}, \bibinfo {author} {\bibfnamefont {J.}~\bibnamefont {{Fei}}},
  \bibinfo {author} {\bibfnamefont {D.}~\bibnamefont {{Grinko}}}, \bibinfo
  {author} {\bibfnamefont {M.}~\bibnamefont {{Larocca}}}, \bibinfo {author}
  {\bibfnamefont {M.}~\bibnamefont {{Ozols}}}, \bibinfo {author} {\bibfnamefont
  {S.}~\bibnamefont {{Timmerman}}},\ and\ \bibinfo {author} {\bibfnamefont
  {V.}~\bibnamefont {{Visnevskyi}}},\ }\bibfield  {title} {\bibinfo {title}
  {{High-dimensional quantum Schur transforms}},\ }\href@noop {} {\bibfield
  {journal} {\bibinfo  {journal} {arXiv e-prints}\ } (\bibinfo {year}
  {2025})},\ \Eprint {https://arxiv.org/abs/2509.22640} {arXiv:2509.22640
  [quant-ph]} \BibitemShut {NoStop}%
\bibitem [{\citenamefont {{Zyczkowski}}\ and\ \citenamefont
  {{Slomczynski}}(2001)}]{2001Zyczkowski_geometry}%
  \BibitemOpen
  \bibfield  {author} {\bibinfo {author} {\bibfnamefont {K.}~\bibnamefont
  {{Zyczkowski}}}\ and\ \bibinfo {author} {\bibfnamefont {W.}~\bibnamefont
  {{Slomczynski}}},\ }\bibfield  {title} {\bibinfo {title} {{The Monge metric
  on the sphere and geometry of quantum states}},\ }\href
  {https://doi.org/10.1088/0305-4470/34/34/311} {\bibfield  {journal} {\bibinfo
   {journal} {Journal of Physics A Mathematical General}\ }\textbf {\bibinfo
  {volume} {34}},\ \bibinfo {pages} {6689} (\bibinfo {year} {2001})},\ \Eprint
  {https://arxiv.org/abs/quant-ph/0008016} {arXiv:quant-ph/0008016 [quant-ph]}
  \BibitemShut {NoStop}%
\bibitem [{\citenamefont {{Zyczkowski}}\ and\ \citenamefont
  {{Sommers}}(2001)}]{2001Zyczkowski_meas}%
  \BibitemOpen
  \bibfield  {author} {\bibinfo {author} {\bibfnamefont {K.}~\bibnamefont
  {{Zyczkowski}}}\ and\ \bibinfo {author} {\bibfnamefont {H.-J.}\ \bibnamefont
  {{Sommers}}},\ }\bibfield  {title} {\bibinfo {title} {{Induced measures in
  the space of mixed quantum states}},\ }\href
  {https://doi.org/10.1088/0305-4470/34/35/335} {\bibfield  {journal} {\bibinfo
   {journal} {Journal of Physics A Mathematical General}\ }\textbf {\bibinfo
  {volume} {34}},\ \bibinfo {pages} {7111} (\bibinfo {year} {2001})},\ \Eprint
  {https://arxiv.org/abs/quant-ph/0012101} {arXiv:quant-ph/0012101 [quant-ph]}
  \BibitemShut {NoStop}%
\bibitem [{\citenamefont {Osipov}\ \emph {et~al.}(2010)\citenamefont {Osipov},
  \citenamefont {Sommers},\ and\ \citenamefont {{\.Z}yczkowski}}]{Osipov2010}%
  \BibitemOpen
  \bibfield  {author} {\bibinfo {author} {\bibfnamefont {V.~A.}\ \bibnamefont
  {Osipov}}, \bibinfo {author} {\bibfnamefont {H.-J.}\ \bibnamefont
  {Sommers}},\ and\ \bibinfo {author} {\bibfnamefont {K.}~\bibnamefont
  {{\.Z}yczkowski}},\ }\bibfield  {title} {\bibinfo {title} {Random bures mixed
  states and the distribution of their purity},\ }\href
  {https://doi.org/10.1088/1751-8113/43/5/055302} {\bibfield  {journal}
  {\bibinfo  {journal} {Journal of Physics A: Mathematical and Theoretical}\
  }\textbf {\bibinfo {volume} {43}},\ \bibinfo {pages} {055302} (\bibinfo
  {year} {2010})},\ \Eprint {https://arxiv.org/abs/0909.5094} {arXiv:0909.5094
  [cond-mat.stat-mech]} \BibitemShut {NoStop}%
\bibitem [{\citenamefont {{Murao}}\ \emph {et~al.}(1999)\citenamefont
  {{Murao}}, \citenamefont {{Jonathan}}, \citenamefont {{Plenio}},\ and\
  \citenamefont {{Vedral}}}]{Murao1999telecloning}%
  \BibitemOpen
  \bibfield  {author} {\bibinfo {author} {\bibfnamefont {M.}~\bibnamefont
  {{Murao}}}, \bibinfo {author} {\bibfnamefont {D.}~\bibnamefont {{Jonathan}}},
  \bibinfo {author} {\bibfnamefont {M.~B.}\ \bibnamefont {{Plenio}}},\ and\
  \bibinfo {author} {\bibfnamefont {V.}~\bibnamefont {{Vedral}}},\ }\bibfield
  {title} {\bibinfo {title} {{Quantum telecloning and multiparticle
  entanglement}},\ }\href {https://doi.org/10.1103/PhysRevA.59.156} {\bibfield
  {journal} {\bibinfo  {journal} {Phys. Rev. A}\ }\textbf {\bibinfo {volume}
  {59}},\ \bibinfo {pages} {156} (\bibinfo {year} {1999})},\ \Eprint
  {https://arxiv.org/abs/quant-ph/9806082} {arXiv:quant-ph/9806082 [quant-ph]}
  \BibitemShut {NoStop}%
\bibitem [{\citenamefont {Dür}\ and\ \citenamefont
  {Cirac}(2000)}]{Dur2000telecloning}%
  \BibitemOpen
  \bibfield  {author} {\bibinfo {author} {\bibfnamefont {W.}~\bibnamefont
  {Dür}}\ and\ \bibinfo {author} {\bibfnamefont {J.~I.}\ \bibnamefont
  {Cirac}},\ }\bibfield  {title} {\bibinfo {title} {Multiparty teleportation},\
  }\href {https://doi.org/10.1080/09500340008244039} {\bibfield  {journal}
  {\bibinfo  {journal} {Journal of Modern Optics}\ }\textbf {\bibinfo {volume}
  {47}},\ \bibinfo {pages} {247} (\bibinfo {year} {2000})}\BibitemShut
  {NoStop}%
\bibitem [{\citenamefont {{Chiribella}}\ \emph {et~al.}(2008)\citenamefont
  {{Chiribella}}, \citenamefont {{D'Ariano}},\ and\ \citenamefont
  {{Perinotti}}}]{Chiribella2008cloning}%
  \BibitemOpen
  \bibfield  {author} {\bibinfo {author} {\bibfnamefont {G.}~\bibnamefont
  {{Chiribella}}}, \bibinfo {author} {\bibfnamefont {G.~M.}\ \bibnamefont
  {{D'Ariano}}},\ and\ \bibinfo {author} {\bibfnamefont {P.}~\bibnamefont
  {{Perinotti}}},\ }\bibfield  {title} {\bibinfo {title} {{Optimal Cloning of
  Unitary Transformation}},\ }\href
  {https://doi.org/10.1103/PhysRevLett.101.180504} {\bibfield  {journal}
  {\bibinfo  {journal} {Phys. Rev. Lett.}\ }\textbf {\bibinfo {volume} {101}},\
  \bibinfo {eid} {180504} (\bibinfo {year} {2008})},\ \Eprint
  {https://arxiv.org/abs/0804.0129} {arXiv:0804.0129 [quant-ph]} \BibitemShut
  {NoStop}%
\bibitem [{\citenamefont {{Sekatski}}\ \emph {et~al.}(2025)\citenamefont
  {{Sekatski}}, \citenamefont {{Guryanova}}, \citenamefont {{Kothakonda}},\
  and\ \citenamefont {{Skotiniotis}}}]{Sekatski2025cloning}%
  \BibitemOpen
  \bibfield  {author} {\bibinfo {author} {\bibfnamefont {P.}~\bibnamefont
  {{Sekatski}}}, \bibinfo {author} {\bibfnamefont {Y.}~\bibnamefont
  {{Guryanova}}}, \bibinfo {author} {\bibfnamefont {N.~B.~T.}\ \bibnamefont
  {{Kothakonda}}},\ and\ \bibinfo {author} {\bibfnamefont {M.}~\bibnamefont
  {{Skotiniotis}}},\ }\bibfield  {title} {\bibinfo {title} {{Cloning Quantum
  Channels}},\ }\href@noop {} {\bibfield  {journal} {\bibinfo  {journal} {arXiv
  e-prints}\ } (\bibinfo {year} {2025})},\ \Eprint
  {https://arxiv.org/abs/2509.08059} {arXiv:2509.08059 [quant-ph]} \BibitemShut
  {NoStop}%
\bibitem [{\citenamefont {Miyazaki}\ \emph {et~al.}(2019)\citenamefont
  {Miyazaki}, \citenamefont {Soeda},\ and\ \citenamefont
  {Murao}}]{Miyazaki2017conjugation}%
  \BibitemOpen
  \bibfield  {author} {\bibinfo {author} {\bibfnamefont {J.}~\bibnamefont
  {Miyazaki}}, \bibinfo {author} {\bibfnamefont {A.}~\bibnamefont {Soeda}},\
  and\ \bibinfo {author} {\bibfnamefont {M.}~\bibnamefont {Murao}},\ }\bibfield
   {title} {\bibinfo {title} {Complex conjugation supermap of unitary quantum
  maps and its universal implementation protocol},\ }\href
  {https://doi.org/10.1103/PhysRevResearch.1.013007} {\bibfield  {journal}
  {\bibinfo  {journal} {Phys. Rev. Res.}\ }\textbf {\bibinfo {volume} {1}},\
  \bibinfo {pages} {013007} (\bibinfo {year} {2019})},\ \Eprint
  {https://arxiv.org/abs/1706.03481} {arXiv:1706.03481 [quant-ph]} \BibitemShut
  {NoStop}%
\bibitem [{\citenamefont {Ebler}\ \emph {et~al.}(2023)\citenamefont {Ebler},
  \citenamefont {Horodecki}, \citenamefont {Marciniak}, \citenamefont
  {Młynik}, \citenamefont {Quintino},\ and\ \citenamefont
  {Studziński}}]{Ebler2022conjugation}%
  \BibitemOpen
  \bibfield  {author} {\bibinfo {author} {\bibfnamefont {D.}~\bibnamefont
  {Ebler}}, \bibinfo {author} {\bibfnamefont {M.}~\bibnamefont {Horodecki}},
  \bibinfo {author} {\bibfnamefont {M.}~\bibnamefont {Marciniak}}, \bibinfo
  {author} {\bibfnamefont {T.}~\bibnamefont {Młynik}}, \bibinfo {author}
  {\bibfnamefont {M.~T.}\ \bibnamefont {Quintino}},\ and\ \bibinfo {author}
  {\bibfnamefont {M.}~\bibnamefont {Studziński}},\ }\bibfield  {title}
  {\bibinfo {title} {Optimal universal quantum circuits for unitary complex
  conjugation},\ }\href {https://doi.org/10.1109/TIT.2023.3263771} {\bibfield
  {journal} {\bibinfo  {journal} {IEEE Transactions on Information Theory}\
  }\textbf {\bibinfo {volume} {69}},\ \bibinfo {pages} {5069} (\bibinfo {year}
  {2023})},\ \Eprint {https://arxiv.org/abs/2206.00107} {arXiv:2206.00107
  [quant-ph]} \BibitemShut {NoStop}%
\bibitem [{\citenamefont {Fulton}\ and\ \citenamefont
  {Harris}(1991)}]{fulton1991representation}%
  \BibitemOpen
  \bibfield  {author} {\bibinfo {author} {\bibfnamefont {W.}~\bibnamefont
  {Fulton}}\ and\ \bibinfo {author} {\bibfnamefont {J.}~\bibnamefont
  {Harris}},\ }\href {https://doi.org/10.1007/978-1-4612-0979-9} {\emph
  {\bibinfo {title} {Representation Theory: A First Course}}},\ Graduate Texts
  in Mathematics\ (\bibinfo  {publisher} {Springer New York},\ \bibinfo {year}
  {1991})\BibitemShut {NoStop}%
\bibitem [{\citenamefont {Ram}\ and\ \citenamefont
  {Wenzl}(1992)}]{ram1992matrix}%
  \BibitemOpen
  \bibfield  {author} {\bibinfo {author} {\bibfnamefont {A.}~\bibnamefont
  {Ram}}\ and\ \bibinfo {author} {\bibfnamefont {H.}~\bibnamefont {Wenzl}},\
  }\bibfield  {title} {\bibinfo {title} {Matrix units for centralizer
  algebras},\ }\href {https://doi.org/10.1016/0021-8693(92)90109-Y} {\bibfield
  {journal} {\bibinfo  {journal} {Journal of Algebra}\ }\textbf {\bibinfo
  {volume} {145}},\ \bibinfo {pages} {378} (\bibinfo {year}
  {1992})}\BibitemShut {NoStop}%
\bibitem [{\citenamefont {Yoshida}\ \emph {et~al.}(2026)\citenamefont
  {Yoshida}, \citenamefont {Koizumi}, \citenamefont {Studzi{\'n}ski},
  \citenamefont {Quintino},\ and\ \citenamefont {Murao}}]{Yoshida2024PBT}%
  \BibitemOpen
  \bibfield  {author} {\bibinfo {author} {\bibfnamefont {S.}~\bibnamefont
  {Yoshida}}, \bibinfo {author} {\bibfnamefont {Y.}~\bibnamefont {Koizumi}},
  \bibinfo {author} {\bibfnamefont {M.}~\bibnamefont {Studzi{\'n}ski}},
  \bibinfo {author} {\bibfnamefont {M.~T.}\ \bibnamefont {Quintino}},\ and\
  \bibinfo {author} {\bibfnamefont {M.}~\bibnamefont {Murao}},\ }\bibfield
  {title} {\bibinfo {title} {One-to-one correspondence between deterministic
  port-based teleportation and unitary estimation},\ }\href
  {https://doi.org/10.1109/tit.2026.3658543} {\bibfield  {journal} {\bibinfo
  {journal} {IEEE Transactions on Information Theory}\ }\textbf {\bibinfo
  {volume} {72}},\ \bibinfo {pages} {2358} (\bibinfo {year} {2026})},\ \Eprint
  {https://arxiv.org/abs/2408.11902} {arXiv:2408.11902 [quant-ph]} \BibitemShut
  {NoStop}%
\bibitem [{\citenamefont {Smith}(2008)}]{Smith2008PrivateCapacity}%
  \BibitemOpen
  \bibfield  {author} {\bibinfo {author} {\bibfnamefont {G.}~\bibnamefont
  {Smith}},\ }\bibfield  {title} {\bibinfo {title} {Private classical capacity
  with a symmetric side channel and its application to quantum cryptography},\
  }\href {https://doi.org/10.1103/PhysRevA.78.022306} {\bibfield  {journal}
  {\bibinfo  {journal} {Physical Review A}\ }\textbf {\bibinfo {volume} {78}},\
  \bibinfo {eid} {022306} (\bibinfo {year} {2008})},\ \Eprint
  {https://arxiv.org/abs/0705.3838} {arXiv:0705.3838 [quant-ph]} \BibitemShut
  {NoStop}%
\bibitem [{\citenamefont {Ceccherini-Silberstein}\ \emph
  {et~al.}(2010)\citenamefont {Ceccherini-Silberstein}, \citenamefont
  {Scarabotti},\ and\ \citenamefont {Tolli}}]{ceccheriniSymmetric}%
  \BibitemOpen
  \bibfield  {author} {\bibinfo {author} {\bibfnamefont {T.}~\bibnamefont
  {Ceccherini-Silberstein}}, \bibinfo {author} {\bibfnamefont {F.}~\bibnamefont
  {Scarabotti}},\ and\ \bibinfo {author} {\bibfnamefont {F.}~\bibnamefont
  {Tolli}},\ }\href {https://doi.org/10.1017/CBO9781139192361} {\emph {\bibinfo
  {title} {Representation Theory of the Symmetric Groups: The Okounkov-Vershik
  Approach, Character Formulas, and Partition Algebras}}},\ Cambridge Studies
  in Advanced Mathematics\ (\bibinfo  {publisher} {Cambridge University
  Press},\ \bibinfo {year} {2010})\BibitemShut {NoStop}%
\bibitem [{\citenamefont {Jucys}(1974)}]{jucysSymmetricPolynomialsCenter1974}%
  \BibitemOpen
  \bibfield  {author} {\bibinfo {author} {\bibfnamefont {A.~A.}\ \bibnamefont
  {Jucys}},\ }\bibfield  {title} {\bibinfo {title} {Symmetric polynomials and
  the center of the symmetric group ring},\ }\href
  {https://doi.org/10.1016/0034-4877(74)90019-6} {\bibfield  {journal}
  {\bibinfo  {journal} {Reports on Mathematical Physics}\ }\textbf {\bibinfo
  {volume} {5}},\ \bibinfo {pages} {107} (\bibinfo {year} {1974})}\BibitemShut
  {NoStop}%
\bibitem [{\citenamefont {Murphy}(1981)}]{murphyNewConstructionYoungs1981}%
  \BibitemOpen
  \bibfield  {author} {\bibinfo {author} {\bibfnamefont {G.~E.}\ \bibnamefont
  {Murphy}},\ }\bibfield  {title} {\bibinfo {title} {A new construction of
  young's seminormal representation of the symmetric groups},\ }\href
  {https://doi.org/10.1016/0021-8693(81)90205-2} {\bibfield  {journal}
  {\bibinfo  {journal} {Journal of Algebra}\ }\textbf {\bibinfo {volume}
  {69}},\ \bibinfo {pages} {287} (\bibinfo {year} {1981})}\BibitemShut
  {NoStop}%
\bibitem [{\citenamefont {Okounkov}\ and\ \citenamefont
  {Vershik}(1996)}]{okounkovNewApproachRepresentation1996}%
  \BibitemOpen
  \bibfield  {author} {\bibinfo {author} {\bibfnamefont {A.}~\bibnamefont
  {Okounkov}}\ and\ \bibinfo {author} {\bibfnamefont {A.}~\bibnamefont
  {Vershik}},\ }\bibfield  {title} {\bibinfo {title} {A new approach to
  representation theory of symmetric groups},\ }\href
  {https://doi.org/10.1007/BF02433451} {\bibfield  {journal} {\bibinfo
  {journal} {Selecta Mathematica}\ }\textbf {\bibinfo {volume} {2}},\ \bibinfo
  {pages} {581} (\bibinfo {year} {1996})}\BibitemShut {NoStop}%
\bibitem [{\citenamefont {Bacon}\ \emph {et~al.}(2006)\citenamefont {Bacon},
  \citenamefont {Chuang},\ and\ \citenamefont {Harrow}}]{Bacon2006a}%
  \BibitemOpen
  \bibfield  {author} {\bibinfo {author} {\bibfnamefont {D.}~\bibnamefont
  {Bacon}}, \bibinfo {author} {\bibfnamefont {I.~L.}\ \bibnamefont {Chuang}},\
  and\ \bibinfo {author} {\bibfnamefont {A.~W.}\ \bibnamefont {Harrow}},\
  }\bibfield  {title} {\bibinfo {title} {Efficient quantum circuits for schur
  and clebsch-gordan transforms},\ }\href
  {https://doi.org/10.1103/physrevlett.97.170502} {\bibfield  {journal}
  {\bibinfo  {journal} {Physical Review Letters}\ }\textbf {\bibinfo {volume}
  {97}},\ \bibinfo {pages} {170502} (\bibinfo {year} {2006})},\ \Eprint
  {https://arxiv.org/abs/quant-ph/0407082} {arXiv:quant-ph/0407082}
  \BibitemShut {NoStop}%
\bibitem [{\citenamefont {Harrow}(2005)}]{Harrow2005}%
  \BibitemOpen
  \bibfield  {author} {\bibinfo {author} {\bibfnamefont {A.~W.}\ \bibnamefont
  {Harrow}},\ }\emph {\bibinfo {title} {Applications of coherent classical
  communication and the Schur transform to quantum information theory}},\
  \href@noop {} {Ph.D. thesis},\ \bibinfo  {school} {Massachusetts Institute of
  Technology} (\bibinfo {year} {2005}),\ \Eprint
  {https://arxiv.org/abs/quant-ph/0512255} {arXiv:quant-ph/0512255 [quant-ph]}
  \BibitemShut {NoStop}%
\bibitem [{\citenamefont {{Krovi}}(2019)}]{Krovi2019}%
  \BibitemOpen
  \bibfield  {author} {\bibinfo {author} {\bibfnamefont {H.}~\bibnamefont
  {{Krovi}}},\ }\bibfield  {title} {\bibinfo {title} {{An efficient high
  dimensional quantum Schur transform}},\ }\href
  {https://doi.org/10.22331/q-2019-02-14-122} {\bibfield  {journal} {\bibinfo
  {journal} {Quantum}\ }\textbf {\bibinfo {volume} {3}},\ \bibinfo {pages}
  {122} (\bibinfo {year} {2019})},\ \Eprint {https://arxiv.org/abs/1804.00055}
  {arXiv:1804.00055 [quant-ph]} \BibitemShut {NoStop}%
\end{thebibliography}
\end{document}